\documentclass[10pt,a4paper]{article}
\usepackage[margin=2.25cm]{geometry}
\usepackage[T1]{fontenc}
\usepackage{lmodern}
\usepackage{amsmath,amssymb,amsthm,bm,mathtools}
\usepackage{graphicx,booktabs,array,tabularx,longtable}
\usepackage{xcolor}
\usepackage{tikz}
\usetikzlibrary{arrows.meta,positioning}
\usepackage[numbers,sort&compress]{natbib}
\usepackage[colorlinks=true,allcolors=blue!45!black]{hyperref}
\usepackage{microtype}
\usepackage{placeins}
\graphicspath{{figures/}}
\newcommand{\F}{\bm F}
\newcommand{\C}{\bm C}
\newcommand{\E}{\bm E}
\newcommand{\Sg}{\bm S}
\newcommand{\I}{\bm I}
\newcommand{\ee}{\bm e}
\newcommand{\ssv}{\bm s}
\newcommand{\qq}{\bm q}
\newcommand{\dd}{\mathrm d}
\newcommand{\R}{\mathbb R}
\newcommand{\GL}{\mathrm{GL}^{+}(2)}
\newcommand{\cof}{\operatorname{cof}}
\newcommand{\tr}{\operatorname{tr}}
\newcommand{\diag}{\operatorname{diag}}
\newcommand{\doi}[1]{\href{https://doi.org/#1}{doi:\nolinkurl{#1}}.}
\newtheorem{proposition}{Proposition}

\newcolumntype{Y}{>{\raggedright\arraybackslash}X}
\title{Polyconvex neural-network surrogates with learned paired directional features for anisotropic hyperelastic homogenization}
\author{
S. Ares de Parga$^{1,4,\dagger}$ \and
A. Cornejo$^{1,2}$ \and
J.~A. Hern\'andez$^{1,3}$ \and
R. Rossi$^{1,2}$
}
\date{\parbox{.97\textwidth}{\centering\footnotesize
$^1$Centre Internacional de M\`{e}todes Num\`{e}rics en Enginyeria (CIMNE), Barcelona, Spain\\[.3em]
$^2$Universitat Polit\`{e}cnica de Catalunya, Department of Civil and Environmental Engineering (DECA), Barcelona, Spain\\[.3em]
$^3$Universitat Polit\`{e}cnica de Catalunya, Departament de Resist\`{e}ncia de Materials i Estructures a l'Enginyeria, E.S. d'Enginyeries Industrial, Aeroespacial i Audiovisual de Terrassa (ESEIAAT), Terrassa, Spain\\[.3em]
$^4$Department of Aeronautics and Astronautics, Stanford University, Stanford, CA 94305, USA\\[.7em]
\textsuperscript{\(\dagger\)} This work was conducted while S. Ares de Parga was affiliated with CIMNE.\\[.7em]
26 September 2026}}

\begin{document}
\maketitle
\begin{abstract}
Approximating the effective anisotropic hyperelastic response of a representative volume element (RVE) requires constitutive expressiveness, mechanical consistency, and computational efficiency. This work develops a polyconvex physics-augmented neural network (PANN) construction based on paired directional features whose orientations, stretch powers, and determinant exponents adapt to the RVE response. The features are composed with convex, componentwise nondecreasing cores, realized as an input-convex neural network (ICNN) or an input-convex Kolmogorov--Arnold network (ICKAN). The complete energies satisfy objectivity, an exactly stress-free reference, two-dimensional polyconvexity, and prescribed growth conditions by construction, and nonnegative energy is certified separately for the selected models. On a multicavity and a single-cavity periodic cell, these energies are compared with an Unconstrained energy, trained identically without the curvature and growth constraints, and with three hyperreduced projection-based reduced-order models (HPROMs), which serve as physics-based references. On the first cell, learning the features lowers the stress error sevenfold, and six features suffice on both. The constrained energies remain below one percent in stress error and in the displacement error of a tensile-coupon FE$^2$ application with a full-order reference, whereas the Unconstrained energy and the HPROMs are more accurate inside the data. Beyond the data, however, the Unconstrained energy loses rank-one convexity, and under confined compression it stops below half of the 500~MPa that the constrained energies reach. These results quantify what the guarantees cost in accuracy and what they secure beyond the data.
\end{abstract}

\section{Introduction}
\label{sec:intro}
The effective response of a heterogeneous material emerges from the deformation of its constituents and their interactions with the microstructure. At finite strains, this interaction can produce anisotropy and coupling between loading modes that are difficult to capture with a predetermined macroscopic constitutive law. Computational homogenization addresses this difficulty by resolving a representative volume element (RVE), in which an imposed deformation drives a microscopic boundary-value problem whose solution yields the effective stresses and constitutive derivatives~\cite{Miehe1999,kouznetsova2001approach,geers2010multi}. The effective law is thus computed rather than prescribed in advance.

This flexibility comes at a computational cost. Characterizing the material over a range of deformations already requires many nonlinear finite element solves, and using the same RVE within a structural calculation multiplies those solves across macroscopic integration points and load increments~\cite{feyel2003multilevel,geers2010multi}. Much of this repeated work concerns closely related microscopic states. Surrogate modeling seeks to exploit that relationship by computing representative solutions offline and reusing their information in less expensive online evaluations. The goal is to reduce the cost of these evaluations without losing the constitutive behavior that made the RVE description useful.

Two complementary approximation routes are relevant, but they play different roles in this work. Non-intrusive constitutive surrogates learn the effective response from RVE input--output data and dispense with the microscopic solve online~\cite{le2015computational}. Intrusive projection-based reduced-order models (PROMs) instead compress the microscopic displacement field and retain projected equilibrium equations~\cite{benner2015survey,yvonnet2007,Hernandez2014}, with hyperreduction limiting the local contributions required to evaluate them~\cite{hernandez2017}. This work follows the first route and learns the effective stored energy while embedding the constitutive restrictions needed for its mechanical use. The intrusive route is retained as a reference, because it preserves an approximate microscopic state and the underlying constituent law.

The central contribution is a family of learnable paired directional features for polyconvex anisotropic energies. Their orientations and exponents adapt to the RVE response, while their paired construction makes the stress vanish exactly at the undeformed state without compromising polyconvexity. Coupling these features to either an input-convex neural network (ICNN) or an input-convex Kolmogorov--Arnold network (ICKAN) gives two realizations of the same constitutive design. Hyperreduced PROMs of the same RVE then provide accuracy--cost references. Two structural examples follow, a tensile coupon and a confined compression test beyond the data, to show how the learned energies behave once they leave isolated material evaluations.

\subsection{Energy-based constitutive learning and polyconvexity}
Direct constitutive learning moves the expensive microscopic solves to the data-generation stage and replaces them online by a trained material model. Neural homogenization has demonstrated this separation~\cite{Aldakheel2023}, making systematic sampling of representative microscopic responses an essential part of the offline workflow~\cite{kalina2023}. What the surrogate must represent depends on the material class. History-dependent responses may require learned history-to-response maps~\cite{Liu2022} or internal descriptors such as maximum directional stretches~\cite{Ghaderi2020}, while interacting scales may motivate staggered data generation and training~\cite{AsadFarhat2026}. The present work, by contrast, considers history-independent hyperelasticity. This removes memory from the state but requires the stress to derive from a stored-energy function.

Learning that function, rather than fitting each stress component independently, makes the stress work conservative and gives a major-symmetric constitutive tangent by differentiation. This is a central motivation for the mechanics-informed energy formulation of As'ad et al.~\cite{asad2022}. Expressing the energy through objective deformation measures then removes dependence on a superposed rigid rotation. In this way, physics-augmented neural networks (PANNs) shift part of the burden from learning mechanical identities from data to satisfying them through the representation itself~\cite{klein2022,linden2023}. Prediction accuracy remains to be established, but potential structure and objectivity need not depend on how densely the training domain has been sampled.

Preserving a potential, however, leaves its curvature free. At finite strains, that control must also be imposed in appropriate variables. Convexity in Green--Lagrange strain~\cite{asad2022} is attractive because stress and tangent follow directly from the learned energy. Its implication for curvature with respect to the deformation gradient nevertheless contains a stress-dependent term~\cite{GaoNeffRoventaThiel2017}, so it is not a global polyconvexity guarantee. Conversely, Xu et al.~\cite{Xu2021} parameterize an incremental tangent through a lower-triangular factorization to enforce positive semidefiniteness. This controls the tangent matrix but does not define the scalar potential required for integrability. The modeling task is therefore to preserve both the potential and the desired finite-strain curvature condition within one construction.

Polyconvexity provides a route to that combination. It represents the energy as a convex function of deformation minors, yielding a sufficient condition for rank-one convexity~\cite{ball1977}. For the two-dimensional model considered here, the independent arguments are the deformation gradient and its determinant. This framework also underpins existence results when the additional coercivity and admissibility hypotheses are met~\cite{Ciarlet1988}. Polyconvexity serves here as a sufficient constitutive restriction, not as a necessary property of every homogenized material or a guarantee of structural uniqueness. The remaining challenge is to keep enough freedom within these restrictions to describe the anisotropic RVE response.

\subsection{Directional features and representational expressiveness}
\label{sec:discussion}
For anisotropy, this challenge is especially consequential, because a flexible scalar core can distinguish only the deformation states that its inputs distinguish. Structural tensors encode material directions in objective scalar quantities and underpin established anisotropic polyconvex energies~\cite{SchroderNeff2003}. The directional description is itself a constitutive choice. Models of fiber-reinforced tissue, for example, must account for a distribution of orientations rather than one perfectly aligned family when that dispersion affects the mechanical response~\cite{Gasser2006}. A learned feature map must likewise expose the relevant directional distinctions while remaining compatible with objectivity and convexity in the deformation minors. It is therefore part of the constitutive model, not merely a preprocessing step.

A convex scalar core does not resolve this choice. An ICNN provides a convex function~\cite{amos2017}, but a sufficient composition with nonlinear features requires those features to be convex in the independent deformation minors and the core to be componentwise nondecreasing~\cite{klein2022}. This monotonicity imposes an order on feature space, since $\bm z(\F_a)\leq\bm z(\F_b)$ componentwise implies $H(\bm z(\F_a))\leq H(\bm z(\F_b))$. This ordering persists for any network width, and the reference corrections and growth terms added later to complete the energy leave it unchanged in the learned contribution, so a poorly chosen feature family can impose a representational obstruction before network capacity or optimization becomes limiting.

The numerical evidence of Klein et al.~\cite{klein2022} illustrates this distinction. For their cubic beam-lattice data, an exactly objective invariant-based polyconvex model fails to represent the shear response, whereas a model formulated in $\F$, $\cof\F$, and $\det\F$ gives a substantially better fit while approximating objectivity through data augmentation. The same invariant-based model performs very well for their analytically generated transversely isotropic data. In these results, accuracy depends on whether the admissible inputs expose the distinctions required by the material response, rather than on polyconvexity alone. This motivates enlarging the feature family before concluding that a convex core is inadequate.

The natural response is to make the directional description adaptive without abandoning the admissible composition. Generalized-structure-tensor models calibrate the structural tensors together with the energy model, allowing the data to identify aspects of the anisotropy~\cite{Kalina2024}. Recent triclinic invariant constructions and group symmetrization further broaden the polyconvex descriptions available for finite material-symmetry groups~\cite{klein2026}. Greater directional freedom, however, must be reconciled with the undeformed state. A learned anisotropic contribution need not produce zero reference stress, while cancelling that stress with an arbitrary strain correction can violate polyconvexity or material symmetry. The invariant-based normalization and analytic growth terms of Linden et al.~\cite{linden2023} demonstrate why these operations must belong to the admissible constitutive construction. The feature family must therefore do two jobs at once: distinguish the relevant directional response and control its reference derivatives so that exact stress normalization remains compatible with polyconvexity. This is the role of the paired features introduced here.

While the feature map determines what the model can distinguish, the scalar core determines how those distinctions are combined into an energy. Prior constructions distribute this flexibility in different ways. Constitutive artificial neural networks (CANNs) learn coefficients and exponents within an interpretable catalog of admissible energy terms~\cite{Linka2023}, whereas neural ordinary differential equations can learn monotone energy derivatives and recover convex energy contributions by integration~\cite{Tac2022}. In the present work, an ICNN combines the paired features through nonnegative weights and fixed convex activations. A Kolmogorov--Arnold network (KAN) instead places trainable univariate functions on its edges, a flexibility explored in constitutive KANs~\cite{Abdolazizi2025}. The feature-composition argument holds only if these edge functions are convex and nondecreasing, as in monotone input-convex KANs (ICKANs)~\cite{thakolkaran2025}. A smooth ICNN and an integrated-hat ICKAN are compared under the same paired features and analytic completion. The comparison isolates how the scalar dependence is represented, while the directional information and the mechanical guarantees stay the same.

\subsection{Projection-based reduced-order models}
Projection-based reduction retains the RVE boundary-value problem but seeks its solution in a space learned from previous microscopic states. In an affine PROM, those states define a basis, commonly obtained through proper orthogonal decomposition (POD), and projected equilibrium determines the coordinates of each new displacement field. Reduced finite-strain homogenization has shown that this strategy can reuse information from earlier RVE solves without replacing the constituent law~\cite{yvonnet2007,Hernandez2014,Hernandez2020}. PROMs therefore provide a natural physics-based reference for the present study, since they preserve more of the microscopic problem than an effective-energy surrogate at the price of an online equilibrium solve.

The affordability of that solve depends first on the geometry of the solution family. A small number of loading parameters does not necessarily imply that the corresponding microscopic states lie close to a small global linear space. When their Kolmogorov $n$-width decays slowly, many basis vectors, and hence many equilibrium unknowns, are required to achieve adequate accuracy. Considerable effort has been devoted to mitigating this Kolmogorov barrier with richer reduced representations, from local linear spaces~\cite{Amsallem2012} to nonlinear manifolds built with autoencoders~\cite{Lee2020}, quadratic corrections~\cite{Barnett2022}, or learned closures in a POD coordinate space~\cite{barnett2023,sibuet2025discrete,aresdeparga2026nonlinear,Chmiel2024}. The last of these, PROM--ANN~\cite{barnett2023}, offers the compromise sought here, because it enriches the geometry of the reduced state without giving up either its POD description or projected equilibrium. Equilibrium determines only a small set of primary coordinates, while a regression model, here an artificial neural network (ANN), reconstructs the secondary coordinates that populate a much richer POD expansion~\cite{aresdeparga2026nonlinear}. The microscopic displacement can thereby draw on many modes while the online nonlinear system remains small, which motivates the primary--secondary construction adopted for the RVE in this work.

Reducing the number of unknowns does not make the reduced equations inexpensive, because assembling their residual and tangent may still require constitutive evaluations over the complete microscopic mesh at every iteration. While hardware acceleration performs this full-mesh work faster~\cite{Fritzen2016}, hyperreduction reduces how much of it must be performed, turning a PROM into a hyperreduced PROM (HPROM). Methods that act directly on the projected equations~\cite{Grimberg2021,AresDeParga2023}, instead of first reconstructing high-dimensional quantities as discrete empirical interpolation does~\cite{Chaturantabut2010}, include the linear-programming empirical quadrature procedure~\cite{patera2017lp,yano2019lp}, energy-conserving sampling and weighting (ECSW)~\cite{Farhat2014,farhat2015}, and the empirical cubature method (ECM)~\cite{hernandez2017,Hernandez2020}. Each replaces full assembly by a nonnegative weighted sum of selected element contributions while leaving the local constituent evaluations intact. This work follows ECM because previous developments extend its fixed rule, in which one support and one weight vector serve the whole training space, toward movable points~\cite{Hernandez2024}, subspace-adaptive weights~\cite{Bravo2024}, and finally the manifold-adaptive weights ECM (MAW--ECM)~\cite{hernandez2026dimensional}, which is the path needed to couple hyperreduction with PROM--ANN. In MAW--ECM, the weights become fields over the latent coordinates, so the same primary coordinates that locate the microscopic state on the PROM--ANN manifold also select the rule used to evaluate its equilibrium and stress.

Together, these choices yield three intrusive references for the comparison. The affine HPROM retains a linear POD state, solves projected equilibrium, and evaluates it with fixed ECM. HPROM--ANN replaces the linear state by the primary--secondary manifold and the fixed rule by MAW--ECM, but still determines its primary coordinates through equilibrium. Its direct variant, D-HPROM--ANN, removes that final online solve by evaluating the decoder at the imposed strain, as in the hyperelastic example of~\cite{hernandez2026dimensional}. Moving across the three models removes successive portions of the online microscopic work while preserving an approximate microscopic state and selected evaluations of the constituent law. They contextualize the accuracy and cost of the proposed energy surrogates but do not share their constitutive guarantees, because the approximate effective stress of a PROM need not derive from a single effective potential, even when the constituent law is retained and the cubature weights are nonnegative.

\subsection{Contribution and outline}
Against these intrusive physics-based references, the central route pursued in this work is an effective-energy surrogate whose anisotropy is learned without relinquishing polyconvexity or exact reference-state normalization. The construction rests on paired directional features with trainable orientations, stretch powers, and determinant exponents constrained by explicit joint-convexity bounds.

The pairing is designed so that the strain derivative of every feature at the undeformed state is proportional to the identity tensor, even though its response away from that state remains anisotropic. The complete learned contribution can consequently be normalized by a single scalar correction affine in the determinant. This correction cancels the reference stress without tying that cancellation to the independently adjustable volumetric stiffness, allowing the directional description to adapt to the RVE while the reference state remains under analytic control.

With the directional information in this form, two monotone convex cores provide alternative representations of its scalar energetic dependence: a smooth ICNN and an integrated-hat spline ICKAN. In both cases, the learned core, reference correction, and growth terms form one complete energy for which the analysis establishes objectivity, an exactly stress-free reference, two-dimensional polyconvexity, and the specified growth limits. Global nonnegative energy is treated separately through a sufficient lower bound evaluated for the trained parameters. The novelty lies in bringing established convex-network and affine-determinant normalization principles together through the paired anisotropic family and its explicit analytical conditions.

The assessment works at two levels. At the constitutive level, comparing the constrained energies on independent RVE responses with the Unconstrained energy, a network trained under the same protocol without the curvature and growth constraints, measures the accuracy that the guarantees cost. On the SC-RVE, the three HPROMs described above enter this comparison as physics-based references. The constitutive study also compares learned with fixed paired directions and locates where the Unconstrained energy loses rank-one convexity beyond the data, while checks of the tangent and of the sign of the energy probe behavior that a fit metric cannot reveal~\cite{Tac2024}. At the structural level, the learned energies replace the RVE query of the first-order two-level finite element formulation, denoted FE$^2$, in which an RVE problem supplies the material response at each macroscopic integration point~\cite{feyel2003multilevel}. In a tensile coupon, they are compared with a nested full-order model (FOM)--FE$^2$ reference, and the HPROMs serve again as physics-based accuracy--cost benchmarks. A confined compression test then drives the same energies beyond the data to ask which of them keep the structural problem solvable.

The conclusions are tied to the two periodic cells, plane-strain kinematics, and the two structural loadings. Agreement with FOM--FE$^2$ measures surrogate error within the adopted homogenization model, not error relative to a structure in which every microscopic cell is resolved explicitly, and the reported timings cover the online workflows only.

The remainder of the paper is organized as follows. Section~\ref{sec:fom} defines the homogenization problem and work-conjugate variables. Section~\ref{sec:learned} introduces the admissibility criteria, constructs the learned energies, and proves their properties, after which Section~\ref{sec:reduced} defines the HPROMs used in the comparison. Section~\ref{sec:cell} reports the numerical examples, first the multicavity RVE (MC-RVE) representation study and then the single-cavity RVE (SC-RVE) deployment study, which ends with the coupon FE$^2$ comparison and a confined compression test beyond the data. Section~\ref{sec:conclusions} closes the paper.

\FloatBarrier

\section{Finite-strain homogenization and the reference problem}
\label{sec:fom}
\subsection{Kinematics and work-conjugate variables}
Let $\F\in\GL$ denote the prescribed in-plane macroscopic deformation gradient. Its associated right Cauchy--Green tensor and Green--Lagrange strain are
\begin{equation}
 \C=\F^{T}\F,\qquad
 \E=\tfrac12(\C-\I),\qquad J=\det\F>0.
\end{equation}
The macroscopic input is written simply as $\F$ to match the constitutive notation used throughout the paper. The subscript $\mu$ marks microscopic quantities, whereas the unsubscripted $W$, $\bm P$, and $\Sg$ denote effective quantities obtained by homogenizing the cell. Plane strain is adopted, so the corresponding three-dimensional gradient is $\diag(\F,1)$, and the same embedding applies locally to $\F_\mu$. This constrains the out-of-plane deformation, not the out-of-plane stress.

The constitutive input used throughout is the three-component engineering representation of the Green--Lagrange strain. With $\Sg$ the effective second Piola--Kirchhoff stress, the work-conjugate convention is
\begin{equation}
 \ee=(E_{11},E_{22},2E_{12})^T,\qquad
 \ssv=(S_{11},S_{22},S_{12})^T,\qquad
 \Sg:\dd\E=\ssv\cdot\dd\ee.
 \label{eq:voigt}
\end{equation}
Thus $\gamma_{12}=2E_{12}$ is paired with $S_{12}$, with no additional factor of two in the stress vector. For a differentiable effective energy density $W(\ee)$, this convention gives the stress and the constitutive tangent
\begin{equation}
 \ssv=\frac{\partial W}{\partial\ee},\qquad
 \bm D=\frac{\partial\ssv}{\partial\ee}=\frac{\partial^2 W}{\partial\ee\,\partial\ee}\in\R^{3\times3},\qquad
 D_{ij}=\frac{\partial s_i}{\partial e_j}.
 \label{eq:tangent}
\end{equation}
This convention is used consistently in training, work integrals, and tangent checks.

An admissible strain input must satisfy $\I+2\E\succ0$. The RVE is queried with the representative
\begin{equation}
 \F(\ee)=\sqrt{\I+2\E},
 \label{eq:strain_representative}
\end{equation}
where the square root is symmetric positive definite. Objectivity allows a superposed macroscopic rotation to be removed without changing the energy or second Piola--Kirchhoff stress. The choice remains compatible with anisotropy, because material directions stay fixed in the reference configuration, and it places no symmetry requirement on the local microscopic gradients.

\subsection{Periodic microscopic equilibrium}
Let $Y$ denote the reference periodic cell and $Y_s\subset Y$ its solid part. The prescribed $\F$ determines the affine part of the microscopic deformation, whereas the fluctuation $\bm w$ is unknown,
\begin{equation}
 \bm\varphi(\bm X)=\F\bm X+\bm w(\bm X),\qquad
 \F_\mu=\F+\nabla\bm w,
 \label{eq:micro_deformation}
\end{equation}
with gradients taken in the reference configuration. Periodicity requires $\bm w(\bm X^+)=\bm w(\bm X^-)$ at paired points on opposite cell faces. Consequently, the microscopic displacement $\bm u_\mu=\bm\varphi-\bm X$ satisfies
\begin{equation}
 \bm u_\mu(\bm X^+)-\bm u_\mu(\bm X^-)
   =(\F-\I)(\bm X^+-\bm X^-).
 \label{eq:periodic_jump}
\end{equation}
The finite-strain displacement jump therefore depends on $\F-\I$, not directly on $\E$. Fixing the fluctuation at one node removes the arbitrary translation without clamping the cell boundary. The internal void surface is traction-free.

Let $\psi_\mu(\bm X,\F_\mu)$ be the solid's objective hyperelastic energy per unit reference volume and $\bm P_\mu=\partial_{\F_\mu}\psi_\mu$ its first Piola--Kirchhoff stress. In the absence of microscopic body forces, equilibrium requires
\begin{equation}
 \int_{Y_s}\bm P_\mu:\nabla\delta\bm w\,\dd V=0
 \label{eq:micro_weak}
\end{equation}
for every admissible periodic test fluctuation. Periodic kinematics and the associated traction balance provide the standard work-consistent cell problem~\cite{Miehe1999}. After finite element discretization and elimination of the periodic constraints, let $\bm a$ collect the independent fluctuation coordinates. The full-order equilibrium problem is
\begin{equation}
 \bm r(\bm a,\ee)=\partial_{\bm a}\Pi_\mu(\bm a,\ee)=\bm0,
 \qquad
 \Pi_\mu=\int_{Y_s}\psi_\mu(\bm X,\F_\mu)\,\dd V,
 \label{eq:microeq}
\end{equation}
where the strain dependence enters through Eq.~\eqref{eq:strain_representative}. Solving Eq.~\eqref{eq:microeq} gives $\bm a^*(\ee)$ and hence the microscopic fields. The microscopic deformation must also remain orientation-preserving, $\det\F_\mu>0$.

\subsection{Effective energy and stress}
On the selected equilibrium branch, the effective energy and stresses are
\begin{equation}
 W(\ee)=\frac{\Pi_\mu(\bm a^*(\ee),\ee)}{|Y|},\qquad
 \bm P=\frac{1}{|Y|}\int_{Y_s}\bm P_\mu\,\dd V,
 \qquad \Sg=\F^{-1}\bm P.
 \label{eq:homstress}
\end{equation}
The normalizing volume $|Y|$ includes the void, so these quantities describe the porous cell as a whole. For the two-dimensional discretization, the same reference thickness enters the integration weights and the cell volume, and cancels in their ratio. The stress conversion in Eq.~\eqref{eq:homstress} is performed \emph{after} averaging the first Piola--Kirchhoff stress. In general, averaging microscopic second Piola--Kirchhoff stresses does not produce $\Sg$, because $\F_\mu$ varies within the cell.

The connection between energy and averaged stress follows from equilibrium. For a macroscopic variation, $\delta\F_\mu=\delta\F+\nabla\delta\bm w$, and Eq.~\eqref{eq:micro_weak} removes the fluctuation contribution, which yields
\begin{equation}
 \frac1{|Y|}\int_{Y_s}\bm P_\mu:\delta\F_\mu\,\dd V
     =\bm P:\delta\F
     =\Sg:\delta\E
     =\ssv\cdot\delta\ee.
 \label{eq:hill_mandel}
\end{equation}
This is the Hill--Mandel work-consistency condition~\cite{Miehe1999}, which also establishes $\ssv=\partial_{\ee}W$ on a smooth branch.

The Newton solver follows an equilibrium branch, possibly through load continuation, and its convergence establishes stationarity rather than global minimization of $\Pi_\mu$ or uniqueness of the cell response. The resulting smooth branch serves as the reference constitutive response over the investigated states, without assuming that microscopic polyconvexity makes the homogenized energy polyconvex.

\subsection{Consistent constitutive tangent}
The tangent must account for the change in microscopic equilibrium as the input strain changes. To make this dependence explicit, write the stress evaluator before substitution of the equilibrium state as $\widehat{\ssv}(\bm a,\ee)$, so that $\ssv(\ee)=\widehat{\ssv}(\bm a^*(\ee),\ee)$. If $\bm r_{,\bm a}$ is nonsingular on the selected branch, implicit differentiation of Eq.~\eqref{eq:microeq} gives
\begin{equation}
 \frac{\dd\bm a^*}{\dd\ee}=-\bm r_{,\bm a}^{-1}\bm r_{,\ee},\qquad
 \bm D=\widehat{\ssv}_{,\ee}
       -\widehat{\ssv}_{,\bm a}\bm r_{,\bm a}^{-1}\bm r_{,\ee}.
 \label{eq:ift}
\end{equation}
All partial derivatives are evaluated at equilibrium, holding the other argument fixed. The first term differentiates the stress at fixed fluctuation coordinates, and the second accounts for microscopic relaxation. Omitting that correction generally does not give the derivative of the deployed constitutive response. For the smooth hyperelastic branch, the result is also $\partial^2_{\ee\ee}W$.

The sensitivities in Eq.~\eqref{eq:ift} follow from linear systems with the converged microscopic tangent, so no explicit matrix inverse is formed. The reference implementation uses one nonlinear equilibrium solution followed by three sensitivity right-hand sides, one per strain component. Independent finite-difference checks, which perturb the input and re-solve equilibrium, verify this linearization. The same consistency requirement applies to each surrogate, whose tangent must differentiate its own stress evaluation.

\subsection{Macroscopic equilibrium and the constitutive interface}
The pair $(\ssv,\bm D)$ is the common interface through which every constitutive evaluator enters the structural problem. Let $\bm U$ collect the free macroscopic nodal displacements, and let $\ee_g(\bm U)$ be the engineering Green strain at integration point $g$. With $\bm B_g=\partial\ee_g/\partial\bm U$, reference-volume integration weights $\Omega_g$, and dead external loads $\bm f_{\rm ext}$, equilibrium reads
\begin{equation}
 \bm R_M(\bm U)=
       \sum_g\Omega_g\bm B_g^T\ssv_g(\ee_g(\bm U))
       -\bm f_{\rm ext}=\bm0.
 \label{eq:macro_residual}
\end{equation}
Its tangent contains material and geometric contributions,
\begin{equation}
 \bm K_M=\sum_g\Omega_g\left[
      \bm B_g^T\bm D_g\bm B_g+
      \sum_{\alpha=1}^{3}(s_g)_\alpha
       \frac{\partial^2(e_g)_\alpha}{\partial\bm U\,\partial\bm U}
       \right].
 \label{eq:macro_tangent}
\end{equation}
The second term arises because the Green strain is nonlinear in the displacements. A structural Newton solver therefore needs both $\ssv_g$ and $\bm D_g$, irrespective of how the constitutive response is computed.

The nested FOM--FE$^2$ evaluation provides the reference instance of this interface, in which every constitutive query solves the full periodic RVE problem defined above. The learned energies of Section~\ref{sec:learned} return $\ssv$ and $\bm D$ by differentiation and remove the microscopic solve from the online stage. The HPROMs of Section~\ref{sec:reduced} instead keep an approximate microscopic state. The affine HPROM and HPROM--ANN obtain it by solving equilibrium in reduced coordinates with hyperreduced integration, whereas D-HPROM--ANN evaluates it without that solve. Despite these different routes, all models supply the same interface while the macroscopic mesh, loading, and equilibrium algorithm remain unchanged.

Holding that structural problem fixed makes the nested FOM--FE$^2$ solution the reference for surrogate errors at fixed macroscopic and microscopic discretizations, so homogenization, boundary-layer, and discretization errors remain outside the reported comparison. The cell itself is an idealized periodic material model rather than a statistically representative sample of a random microstructure.

\section{Learnable paired-feature polyconvex energies}
\label{sec:learned}
The non-intrusive route developed here replaces the repeated online solution of the microscopic problem by a learned effective stored energy. Its central constitutive relation is
\begin{equation}
 \ssv(\ee)=\nabla_{\ee}W(\ee),
 \label{eq:energy_stress_relation}
\end{equation}
where $\ssv$ is the stress work-conjugate to the engineering-strain vector $\ee$. All stress components then derive from a common potential, and the Hessian of $W$ supplies a major-symmetric constitutive tangent.

One could stop at this point and construct a flexible hyperelastic network without restricting the curvature of its learned potential. With objective inputs and a stress-free reference state, this option defines the \emph{Unconstrained energy} (Section~\ref{sec:tiers}), which shows what can be represented before stronger finite-strain restrictions are introduced.

Potential structure and reference normalization still leave the deformation-gradient curvature uncontrolled and permit finite energy as $J\to0^+$. This work instead pursues polyconvexity by constructing the energy as a convex function of independent deformation minors, together with analytic terms that control the reference state and growth. The learned contribution combines a common family of paired directional features with one of two monotone convex cores, an ICNN or an ICKAN. These cores provide alternative representations while preserving the same feature construction and mechanical guarantees.

Figure~\ref{fig:constitutive_pipeline} summarizes the construction, tracing the path from RVE data through the paired directional features and either convex core to the analytic completion of the stored energy, whose derivatives supply stress and tangent. Throughout, $\F\in\GL$, so that $J=\det\F>0$, and derivatives with respect to deformation are taken with the learned parameters fixed.

\subsection{Admissibility criteria and scope of the guarantees}
\label{sec:requirements}
A small prediction error is not sufficient to define an admissible hyperelastic surrogate. Six requirements, standard in constitutive PANN formulations~\cite{klein2022,linden2023} and labeled here for later reference, organize the construction:
\begin{enumerate}
 \item[(C1)] \emph{Potential structure:} stress is the derivative of a single-valued stored energy in work-conjugate variables.
 \item[(C2)] \emph{Objectivity:} a superposed rigid rotation does not change that energy.
 \item[(C3)] \emph{Stress symmetry and anisotropy:} the second Piola--Kirchhoff stress is symmetric, while directional dependence is allowed without prescribing a material symmetry group.
 \item[(C4)] \emph{Reference normalization:} the undeformed state has zero energy and zero stress.
 \item[(C5)] \emph{Polyconvexity:} the energy admits a convex representation in deformation minors, which is sufficient for rank-one convexity.
 \item[(C6)] \emph{Growth:} energy diverges under volume collapse, unbounded volume expansion, and unbounded distortion at bounded positive volume.
\end{enumerate}
The Unconstrained energy satisfies (C1)--(C4) by construction, whereas (C5) and (C6) are not imposed on it. The ICNN and ICKAN energies satisfy all six criteria by construction for every finite parameter set that obeys the stated architectural constraints.\footnote{Global nonnegative energy is not part of (C1)--(C6). It is checked separately for the trained parameters in Sections~\ref{sec:mc_guarantees} and~\ref{sec:constitutive_accuracy}, using the sufficient bound of Appendix~\ref{app:nonnegative}.} A property counts here as guaranteed only if it follows from the representation itself or from a sufficient bound under stated assumptions. Sampled numerical checks verify the implementation but do not replace such an argument.

\begin{figure}[p]
 \centering\input{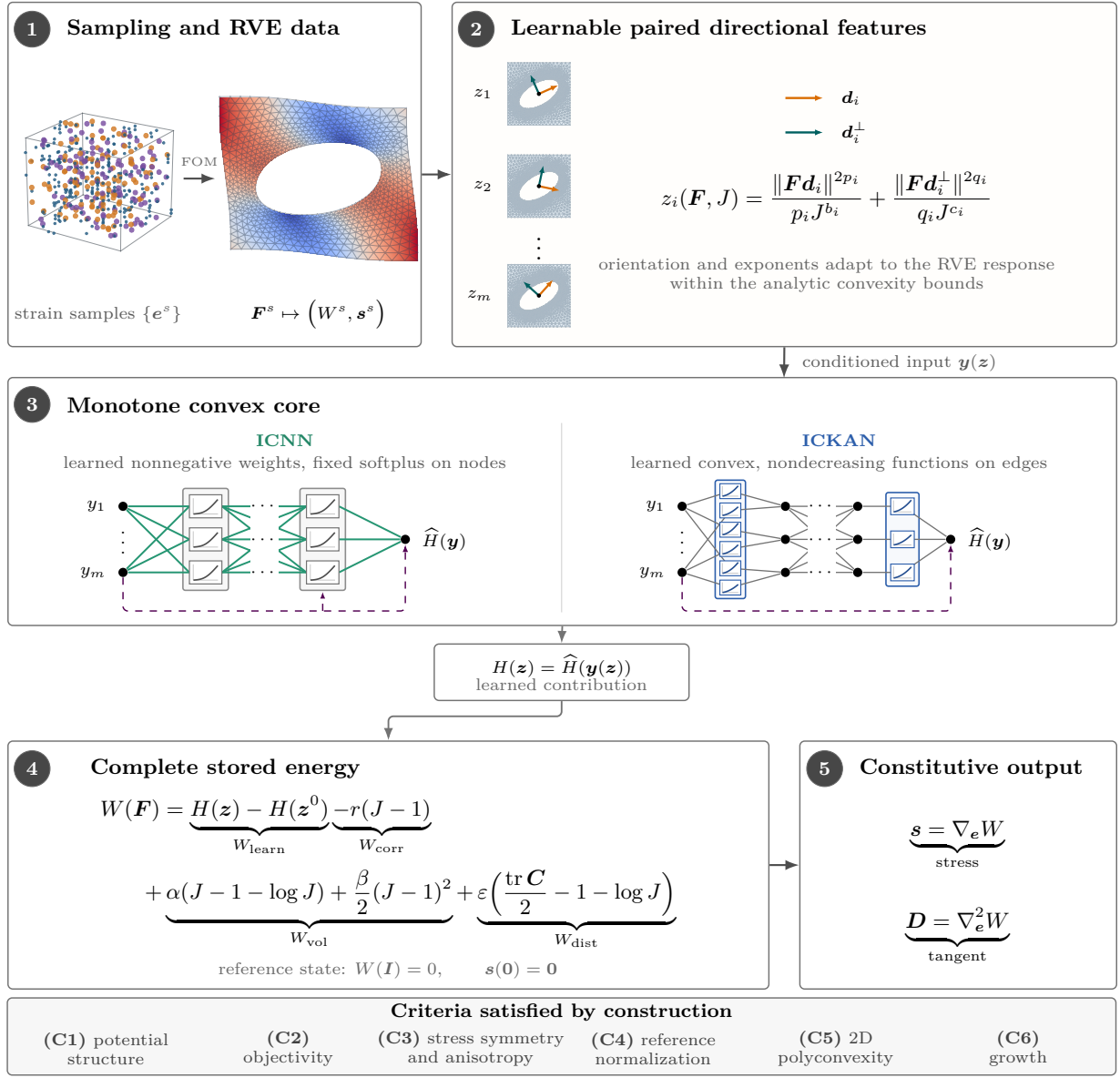}
 \caption{Overview of the proposed constitutive construction. Sampled strain states are evaluated with the periodic RVE model to supply energy and stress data. These responses train learnable paired directional features and a monotone convex ICNN or ICKAN core. Analytic reference and growth terms complete the stored energy. Differentiation then supplies the work-conjugate stress and constitutive tangent. Criteria (C1)--(C6), defined in Section~\ref{sec:requirements}, follow from the complete construction rather than from sampled numerical checks.}
 \label{fig:constitutive_pipeline}
\end{figure}

Potential structure (C1), supplied by Eq.~\eqref{eq:energy_stress_relation}, implies that the work around every closed strain path $\Gamma$ satisfies
\begin{equation}
 \oint_{\Gamma}\ssv\cdot\dd\ee=0 .
 \label{eq:conservative_requirement}
\end{equation}
For the isothermal purely elastic setting considered here, (C1) likewise gives $\ssv\cdot\dot{\ee}-\dot W=0$, so the mechanical dissipation vanishes. The stress Jacobian is then the Hessian of $W$ and obeys the cross-derivative identities $\partial s_i/\partial e_j=\partial s_j/\partial e_i$, that is, the major symmetry $\bm D=\bm D^T$ of the constitutive tangent in Eq.~\eqref{eq:tangent}.

Objectivity (C2) must remove dependence on the spatial observer without inadvertently imposing material isotropy. For a superposed proper rotation $\bm Q$, it requires
\begin{equation}
 W(\bm Q\F)=W(\F),\qquad
 \bm P(\bm Q\F)=\bm Q\bm P(\F).
\end{equation}
Using $\C=\F^T\F$ as an energy input enforces the first identity, because a superposed spatial rotation leaves $\C$ unchanged. Differentiating the objective energy then gives the stress transformation. The stress symmetry required by (C3), $\Sg=\Sg^T$, follows for the energy models considered here from differentiation with respect to symmetric strain. It implies a symmetric Cauchy stress $\bm\sigma=J^{-1}\F\Sg\F^T$, as required by angular-momentum balance in the present continuum.

Objectivity still permits anisotropy, because material symmetry is a different invariance. A right action $\F\mapsto\F\bm Q_m$ leaves the energy unchanged precisely when $\bm Q_m$ belongs to the material symmetry group. Section~\ref{sec:feature_order} implements this distinction with directions attached to the reference material. A superposed spatial rotation acts on $\F$ from the left and does not alter the directional stretch measured in that reference. The major symmetry of $\bm D$ is a third, separate property, which follows from (C1) rather than from the symmetry of the stress tensor.

Reference normalization (C4) makes the learned energy consistent with the undeformed state by requiring
\begin{equation}
 W(\I)=0,\qquad \ssv(\bm0)=\bm0,
\end{equation}
where $\ssv(\bm0)=\bm0$ is equivalent to $\bm P(\I)=\bm0$. The first condition fixes an arbitrary additive energy constant. The second is a substantive stress condition, and the admissible correction depends on the model class. For the Unconstrained energy, Section~\ref{sec:tiers} subtracts the raw reference value and a term linear in $\E$. That correction preserves its scalar potential, but an arbitrary linear function of $\E$ need not preserve polyconvexity in $\F$. The ICNN and ICKAN energies instead use a correction that is affine in the determinant, which preserves convexity in the independent minors (Section~\ref{sec:energy_admissibility}).

Polyconvexity (C5) supplies the curvature control that potential structure leaves open. Direct convexity in $\F$ is generally too restrictive for objective finite-strain energies with singular-volume growth, so convexity is imposed instead in independent deformation minors~\cite{ball1977,Ciarlet1988}. In two dimensions,\footnote{The standard three-dimensional definition uses $(\F,\cof\F,\det\F)$. The two-dimensional guarantee established here does not imply an admissible extension to arbitrary out-of-plane deformation.} this means that the energy admits the representation
\begin{equation}
 W(\F)=\mathcal P(\F,\det\F),\qquad
 \mathcal P\text{ convex in its independent arguments}.
 \label{eq:polyconvex_definition}
\end{equation}
Here $\mathcal P$ is defined on the convex set $\R^{2\times2}\times(0,\infty)$, treating its matrix and scalar arguments as independent. Only its evaluation in $W$ sets the second argument equal to $\det\F$. The cofactor is linear in the entries of a $2\times2$ matrix, so it need not be introduced as an additional independent argument. Convexity of $\mathcal P$ implies that $W$ cannot curve downward along an admissible rank-one perturbation of $\F$.

Convexity in $\C$ or $\E$, as imposed in strain-space formulations such as that of As'ad et al.~\cite{asad2022}, does not imply polyconvexity~\cite{GaoNeffRoventaThiel2017}, because $\E=\tfrac12(\F^T\F-\I)$ depends quadratically on $\F$. Appendix~\ref{app:finite_strain_curvature} derives the rank-one consequence of (C5), states its differential form, and gives an analytical counterexample for convexity in $\E$.

Growth (C6) guards against volume collapse and unbounded deformation, which a polyconvex energy may still allow. It requires $W\to\infty$ as $J\to0^+$, as $J\to\infty$, and as $\|\F\|\to\infty$ while $J$ remains bounded away from zero and infinity. These conditions contribute to coercivity. Existence of a minimizer also requires admissible boundary data and finite-energy competitors~\cite{Ciarlet1988}, and no constitutive barrier guarantees convergence of a full Newton iteration.

\subsection{Unconstrained energy reference}
\label{sec:tiers}
The Unconstrained energy is a smooth neural network $\widehat W_{\theta}$ acting on
\begin{equation}
 \bm x(\ee)=\big(C_{11}-1,C_{22}-1,C_{12},J-1\big)^T.
 \label{eq:unconstrained_inputs}
\end{equation}
Here $\C=\I+2\E$ and $J=\sqrt{\det\C}$, so the four inputs are functions of the same three engineering-strain components. Using components of $\C$ in a material reference basis ensures objectivity under superposed spatial rotations but does not impose isotropy, since the network may depend differently on $C_{11}$, $C_{22}$, and $C_{12}$. Denote the resulting raw energy by $\widetilde W(\ee)=\widehat W_{\theta}(\bm x(\ee))$. A generic network does not satisfy the reference conditions automatically, so the normalized energy is defined, following the consistency correction of As'ad et al.~\cite{asad2022}, as
\begin{equation}
 W_{\mathrm{unc}}(\ee)=\widetilde W(\ee)-\widetilde W(\bm0)
 -\partial_{\ee}\widetilde W(\bm0)\cdot\ee.
 \label{eq:unconstrained_energy}
\end{equation}
By construction, $W_{\mathrm{unc}}(\bm0)=0$ and $\partial_{\ee}W_{\mathrm{unc}}(\bm0)=\bm0$. Because the correction is affine in $\ee$, it leaves the Hessian of the raw network unchanged, so the reference energy and stress are fixed while the learned curvature remains unrestricted. The constructions that follow restrict this curvature while retaining a learned anisotropic response.

\subsection{Feature design: expressiveness and admissibility}
\label{sec:feature_order}
Classical invariants provide the natural starting point for describing deformation. In two dimensions, the principal invariants of $\C$ are $I_1=\tr\C$ and $\det\C=J^2$, and they determine all three invariants of the plane-strain embedding $\diag(\C,1)$. They describe principal stretches without recording their orientation relative to the material, so an energy depending only on them is isotropic.

To retain directional information, one can supplement the isotropic measures with a structural direction $\bm d$. The familiar anisotropic stretch invariant is
\begin{equation}
 I_4(\bm d)=\bm d^T\C\bm d
          =\C:(\bm d\otimes\bm d)=\|\F\bm d\|^2,
 \qquad \|\bm d\|=1.
 \label{eq:directional_invariant}
\end{equation}
It measures squared stretch along a material direction and remains unchanged by a superposed spatial rotation. Such directional invariants are standard ingredients of anisotropic hyperelastic energies~\cite{SchroderNeff2003,Gasser2006}. Here $\bm d$ need not represent a physical fiber. Learned directions serve as objective scalar descriptors of the RVE's effective anisotropy. In two dimensions, $\bm d$ is naturally accompanied by its perpendicular $\bm d^\perp$. For this orthonormal pair,
\begin{equation}
 I_4(\bm d)+I_4(\bm d^\perp)=I_1,\qquad
 \|\cof\F\,\bm d\|^2=I_4(\bm d^\perp).
 \label{eq:directional_invariant_2d}
\end{equation}
The first identity shows that the pair splits $I_1$ between two perpendicular directions, so that its two members retain the directional information that $I_1$ discards. The second shows that a cofactor-based directional measure reads the stretch in the perpendicular direction rather than introducing an independent third in-plane descriptor.

Objective directional descriptors solve only the first part of the design problem. Because the learned contribution $H(\bm z)$ is nondecreasing in each feature, the feature vector $\bm z(\F)$ also determines which deformation states it must order. For any two deformation states,
\begin{equation}
 \bm z(\F_B)\geq\bm z(\F_A)\text{ componentwise}
 \quad\Longrightarrow\quad H(\bm z(\F_B))\geq H(\bm z(\F_A)),
 \label{eq:feature_ordering}
\end{equation}
whatever the width of the network. As anticipated in Section~\ref{sec:discussion}, this ordering constrains the learned contribution but not the analytic normalization and growth terms, so the feature map must be designed together with them. Additional features remain functions of the same three engineering-strain components and introduce no new strain variable, but they embed the strain state into more variables on which the core is monotone and convex. They can enlarge the admissible response family without guaranteeing that it contains the target response.

The paired construction below is designed to retain directional stretch, reshape the componentwise ordering seen by the monotone core, and remain convex in the independent variables $(\F,J)$. It meets these three objectives by coupling directional stretch to volume through admissible determinant exponents. A feature can consequently decrease along a path when its denominator grows sufficiently, even if a numerator increases. Its reference derivative is also controlled analytically, which will allow the reference stress to be removed without sacrificing convexity.

\subsubsection{Paired directional features}
Let $m$ denote the number of scalar features. For each $i=1,\ldots,m$, let $\bm d_i=(\cos\theta_i,\sin\theta_i)^T$ and $\bm d_i^\perp=(-\sin\theta_i,\cos\theta_i)^T$. Define
\begin{equation}
 \begin{aligned}
 z_i(\F,J)
 &=\frac{I_4(\bm d_i)^{p_i}}{p_i J^{b_i}}
          +\frac{I_4(\bm d_i^\perp)^{q_i}}{q_i J^{c_i}}\\
 &=\frac{\|\F\bm d_i\|^{2p_i}}{p_i J^{b_i}}
          +\frac{\|\F\bm d_i^\perp\|^{2q_i}}{q_i J^{c_i}},
 \end{aligned}
 \label{eq:features}
\end{equation}
where $\theta_i$ sets the orientation of the reference directions, $p_i,q_i$ are the powers of the squared-stretch invariants, and $b_i,c_i$ control each contribution's coupling to the area ratio $J$. These roles provide expressiveness, but the exponents cannot be chosen arbitrarily if the feature is to remain convex. They satisfy
\begin{equation}
 p_i,q_i\geq\tfrac12,\qquad
 0\leq b_i\leq2p_i-1,\qquad 0\leq c_i\leq2q_i-1.
 \label{eq:featureconstraints}
\end{equation}
The lower bounds on $p_i$ and $q_i$ make the corresponding powers of the directional norms convex, while the upper bounds on $b_i$ and $c_i$ prevent their coupling with $J$ from introducing a negative-curvature direction. Appendix~\ref{app:convex} derives these conditions from the Hessian of a single stretch--determinant contribution.
In the learned-feature variants, $\theta_i$ and the four exponent parameters are optimized, with $p_i,q_i,b_i,c_i$ subject to Eq.~\eqref{eq:featureconstraints}. The fixed-feature controls of Section~\ref{sec:constitutive_training} hold the same initialized feature parameters constant while training the core. With unconstrained parameters $\widehat p_i,\widehat q_i,\widehat b_i,\widehat c_i\in\R$ and the functions $\operatorname{softplus}(x)=\log(1+e^x)$ and $\operatorname{sigmoid}(x)=(1+e^{-x})^{-1}$, whose ranges are $(0,\infty)$ and $(0,1)$, the exponents are
\begin{equation}
 \begin{aligned}
 p_i&=\tfrac12+\operatorname{softplus}(\widehat p_i),
 &b_i&=(2p_i-1)\operatorname{sigmoid}(\widehat b_i),\\
 q_i&=\tfrac12+\operatorname{softplus}(\widehat q_i),
 &c_i&=(2q_i-1)\operatorname{sigmoid}(\widehat c_i).
 \end{aligned}
 \label{eq:feature_parameterization}
\end{equation}
The angles require no admissibility bound, because the maps $\F\mapsto\F\bm d_i$ and $\F\mapsto\F\bm d_i^\perp$ remain linear in $\F$ for any fixed $\theta_i$. Equation~\eqref{eq:feature_parameterization} satisfies the bounds of Eq.~\eqref{eq:featureconstraints} strictly for all finite parameter values, approaching the boundary only as a limit, so no penalty is needed during training. Each $z_i$ is one scalar feature with two perpendicular contributions. Because their exponents need not be equal, the pairing does not impose invariance under their interchange.

The construction was also designed to simplify reference normalization. Its paired and scaled contributions make the reference strain gradient of every feature isotropic. Consequently, the learned contribution has an isotropic reference stress, which a single scalar determinant correction can cancel in Section~\ref{sec:energy_admissibility}. The following result gathers the objectivity, convexity, and reference-state properties needed there.

\begin{proposition}
\label{prop:paired_features}
Under the constraints in Eq.~\eqref{eq:featureconstraints}, each feature in Eq.~\eqref{eq:features} is objective and convex in the independent variables $(\F,J)$ on $J>0$. Its value and engineering-strain gradient at the identity, evaluated with $J=\det\F$, are
\begin{equation}
 z_i^0=\frac1{p_i}+\frac1{q_i},\qquad
 \partial_{\ee}z_i\big|_{\bm0}=\rho_i(1,1,0)^T,
 \qquad \rho_i=2-\frac{b_i}{p_i}-\frac{c_i}{q_i}.
 \label{eq:referencefeatures}
\end{equation}
\end{proposition}
The proof is given in Appendix~\ref{app:convex}. The mechanism behind the reference gradient is that $1/p_i$ and $1/q_i$ cancel the stretch-power multipliers upon differentiation. The two perpendicular directional contributions then sum to $2\I$, while the determinant terms contribute $-(b_i/p_i+c_i/q_i)\I$. Hence $\rho_i$ is the amplitude of an isotropic reference gradient, with equal normal components and zero shear. This reference property does not imply isotropy of the response away from the reference state.

The raw features are now mechanically admissible, but their numerical magnitudes may differ. Before the features enter the core, they are centered and scaled as $y_i=(z_i-z_i^{\rm cen})/\tau_i$, with a centering value $z_i^{\rm cen}$ and a scale $\tau_i>0$ that are independent of deformation for any fixed parameter set. This positive affine conditioning preserves both convexity in independent $(\F,J)$ and the componentwise ordering required by the monotone core. The choices used in the constitutive assessments are specified in Section~\ref{sec:training_protocol}.

\subsection{Monotone ICNN and integrated-hat ICKAN cores}
\label{sec:convex_cores}
Composition with the convex feature map preserves convexity only when the core, acting on the conditioned vector $\bm y=(y_1,\ldots,y_m)^T$, is convex and componentwise nondecreasing~\cite{klein2022,linden2023}. The two cores enforce these properties differently. The ICNN learns nonnegative linear maps and applies fixed softplus activations at its hidden nodes, whereas the ICKAN learns a convex, nondecreasing scalar function on every edge and sums the incoming contributions at each node (Fig.~\ref{fig:convex_core_architectures}).

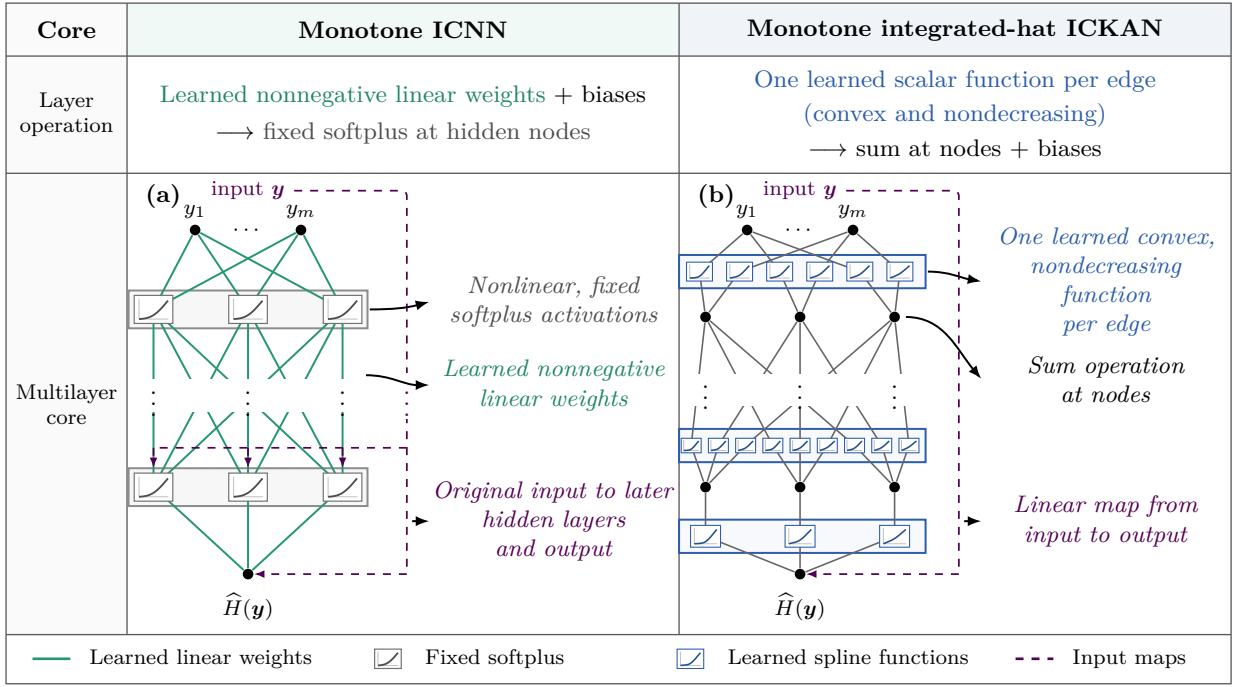
\begin{figure}[tbp]
 \centering
\definecolor{icnncolor}{HTML}{2AA780}%
\definecolor{ickancolor}{HTML}{2B5DAA}%
\begin{tikzpicture}[x=1cm,y=1cm,font=\small,
  weight/.style={draw=icnncolor!90!black,line width=.75pt},
  connection/.style={draw=black!60,line width=.6pt},
  inputmap/.style={draw=violet!65!black,dashed,line width=.65pt},
  maparrow/.style={inputmap,-{Latex[length=1.4mm]}},
  dot/.style={circle,fill=black,inner sep=0pt,minimum size=4pt},
  callout/.style={-{Latex[length=1.8mm]},draw=black,line width=.75pt},
  annotation/.style={font=\small\itshape,align=center,text width=3.15cm,inner sep=2pt},
  pics/softplus/.style={code={
    \draw[draw=black!55,fill=white,line width=.45pt]
      (-.44,-.32) rectangle (.44,.32);
    \draw[black!25,line width=.2pt]
      (-.33,-.22) -- (.34,-.22) (-.33,-.22) -- (-.33,.24);
    \draw[black!70,line width=.7pt]
      plot[domain=-3:3,samples=35]
      ({.11*\x},{-.22+.42*(ln(1+exp(\x))-ln(1+exp(-3)))/3});
  }},
  pics/hatcurve/.style args={#1/#2}{code={
    \draw[draw=ickancolor,fill=white,line width=.4pt]
      (-.44,-.32) rectangle (.44,.32);
    \draw[black!25,line width=.2pt]
      (-.33,-.22) -- (.34,-.22) (-.33,-.22) -- (-.33,.24);
    \draw[ickancolor!85!black,line width=.6pt]
      plot[domain=-1.5:1.5,samples=41]
      ({.22*\x},{-.22+.42*(#1*(\x+1.5)
        +#2*(max(\x+1,0)^3-2*max(\x,0)^3+max(\x-1,0)^3)/6)
        /(3*#1+1.5*#2)});
  }}
]
 \path[use as bounding box] (0,.45) rectangle (16.2,9.45);
 \fill[icnncolor!7] (1.6,8.75) rectangle (8.9,9.45);
 \fill[ickancolor!7] (8.9,8.75) rectangle (16.2,9.45);
 \fill[black!2] (0,1.1) rectangle (1.6,9.45);
 \draw[black!65,line width=.5pt] (0,.45) rectangle (16.2,9.45);
 \foreach \tablex in {1.6,8.9}
   \draw[black!65,line width=.5pt] (\tablex,1.1) -- (\tablex,9.45);
 \foreach \tabley in {1.1,7.2,8.75}
   \draw[black!65,line width=.5pt] (0,\tabley) -- (16.2,\tabley);
 \node[font=\small\bfseries] at (.8,9.1) {Core};
 \node[font=\small\bfseries] at (5.25,9.1) {Monotone ICNN};
 \node[font=\small\bfseries] at (12.55,9.1) {Monotone integrated-hat ICKAN};
 \node[align=center,font=\footnotesize] at (.8,7.98) {Layer\\operation};
 \node[align=center,font=\footnotesize] at (.8,4.15) {Multilayer\\core};
 \node[align=center,text width=6.6cm] at (5.25,7.98)
   {\textcolor{icnncolor!80!black}{Learned nonnegative linear weights} + biases\\[3pt]
    $\longrightarrow$ \textcolor{black!70}{fixed softplus at hidden nodes}};
 \node[align=center,text width=6.6cm] at (12.55,7.98)
   {\textcolor{ickancolor}{One learned scalar function per edge}\\[2pt]
    \textcolor{ickancolor}{(convex and nondecreasing)}\\[2pt]
    $\longrightarrow$ sum at nodes + biases};

 \begin{scope}[xshift=1.6cm,yshift=1.35cm]
  \node[anchor=west,font=\small\bfseries] at (.12,5.55) {(a)};
  \node[font=\footnotesize,text=violet!65!black] (icinputvector) at (1.6,5.62)
    {input $\bm y$};
  \foreach \level in {4.05,1.70}
    \filldraw[black!45,fill=black!3,line width=.75pt]
      (.02,\level-.25) rectangle (3.18,\level+.25);
  \node[dot] (icinput1) at (.9,5.10) {};
  \node[dot] (icinputm) at (2.3,5.10) {};
  \node[font=\footnotesize] at (.9,5.34) {$y_1$};
  \node[font=\footnotesize] at (2.3,5.34) {$y_m$};
  \node[font=\footnotesize] at (1.6,5.10) {$\cdots$};
  \node[dot] (icoutput) at (1.6,.55) {};
  \node[font=\footnotesize] at (1.6,.12) {$\widehat H(\bm y)$};
  \foreach \j in {0,...,2} {
    \pgfmathsetmacro{\hiddenx}{.35+1.25*\j}
    \draw[weight] (icinput1) -- (\hiddenx,4.05);
    \draw[weight] (icinputm) -- (\hiddenx,4.05);
    \draw[weight] (\hiddenx,1.70) -- (icoutput);
    \foreach \k in {0,...,2} {
      \pgfmathsetmacro{\nextx}{.35+1.25*\k}
      \draw[weight] (\hiddenx,4.05) -- (\nextx,1.70);
    }
  }
  \fill[white] (.02,2.70) rectangle (3.23,3.13);
  \foreach \hiddenx in {.35,1.6,2.85}
    \node at (\hiddenx,2.91) {$\vdots$};

  \draw[inputmap,rounded corners=2pt] (icinputvector.east)
    -- (3.7,5.62) -- (3.7,.55);
  \draw[maparrow] (3.7,.55) -- (icoutput.east);
  \draw[inputmap] (3.7,2.22) -- (.35,2.22);
  \foreach \hiddenx in {.35,1.6,2.85}
    \draw[maparrow] (\hiddenx,2.22) -- (\hiddenx,1.95);
  \node[fill=white,inner sep=1pt,text=violet!65!black] at (3.7,2.91) {$\vdots$};
  \foreach \j in {0,...,2} {
    \pgfmathsetmacro{\hiddenx}{.35+1.25*\j}
    \pic[scale=.58] at (\hiddenx,4.05) {softplus};
    \pic[scale=.58] at (\hiddenx,1.70) {softplus};
  }
  \node[annotation,text=black!70] (icfixed) at (5.65,4.15)
    {Nonlinear, fixed\\softplus activations};
  \node[annotation,text=icnncolor!80!black] (icweights) at (5.65,3.05)
    {Learned nonnegative\\linear weights};
  \node[annotation,text=violet!65!black] (icmaps) at (5.65,1.25)
    {Original input to later\\hidden layers and output};
  \draw[callout] (3.21,4.05) to[out=0,in=195] (icfixed.west);
  \draw[callout] (3.16,3.18) to[out=0,in=180] (icweights.west);
  \draw[callout] (3.77,1.25) -- (icmaps.west);
 \end{scope}

 \begin{scope}[xshift=8.9cm,yshift=1.35cm]
  \node[anchor=west,font=\small\bfseries] at (.12,5.55) {(b)};
  \node[font=\footnotesize,text=violet!65!black] (kaninputvector) at (1.6,5.62)
    {input $\bm y$};
  \foreach \level in {4.55,2.25,1.05}
    \filldraw[ickancolor,fill=ickancolor!4,line width=.75pt]
      (0,\level-.22) rectangle (3.26,\level+.22);
  \node[dot] (kaninput1) at (.9,5.10) {};
  \node[dot] (kaninputm) at (2.3,5.10) {};
  \node[font=\footnotesize] at (.9,5.34) {$y_1$};
  \node[font=\footnotesize] at (2.3,5.34) {$y_m$};
  \node[font=\footnotesize] at (1.6,5.10) {$\cdots$};
  \node[dot] (kanoutput) at (1.6,.55) {};
  \node[font=\footnotesize] at (1.6,.12) {$\widehat H(\bm y)$};
  \foreach \j in {0,...,2} {
    \pgfmathsetmacro{\hiddenx}{.35+1.25*\j}
    \node[dot] (kanfirst\j) at (\hiddenx,3.95) {};
    \node[dot] (kanlast\j) at (\hiddenx,1.70) {};
  }
  \foreach \j in {0,...,2} {
    \pgfmathsetmacro{\hiddenx}{.35+1.25*\j}
    \draw[connection] (kanlast\j) -- (\hiddenx,1.05) -- (kanoutput);
    \foreach \r in {0,1} {
      \pgfmathsetmacro{\edgex}{.27+.53*(2*\j+\r)}
      \ifnum\r=0
        \draw[connection] (kaninput1) -- (\edgex,4.55) -- (kanfirst\j);
      \else
        \draw[connection] (kaninputm) -- (\edgex,4.55) -- (kanfirst\j);
      \fi
    }
    \foreach \k in {0,...,2} {
      \pgfmathsetmacro{\edgex}{.16+.36*(3*\j+\k)}
      \draw[connection] (kanfirst\k) -- (\edgex,2.25) -- (kanlast\j);
    }
  }
  \fill[white] (.02,2.78) rectangle (3.27,3.20);
  \foreach \hiddenx in {.35,1.6,2.85}
    \node at (\hiddenx,2.99) {$\vdots$};
  \foreach \j in {0,...,2} {
    \pgfmathsetmacro{\hiddenx}{.35+1.25*\j}
    \pgfmathsetmacro{\slope}{.08+.25*\j}
    \pgfmathsetmacro{\curvature}{1.5-.45*\j}
    \pic[scale=.45] at (\hiddenx,1.05) {hatcurve=\slope/\curvature};
    \foreach \r in {0,1} {
      \pgfmathsetmacro{\edgex}{.27+.53*(2*\j+\r)}
      \pgfmathsetmacro{\slope}{.08+.15*\j+.2*\r}
      \pgfmathsetmacro{\curvature}{1.4-.3*\j-.2*\r}
      \pic[scale=.39] at (\edgex,4.55) {hatcurve=\slope/\curvature};
    }
    \foreach \k in {0,...,2} {
      \pgfmathsetmacro{\edgex}{.16+.36*(3*\j+\k)}
      \pgfmathsetmacro{\slope}{.08+.11*\j+.14*\k}
      \pgfmathsetmacro{\curvature}{1.5-.23*\j-.25*\k}
      \pic[scale=.30] at (\edgex,2.25) {hatcurve=\slope/\curvature};
    }
  }
  \draw[inputmap,rounded corners=2pt] (kaninputvector.east)
    -- (3.7,5.62) -- (3.7,.55);
  \draw[maparrow] (3.7,.55) -- (kanoutput.east);
  \node[annotation,text=ickancolor] (kanfunctions) at (5.65,4.42)
    {One learned convex,\\nondecreasing function\\per edge};
  \node[annotation] (kansums) at (5.65,3.13) {Sum operation\\at nodes};
  \node[annotation,text=violet!65!black] (kanmap) at (5.65,1.25)
    {Linear map from\\input to output};
  \draw[callout] (3.29,4.55) to[out=0,in=170] (kanfunctions.west);
  \draw[callout] (3.0,3.95) to[out=0,in=130] (kansums.west);
  \draw[callout] (3.77,1.25) -- (kanmap.west);
 \end{scope}

 \draw[weight,line width=1pt] (.35,.78) -- (.85,.78);
 \node[anchor=west,font=\footnotesize] at (.98,.78) {Learned linear weights};
 \pic[scale=.40] at (5.05,.78) {softplus};
 \node[anchor=west,font=\footnotesize] at (5.42,.78) {Fixed softplus};
 \pic[scale=.40] at (9.05,.78) {hatcurve=.1/1.3};
 \node[anchor=west,font=\footnotesize] at (9.42,.78) {Learned spline functions};
 \draw[inputmap,line width=1pt] (13.35,.78) -- (13.85,.78);
 \node[anchor=west,font=\footnotesize] at (13.98,.78) {Input maps};
\end{tikzpicture}
 \caption{Multilayer monotone cores mapping the conditioned paired features $\bm y$ to the scalar contribution $\widehat H(\bm y)$. (a) The ICNN learns nonnegative linear maps and applies fixed softplus activations at hidden nodes. (b) The ICKAN learns one convex, nondecreasing scalar function per edge and sums them at nodes. Purple dashed paths denote nonnegative maps of the complete input, and ellipses denote intervening layers. Glyph counts indicate where each operation acts, not deployed widths or parameter counts. Unrestricted biases are omitted.}
 \label{fig:convex_core_architectures}
\end{figure}

\subsubsection{Monotone ICNN}
\label{sec:icnn_core}
Let $L$ denote the number of hidden layers and $\bm a^{(\ell)}$ the vector of activations in layer $\ell$, for $\ell=1,\ldots,L$. Here $\widehat H(\bm y)$ denotes the scalar neural core acting on the conditioned features $\bm y$, to which the complete energy of Section~\ref{sec:energy_admissibility} adds the reference-normalization and growth terms. The deployed ICNN has the form
\begin{align}
 \bm a^{(1)}&=\phi(\bm U_1\bm y+\bm b_1),\nonumber\\
 \bm a^{(\ell)}&=\phi(\bm W_\ell\bm a^{(\ell-1)}+\bm U_\ell\bm y+\bm b_\ell),
 \qquad \ell=2,\ldots,L,\nonumber\\
 \widehat H(\bm y)&=\bm v^T\bm a^{(L)}+\bm w^T\bm y+b_{\rm out},
 \label{eq:monotone_icnn}
\end{align}
where $\phi=\operatorname{softplus}$ is applied componentwise. Its derivatives satisfy $\phi'(x)=\operatorname{sigmoid}(x)>0$ and $\phi''(x)=\operatorname{sigmoid}(x)[1-\operatorname{sigmoid}(x)]>0$. The activation is smooth, increasing, and convex, and its smoothness makes the core $C^2$, as required for a continuous constitutive tangent. The matrices $\bm U_\ell$ and $\bm W_\ell$ contain the weights from $\bm y$ and the preceding hidden layer, respectively, and $\bm b_\ell$ is the layer bias vector. The output weight vectors $\bm v$ and $\bm w$ act on $\bm a^{(L)}$ and $\bm y$, respectively, and $b_{\rm out}$ is an additive scalar output bias.

Every entry of $\bm W_\ell,\bm U_\ell,\bm v,\bm w$ is nonnegative, enforced by softplus transformations of unconstrained weight parameters. These signs play two related but distinct roles. The matrices $\bm W_\ell$ and the output vector $\bm v$ multiply convex hidden activations, so their nonnegativity preserves convexity as well as monotonicity. The maps $\bm U_\ell\bm y$ and $\bm w^T\bm y$ are affine, so they leave convexity unchanged for either sign, but nonnegative $\bm U_\ell$ and $\bm w$ are required for componentwise monotonicity. Biases are unrestricted because adding constants to a preactivation preserves both properties. Hidden biases can change the core's derivatives with respect to $\bm y$, whereas $b_{\rm out}$ only shifts $\widehat H$ and leaves its derivatives unchanged.

The preservation argument proceeds layer by layer. The first hidden layer is convex because softplus is convex and its argument is affine, and it is nondecreasing because $\bm U_1$ is nonnegative. Assuming that $\bm a^{(\ell-1)}$ is convex and nondecreasing, the nonnegative map $\bm W_\ell\bm a^{(\ell-1)}$, the nonnegative affine input map, and composition with $\phi$ preserve both properties in layer $\ell$. The same argument applies to the output, proving by induction that $\widehat H$ is convex and componentwise nondecreasing. This is the ICNN composition argument of Amos et al.~\cite{amos2017}, supplemented here by nonnegative direct-input maps to guarantee monotonicity.

\subsubsection{Monotone integrated-hat ICKAN}
\label{sec:ickan_core}
With $\bm h^{(0)}=\bm y$ and $n_\ell$ the width of layer $\ell$, where $n_0=m$, the hidden layers and scalar output of the ICKAN are
\begin{align}
 h_j^{(\ell)}
 &=\sum_{i=1}^{n_{\ell-1}}f_{ji}^{(\ell)}\!\left(h_i^{(\ell-1)}\right)
   +b_j^{(\ell)},
 &&j=1,\ldots,n_\ell,\quad \ell=1,\ldots,L,\nonumber\\
 \widehat H(\bm y)
 &=\sum_{i=1}^{n_L}f_i^{\rm out}\!\left(h_i^{(L)}\right)
   +\bm w^T\bm y+b_{\rm out},
 &&\bm w\geq\bm0.
 \label{eq:monotone_ickan}
\end{align}
Here $f_{ji}^{(\ell)}$ transforms the scalar carried by one edge and $h_j^{(\ell)}$ is the sum formed at a node. The nonnegative linear skip $\bm w^T\bm y$ connects the conditioned input directly to the scalar output.
\begin{samepage}
It remains to specify the functions appearing on the edges. For any $f_{ji}^{(\ell)}$ or $f_i^{\rm out}$ in Eq.~\eqref{eq:monotone_ickan}, let $x$ denote its scalar input. Each edge uses the common form
\begin{equation}
 f(x)=a x+\sum_k c_k B_k(x),\qquad a,c_k\geq0,
 \label{eq:splineedge}
\end{equation}
where $a$ is the learned linear slope, $c_k$ are learned weights of the spline bases $B_k$, and $k$ indexes these bases. Softplus transformations of unconstrained parameters impose the signs of $a$ and $c_k$. Additive constants are represented once, through the unrestricted node biases in Eq.~\eqref{eq:monotone_ickan}, rather than redundantly on every incoming edge.
\end{samepage}

To exclude negative slope and curvature by construction, the curvature of each edge is parameterized directly. Each $B_k''$ is a nonnegative triangular hat centered at a knot $\kappa_k$. The knots form a fixed uniform grid with spacing $h>0$ that is not optimized. Writing $s=(x-\kappa_k)/h$, one curvature basis is
\begin{equation}
 B_k''(x)=\max(1-|s|,0).
 \label{eq:hat_curvature}
\end{equation}
Integrating twice from the left conditions $B_k(x)=B_k'(x)=0$ for $x\leq\kappa_k-h$ gives
\begin{equation}
 B_k(x)=h^2\begin{cases}
 0,&s\leq-1,\\
 (s+1)^3/6,&-1<s\leq0,\\
 (-s^3+3s^2+3s+1)/6,&0<s<1,\\
 s,&s\geq1.
 \end{cases}
 \label{eq:hat}
\end{equation}
Figure~\ref{fig:integrated_hat_construction} shows both levels of the construction: integration of one curvature basis and assembly of a learned edge from several bases.
\begin{figure}[tbp]
 \centering
\begin{tikzpicture}[x=1cm,y=1cm,font=\small,
  axis/.style={draw=black!45,line width=.45pt},
  guide/.style={draw=black!30,dashed,line width=.4pt},
  curvature/.style={draw=orange!85!black,line width=1.25pt},
  slope/.style={draw=teal!75!black,line width=1.25pt},
  basis/.style={draw=blue!65!black,line width=1.25pt},
  total/.style={draw=black,line width=1.35pt},
  integrate/.style={-{Latex[length=1.5mm]},draw=black!60,line width=.55pt}
]
 \path[use as bounding box] (0,0) rectangle (16.2,6.35);
 \draw[black!55,line width=.5pt] (0,0) rectangle (7.95,6.35);
 \draw[black!55,line width=.5pt] (8.25,0) rectangle (16.2,6.35);

 \node[font=\small\bfseries] at (3.98,5.96) {(a) One integrated-hat basis};
 \node[font=\small\bfseries] at (12.23,5.96) {(b) One learned ICKAN edge};

 \begin{scope}[xshift=3.78cm,yshift=.62cm,x=1.40cm,y=.80cm]
  \foreach \yy/\lab in {4.65/{$B_k''$},2.70/{$B_k'$},.55/{$B_k$}} {
    \draw[axis] (-1.45,\yy) -- (1.55,\yy);
    \node[anchor=east] at (-1.60,\yy) {\lab};
  }
  \foreach \xx in {-1,0,1}
    \draw[guide] (\xx,.42) -- (\xx,5.55);

  \draw[curvature] (-1.45,4.65) -- (-1,4.65) -- (0,5.55)
    -- (1,4.65) -- (1.55,4.65);

  \draw[slope] (-1.45,2.70) -- (-1,2.70);
  \draw[slope] plot[domain=-1:0,samples=35]
    (\x,{2.70+.72*(\x+1)^2/2});
  \draw[slope] plot[domain=0:1,samples=35]
    (\x,{2.70+.72*(.5+\x-.5*\x^2)});
  \draw[slope] (1,3.42) -- (1.55,3.42);

  \draw[basis] (-1.45,.55) -- (-1,.55);
  \draw[basis] plot[domain=-1:0,samples=35]
    (\x,{.55+.72*(\x+1)^3/6});
  \draw[basis] plot[domain=0:1,samples=35]
    (\x,{.55+.72*(-\x^3+3*\x^2+3*\x+1)/6});
  \draw[basis] plot[domain=1:1.55,samples=15]
    (\x,{.55+.72*\x});

  \draw[integrate] (1.78,4.15) -- node[right,font=\footnotesize] {$\int\!\dd x$} (1.78,3.35);
  \draw[integrate] (1.78,2.18) -- node[right,font=\footnotesize] {$\int\!\dd x$} (1.78,1.38);
  \node[font=\footnotesize,anchor=north] at (-1,.34) {$\kappa_k-h$};
  \node[font=\footnotesize,anchor=north] at (0,.34) {$\kappa_k$};
  \node[font=\footnotesize,anchor=north] at (1,.34) {$\kappa_k+h$};
  \node[font=\footnotesize,anchor=west] at (1.58,.55) {$x$};
 \end{scope}

 \begin{scope}[xshift=12.11cm,yshift=.62cm,x=1.05cm,y=.80cm]
  \foreach \yy/\lab in {4.65/{$f''$},2.70/{$f'$},.55/{$f$}} {
    \draw[axis] (-2.05,\yy) -- (2.05,\yy);
    \node[anchor=east] at (-2.20,\yy) {\lab};
  }

  \draw[orange!85!black,line width=.6pt]
    (-1.75,4.65) -- (-1.05,4.97) -- (-.35,4.65);
  \draw[orange!85!black,line width=.6pt]
    (-.80,4.65) -- (0,5.55) -- (.80,4.65);
  \draw[orange!85!black,line width=.6pt]
    (.15,4.65) -- (.90,5.15) -- (1.65,4.65);
  \draw[total] plot[domain=-1.85:1.75,samples=90]
    (\x,{4.65+.88*(.35*max(1-abs((\x+1.05)/.7),0)
      +max(1-abs(\x/.8),0)+.55*max(1-abs((\x-.9)/.75),0))});
  \node[font=\scriptsize,text=orange!85!black] at (-1.05,5.15) {$c_1$};
  \node[font=\scriptsize,text=orange!85!black] at (0,5.74) {$c_2$};
  \node[font=\scriptsize,text=orange!85!black] at (.90,5.35) {$c_3$};

  \draw[slope] plot[smooth,tension=.65] coordinates {
    (-2.05,2.82) (-1.55,2.83) (-1.10,2.96) (-.65,3.08)
    (-.25,3.28) (.15,3.66) (.55,3.88) (.95,4.00)
    (1.35,4.16) (1.75,4.20) (2.05,4.20)};
  \draw[basis] plot[smooth,tension=.55] coordinates {
    (-2.05,.55) (-1.55,.57) (-1.05,.62) (-.55,.72)
    (-.05,.90) (.45,1.16) (.95,1.50) (1.45,1.85)
    (2.05,2.15)};

  \draw[integrate] (2.25,4.15) -- node[right,font=\footnotesize] {$\int\!\dd x$} (2.25,3.35);
  \draw[integrate] (2.25,2.18) -- node[right,font=\footnotesize] {$\int\!\dd x$} (2.25,1.38);
  \node[font=\footnotesize,anchor=west] at (2.08,.55) {$x$};
  \node[font=\scriptsize,align=center] at (0,-.03)
    {$a\geq0$: slope,\quad $c_k\geq0$: curvature};
 \end{scope}
\end{tikzpicture}
 \caption{Integrated-hat parameterization of an ICKAN edge. (a) A localized nonnegative curvature hat integrates to a nonnegative slope and a convex $C^2$ basis with a zero left tail and a linear right tail. (b) Nonnegative learned coefficients combine several fixed bases into a nonnegative edge curvature, and integration gives a nondecreasing slope and a convex, nondecreasing connection. The curves illustrate the analytic construction, not the deployed grid.}
 \label{fig:integrated_hat_construction}
\end{figure}
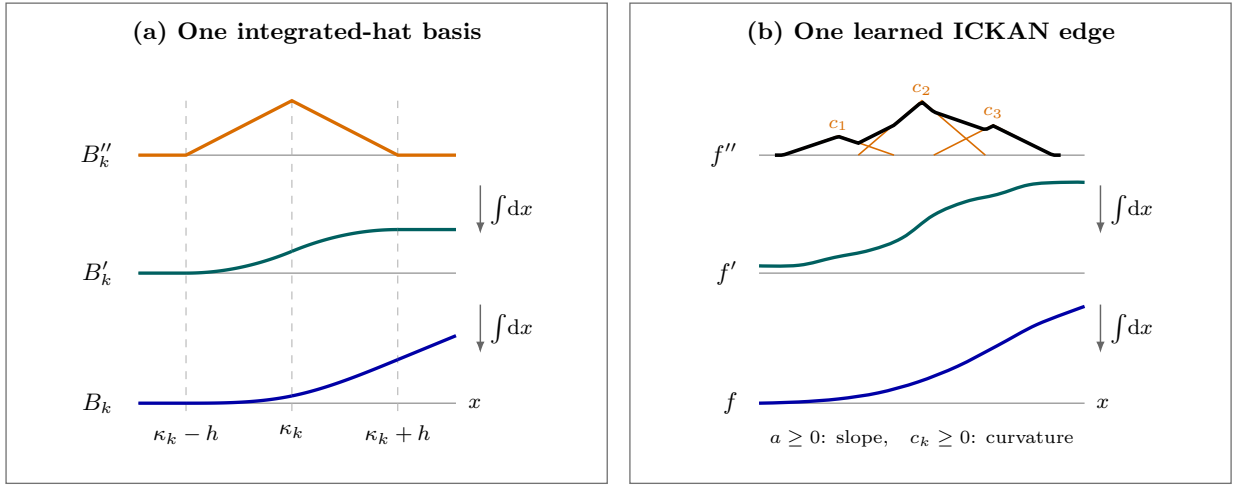

The value and first two derivatives in Eq.~\eqref{eq:hat} match at every join, and $B_k''$ vanishes at the outer joins, so $B_k$ is globally $C^2$. Its curvature is localized to $[\kappa_k-h,\kappa_k+h]$, whereas the integrated basis has a zero left tail and a linear right tail. Moreover, $B_k'(-\infty)=0$ and $B_k''\geq0$ imply $B_k'\geq0$. The complete edge therefore satisfies
\begin{equation}
 f''(x)=\sum_k c_kB_k''(x)\geq0,
 \qquad
 f'(x)=a+\sum_k c_kB_k'(x)\geq0,
 \label{eq:splineedge_guarantees}
\end{equation}
for every learned parameter set with $a,c_k\geq0$, which makes each edge convex and nondecreasing by construction.

Because every edge is convex and nondecreasing, induction across the layers of Eq.~\eqref{eq:monotone_ickan} shows that every $h_j^{(\ell)}$, and finally $\widehat H$, is convex and componentwise nondecreasing in $\bm y$. The nonnegative input-to-output skip preserves both properties. The fixed knot grid and finite number of bases limit where an individual edge can learn curvature, since outside their combined support the edge has linear tails. These tails are linear in the scalar input $x$ only, and the complete energy remains nonlinear in $\F$ through the features and the analytic growth terms. The grid used in the constitutive assessment is reported in Section~\ref{sec:training_protocol}. This core is inspired by the monotonic input-convex KAN framework~\cite{thakolkaran2025}, with the integrated-hat bases and analytic tails specified here.

\subsection{The complete energy and its admissibility guarantees}
\label{sec:energy_admissibility}
The physical feature vector $\bm z$, its conditioned representation $\bm y(\bm z)$, and the scalar neural core $\widehat H(\bm y)$ together define the learned contribution
\begin{equation}
 H(\bm z)=\widehat H\bigl(\bm y(\bm z)\bigr).
 \label{eq:conditioned_neural_contribution}
\end{equation}
This term represents the anisotropic response. The energy and stress references and the growth toward the boundary of the admissible deformation domain are set by analytic terms.

The residual stress produced by $H$ at the undeformed state is quantified first. With $\bm z^0=\bm z(\I,1)$ and the feature-gradient factors $\rho_i$ of Eq.~\eqref{eq:referencefeatures}, the core sensitivities and the scalar amplitude of this residual stress are
\begin{equation}
 h_i=\partial_{z_i}H(\bm z^0)\geq0,
 \qquad
 r=\sum_i h_i\rho_i.
\end{equation}
For a fixed parameter set, $H(\bm z^0)$ and $r$ are constants with respect to deformation. The complete stored energy then combines the learned contribution with analytic reference-normalization and growth terms,
\begin{align}
 W(\F)={}&\underbrace{H(\bm z)-H(\bm z^0)}_{\text{(i) }W_{\rm learn}}
 +\underbrace{\big[-r(J-1)\big]}_{\text{(ii) }W_{\rm corr}}\nonumber\\
 &+\underbrace{\alpha(J-1-\log J)+\frac\beta2(J-1)^2}_{\text{(iii) }W_{\rm vol}}\nonumber\\
 &+\underbrace{\varepsilon\left(\frac{\tr\C}{2}-1-\log J\right)}_{\text{(iv) }W_{\rm dist}},
 \qquad \alpha,\beta,\varepsilon>0.
 \label{eq:energy}
\end{align}
Here $J=\det\F$, $\C=\F^T\F$, and $\bm z=\bm z(\F,J)$. The positive energy-density coefficients $\alpha$, $\beta$, and $\varepsilon$ weight the analytic growth terms. For either core, the composition in Eq.~\eqref{eq:conditioned_neural_contribution} is twice continuously differentiable, convex, and componentwise nondecreasing in $\bm z$.

Each of the four contributions named by the underbraces solves a distinct constitutive problem:
\begin{itemize}
 \item[(i)] The term $W_{\rm learn}$ carries the learned response and fixes the energy reference. At $\F=\I$, the features equal $\bm z^0$, and hence
 \begin{equation}
  W(\I)=H(\bm z^0)-H(\bm z^0)=0.
  \label{eq:energy_reference_value}
 \end{equation}
 The remaining terms also vanish because $J=1$ and $\tr\C=2$. Since $H(\bm z^0)$ is constant with respect to deformation, its subtraction fixes the energy reference without changing the stress or tangent.

 \item[(ii)] The correction $W_{\rm corr}$ removes the reference stress, which a constant subtraction cannot change. By the chain rule and Eq.~\eqref{eq:referencefeatures}, the learned contribution has the reference value
 \begin{equation}
  \left.\nabla_{\ee}H(\bm z)\right|_0
  =\sum_i h_i\left.\nabla_{\ee}z_i\right|_0
  =\sum_i h_i\rho_i(1,1,0)^T
  =r(1,1,0)^T.
  \label{eq:neural_reference_stress}
 \end{equation}
 The identity $J=\sqrt{(1+2E_{11})(1+2E_{22})-(2E_{12})^2}$ gives $\nabla_{\ee}J|_0=(1,1,0)^T$. Holding $r$ fixed when differentiating yields
 \begin{equation}
  \left.\nabla_{\ee}\big[H(\bm z)-r(J-1)\big]\right|_0
  =r(1,1,0)^T-r(1,1,0)^T=\bm0.
  \label{eq:energy_reference_stress}
 \end{equation}
 A single scalar $r$ suffices because the paired features make the reference stress isotropic (Proposition~\ref{prop:paired_features}). The correction is zero in value at $J=1$ and, being affine in the independent determinant variable, preserves convexity for either sign of $r$. Linden et al.~\cite{linden2023} use the same affine determinant term to normalize isotropic energies but need additional invariant terms for anisotropic ones, whereas Eq.~\eqref{eq:referencefeatures} makes the isotropic form sufficient for the paired features.

 \item[(iii)] The volumetric term $W_{\rm vol}$ controls growth in $J$. Its logarithmic part diverges as $J\to0^+$, where the quadratic part alone would remain finite with limit $\beta/2$. The quadratic part dominates the linear correction $-r(J-1)$ as $J\to\infty$, so the energy diverges under expansion for either sign of $r$. At the reference,
 \begin{equation}
  W_{\rm vol}(1)=0,\qquad
  W_{\rm vol}'(J)=\alpha(1-J^{-1})+\beta(J-1),\qquad
  W_{\rm vol}'(1)=0,
 \end{equation}
 so this term leaves the reference energy and stress unchanged.

 \item[(iv)] The distortional term $W_{\rm dist}$ supplies growth under area-preserving changes of shape, which a function of $J$ alone cannot detect. For example, with $\lambda>0$,
 \begin{equation}
  \F=\diag(\lambda,\lambda^{-1}),\qquad
  J=1,\qquad
  \varepsilon\left(\frac{\tr\C}{2}-1-\log J\right)
  =\frac\varepsilon2\left(\lambda^2+\lambda^{-2}-2\right).
 \end{equation}
 This contribution diverges as $\lambda\to\infty$, although the area remains constant. More generally, $\tr\C=\|\F\|_F^2$ supplies growth under unbounded deformation at bounded positive $J$. At the reference, the constant $-1$ sets its value to zero. Using $\C=\I+2\E$ and $\partial_{\E}J|_0=\I$, its stress contribution is
 \begin{equation}
  \left.\partial_{\E}\left[\varepsilon\left(\frac{\tr\C}{2}-1-\log J\right)\right]\right|_0
  =\varepsilon\left[\I-\left.J^{-1}\partial_{\E}J\right|_0\right]
  =\bm0.
 \end{equation}
\end{itemize}
Together, $W_{\rm learn}$ and $W_{\rm corr}$ set the learned response and its reference state, while $W_{\rm vol}$ and $W_{\rm dist}$ prescribe growth toward the boundary of the admissible domain. Although the analytic terms vanish with their first derivatives at the reference, they still contribute to the reference stiffness, because $J$ is nonlinear in $\ee$. The following proposition gathers the guarantees of their complete sum.

\begin{proposition}
\label{prop:complete_energy}
Let the features satisfy Eq.~\eqref{eq:featureconstraints}, let the core be either of the $C^2$ convex, componentwise nondecreasing constructions in Section~\ref{sec:convex_cores}, and let $\alpha,\beta,\varepsilon>0$. Then Eq.~\eqref{eq:energy} defines a $C^2$, objective, two-dimensional polyconvex energy on $\GL$. Its associated second Piola--Kirchhoff stress $\Sg$ is symmetric, its work-conjugate representation is $\ssv=\nabla_{\ee}W$, $W(\I)=0$, $\ssv(\bm0)=\bm0$, and $W$ diverges under collapse, unbounded expansion, and unbounded distortion at bounded positive area. Thus the construction satisfies (C1)--(C6).
\end{proposition}
The proof is given in Appendix~\ref{app:complete_energy_proof}. In brief, the convex, componentwise nondecreasing core preserves feature convexity, the remaining terms are constant, affine, or convex in independent $(\F,J)$, and the reference identities and growth limits complete the argument.

Table~\ref{tab:tiers} summarizes the three constitutive constructions carried into the numerical assessment, whose ICNN and ICKAN columns include both fixed- and learned-feature variants. Importantly, learning the feature parameters changes the representation while leaving the guarantees above intact.
\begin{table}[htbp]
\centering\small
\caption{Constitutive constructions carried into the numerical assessment. Criteria (C1)--(C6) are defined in Section~\ref{sec:requirements}, and the guarantees refer to the constructions of this section.}
\label{tab:tiers}
\begin{tabularx}{\linewidth}{lYYY}
\toprule
 & Unconstrained energy & ICNN & ICKAN\\
\midrule
Role & Flexible hyperelastic reference & Guaranteed paired-feature energy & Guaranteed paired-feature energy\\
\midrule
Inputs & Objective measures from $(\C,J)$ & Paired features\newline (fixed or learned) & Paired features\newline (fixed or learned)\\
\midrule
Learned scalar map & Smooth energy network & Nonnegative weights,\newline fixed activations & Convex, nondecreasing\newline edge functions\\
\midrule
(C1)--(C4) & By construction & By construction & By construction\\
\midrule
(C5)--(C6) & Not imposed & By construction & By construction\\
\bottomrule
\end{tabularx}
\end{table}
\FloatBarrier

\section{Projection-based reduced-order models}
\label{sec:reduced}
The cost of the homogenization procedure in Section~\ref{sec:fom} lies in repeatedly resolving microscopic equilibrium as the macroscopic strain changes. Previously computed solutions offer a way to reduce this work by revealing displacement patterns that can represent subsequent states. Projection-based model order reduction builds on this information, restricting the search for equilibrium to a low-dimensional space learned from those solutions~\cite{benner2015survey,yvonnet2007,Hernandez2014}. Combined with hyperreduction to limit element evaluations, this approach retains the constituent law while reducing the cost of the microscopic problem. The resulting HPROMs provide physics-based references for assessing the accuracy and cost of the learned energies in Sections~\ref{sec:constitutive_accuracy} and~\ref{sec:coupon}.

\subsection{Periodic microscopic reference}
The reduction starts from a displacement representation that satisfies the periodic boundary conditions by construction. Let $\bm d\in\R^{N_{\rm dof}}$ collect the physical nodal displacements and $\bm u\in\R^N$ their independent coordinates after enforcing periodicity and removing rigid translation. In these coordinates, the full-order reconstruction and equilibrium equations are
\begin{equation}
 \bm d(\bm u,\ee)=\bm T\bm u+\bm g_{\rm per}(\ee),
 \qquad \bm R(\bm u;\ee)=\bm T^T\bm f_{\rm int}
                 \bigl(\bm d(\bm u,\ee)\bigr)=\bm0.
 \label{eq:hdm_lifting}
\end{equation}
Here $\bm T$ identifies periodic degrees of freedom, and the lifting $\bm g_{\rm per}$ supplies the strain-dependent displacement jumps and anchor values of Section~\ref{sec:fom}, so any independent state $\bm u$ produces a physical field that satisfies the prescribed kinematic constraints. With $\bm K_{\rm phys}=\partial_{\bm d}\bm f_{\rm int}$ the assembled physical stiffness, the FOM tangent is $\bm K=\partial_{\bm u}\bm R=\bm T^T\bm K_{\rm phys}\bm T$, and all residuals follow the internal-force sign convention.

The reduced bases are built from snapshots of this affine-retaining state, which relates to the fluctuation coordinates of Eq.~\eqref{eq:microeq} through $\bm a=\bm u-\bm u_{\rm aff}(\ee)$, with $\bm u_{\rm aff}$ the affine displacement at the independent nodes. This choice leaves the periodic boundary conditions unchanged and keeps the load-correlated information used by the strain-informed identification below, so the reported POD dimension characterizes this representation and its truncation tolerance.

\subsection{POD approximation and the affine PROM}
FOM solutions $\bm u^s$ at sampled strains $\ee^s$, $s=1,\ldots,N_s$, are centered about a reference state $\bm u_{\rm ref}$, and the singular value decomposition (SVD) of the snapshot matrix
\begin{equation}
 \bm S_{\rm snap}=
 [\,\bm u^1-\bm u_{\rm ref},\ldots,\bm u^{N_s}-\bm u_{\rm ref}\,]
 =\bm U_S\bm\Sigma_S\bm Y_S^T.
 \label{eq:snapshot_svd}
\end{equation}
supplies the POD modes. The $n_{\rm tra}$ leading left singular vectors form the orthonormal reduced-order basis (ROB) $\bm V_{\rm tra}$ of the affine approximation
\begin{equation}
 \widetilde{\bm u}=\bm u_{\rm ref}+\bm V_{\rm tra}\qq_{\rm tra},
 \qquad \qq_{\rm tra}\in\R^{n_{\rm tra}},
 \qquad \bm V_{\rm tra}^T\bm V_{\rm tra}=\bm I.
 \label{eq:linear_prom}
\end{equation}
which replaces the $N$ independent displacement coordinates by $n_{\rm tra}\ll N$ modal amplitudes. The subscript ``tra'' marks this traditional affine approximation, and the term affine refers only to the approximation space, since the constitutive law and the reduced equilibrium equations remain nonlinear.

The retained dimension is selected from the relative squared reconstruction error, which for a Euclidean POD is
\begin{equation}
 \epsilon_{\rm POD}^2(n_{\rm tra})=
 \frac{\sum_{j>n_{\rm tra}}\sigma_j^2}{\sum_j\sigma_j^2}.
 \label{eq:pod_tail}
\end{equation}
This criterion measures how well the snapshots are compressed but does not bound the stress error at unsampled strains, because force evaluation, constitutive differentiation, and equilibrium conditioning can make small displacement components important.

Galerkin projection onto the admissible variations $\delta\widetilde{\bm u}=\bm V_{\rm tra}\delta\qq_{\rm tra}$ gives the reduced equations and tangent at a prescribed strain $\ee$,
\begin{equation}
 \bm r_{\rm tra}(\qq_{\rm tra};\ee)
 =\bm V_{\rm tra}^T
  \bm R(\bm u_{\rm ref}+\bm V_{\rm tra}\qq_{\rm tra};\ee)=\bm0,
 \quad
 \bm K_{\rm tra}=\bm V_{\rm tra}^T\bm K\bm V_{\rm tra}.
 \label{eq:linear_galerkin}
\end{equation}
The system now has $n_{\rm tra}$ unknowns, but evaluating $\bm R$ and $\bm K$ still visits every microscopic element, which the hyperreduction of Section~\ref{sec:hyperreduction} addresses.

\subsection{Nonlinear reduction through latent-space closure}
\label{sec:latent_closure}
The POD dimension required to reconstruct a displacement field can exceed the number of coordinates needed to parametrize the loading family. Latent-space closure exploits this distinction by retaining a rich linear span while learning dependencies among its modal amplitudes, so that only a smaller subset remains unknown in equilibrium~\cite{barnett2023,aresdeparga2026nonlinear}. Because the leading POD coefficients need not provide a well-conditioned primary-to-secondary map, the input-informed construction of Hern\'andez et al.~\cite{hernandez2026dimensional} instead combines the retained modal amplitudes linearly into $n=3$ primary coordinates $\bm\xi$, approximately aligned with the three macroscopic strain components. It splits the retained $n_{\rm tot}$-dimensional span into orthonormal primary and secondary blocks $\bm V\in\R^{N\times n}$ and $\overline{\bm V}\in\R^{N\times\bar n}$, with $n_{\rm tot}=n+\bar n$, and an ANN closure $\widehat{\mathcal N}:\R^n\to\R^{\bar n}$ predicts the secondary coordinates $\overline{\qq}=\overline{\bm V}^T(\bm u-\bm u_{\rm ref})$. The deployed decoder and its tangent are
\begin{equation}
 \widetilde{\bm u}(\bm\xi)=\bm u_{\rm ref}
   +\bm V\bm A_m\bm\xi
   +\overline{\bm V}\widehat{\mathcal N}(\bm\xi),
 \qquad
 \bm B_\xi=\partial_{\bm\xi}\widetilde{\bm u},
 \label{eq:deployed_decoder}
\end{equation}
where $\bm A_m$ is the nonsingular coordinate transformation of the construction. The ANN predicts modal amplitudes of the microscopic displacement, not stress or energy, and the secondary modes remain in that field without appearing as independent equilibrium unknowns.

The closure is trained on the snapshot pairs $(\bm\xi^s,\overline{\qq}^{\,s})$ by minimizing
\begin{equation}
 \mathcal L_{\rm cl}(\bm\eta)=
 \frac{1}{|\mathcal I_{\rm fit}|\bar n}
 \sum_{s\in\mathcal I_{\rm fit}}
 \|\widehat{\mathcal N}(\bm\xi^s;\bm\eta)-\overline{\qq}^{\,s}\|_2^2.
 \label{eq:closure_loss}
\end{equation}
The closure must be smooth, because projection requires the derivatives of the decoder. This coordinate loss measures displacement reconstruction within the secondary space.

Because the alignment $\bm\xi^s\approx\ee^s$ is only fitted from the snapshots, the iterative PROM--ANN solves at fixed $\ee$ for the state satisfying
\begin{equation}
 \bm r(\bm\xi;\ee)=\bm B_\xi^T
       \bm R(\widetilde{\bm u}(\bm\xi);\ee)=\bm0,
 \label{eq:manifold_galerkin}
\end{equation}
rather than imposing $\bm\xi=\ee$. The exact Newton derivative then includes the curvature of the decoder manifold, since $\bm B_\xi$ depends on the state. Appendix~\ref{app:latent_closure} gives the coordinate construction, the reconstruction-error split, and this tangent explicitly.

\subsection{Hyperreduction of the projected equations}
\label{sec:hyperreduction}
Hyperreduction replaces each full element sum in the projected equations with a sparse weighted sum, a second approximation introduced after the displacement representation has been chosen. Two integration tasks must be distinguished: assembling the projected equilibrium equations and recovering the homogenized stress. Hernández et al.~\cite{hernandez2026dimensional} formulate both tasks and discuss the use of either a common or a dedicated stress rule. Separate elementwise rules are used here because the two integrands and their accuracy requirements differ.

\subsubsection{Fixed elementwise cubature}
Let $\mathcal E$ be the microscopic element set and $\bm L_e$ extract element nodal values from $\bm d$. At fixed macroscopic input, define the reconstructed element state and its reduced tangent by
\begin{equation}
 \bm d_e=\bm L_e(\bm T\widetilde{\bm u}+\bm g_{\rm per}),
 \qquad \bm B_e=\bm L_e\bm T\bm B .
\end{equation}
The full reduced residual is $\sum_{e\in\mathcal E}\bm h_e$, where
$\bm h_e=\bm B_e^T\bm f_{{\rm int},e}$. Fixed ECM approximates it as
\begin{equation}
 \bm r_{\rm ECM}(\qq;\ee)=
       \sum_{e\in\mathcal Z_r}w_e\bm h_e(\qq;\ee)=\bm0,
 \quad \mathcal Z_r\subset\mathcal E,\quad w_e\geq0.
 \label{eq:ecm}
\end{equation}
Here $\qq$ denotes the equilibrium unknowns and $\bm B=\partial_{\qq}\widetilde{\bm u}$ the decoder tangent, namely $\qq_{\rm tra}$ and $\bm V_{\rm tra}$ for the affine HPROM and $\bm\xi$ and $\bm B_\xi$ for HPROM--ANN. A selected element retains its ordinary within-element quadrature, so selected elements, selected Gauss points, and independent coordinates are distinct counts.

Offline, the element contributions are collected over $N_c$ training states and compressed by an SVD truncated, as in Eq.~\eqref{eq:pod_tail}, at a tolerance $\varepsilon_{\rm SVD}$. ECM then selects a sparse support and nonnegative multipliers that reproduce the retained integration modes~\cite{hernandez2017}. If $\bm G_{\rm cub}$ denotes those modes augmented by the constant-integration constraint (Appendix~\ref{app:cubature}) and $\bm b_{\rm cub}=\bm G_{\rm cub}\bm1$, selection stops when
\begin{equation}
 \frac{\|\bm G_{\rm cub}\bm w-\bm b_{\rm cub}\|_2}
      {\|\bm b_{\rm cub}\|_2}\leq\varepsilon_{\rm ECM}.
 \label{eq:ecm_tolerance}
\end{equation}
The residual rule integrates the element contributions $\bm h_e$, and the stress rule those to the average first Piola--Kirchhoff stress (Appendix~\ref{app:cubature}). Their supports and weights remain fixed online.

\subsubsection{Manifold-adaptive weights}
\label{sec:adaptive_weights}
A fixed rule must represent every sampled state with the same multipliers. MAW--ECM instead retains one common support but lets its weights vary over the reduced manifold~\cite{hernandez2026dimensional}. Applied to the residual, this gives
\begin{equation}
 \bm r_{\rm MAW}(\qq;\ee)=
 \sum_{e\in\mathcal Z_r}w_e(\qq)\bm B_e(\qq)^T
                   \bm f_{{\rm int},e}(\bm d_e)=\bm0.
 \label{eq:maw}
\end{equation}
Starting from a fixed ECM rule, the MAW--ECM pruning algorithm removes elements while redistributing nonnegative statewise weights, through the local and graph-regularized stages described by Hernández et al.~\cite{hernandez2026dimensional}. Pruning identifies the support but not the continuous online map, and neural fields are used here for that final regression stage. Because the equilibrium and stress rules are separate, two maps, $\mathcal N_r$ and $\mathcal N_s$, are trained on independently pruned supports $\mathcal Z_r$ and $\mathcal Z_s$,
\begin{equation}
 \begin{aligned}
  \bm w(\bm\xi)&=N_{\rm el}\operatorname{softmax}
                 \bigl(\mathcal N_r(\bm\xi;\bm\theta_r)\bigr),
   &\sum_{e\in\mathcal Z_r}w_e&=N_{\rm el},\\
  \bm v(\bm\xi)&=N_{\rm el}\operatorname{softmax}
                 \bigl(\mathcal N_s(\bm\xi;\bm\theta_s)\bigr),
   &\sum_{e\in\mathcal Z_s}v_e&=N_{\rm el}.
 \end{aligned}
 \label{eq:softmax_weights}
\end{equation}
The scaled softmax guarantees positivity and normalization. Unlike the closure, these networks predict integration weights, fitted against full-mesh element sums with the loss of Appendix~\ref{app:cubature}.

Online, the current reduced iterate is passed to both the closure and $\mathcal N_r$, and $\mathcal N_s$ supplies the stress weights, with the supports and all network parameters fixed. Differentiating the adaptive residual produces three contributions: element stiffness, decoder curvature, and weight variation. Only the first enters the microscopic modified-Newton iteration, which changes the convergence rate but not the converged state. For the affine HPROM, the decoder curvature and weight-variation terms vanish, so its projected stiffness is the exact residual Jacobian.

To translate sparse cubature into wall-clock savings, each optimized reduced mesh contains only the selected elements and the nodes they require. The corresponding basis operators and periodic lifting are restricted to that support before online evaluation. Online cost therefore depends on the support size, the retained secondary dimension, the equilibrium iterations, and the derivative evaluation, and not only on the number of primary coordinates.

\subsection{Constitutive output}
\label{sec:reduced_outputs}
Each HPROM must return the same work-conjugate stress and tangent required by the FE$^2$ problem. The stress rule reconstructs the average first Piola--Kirchhoff stress as
\begin{equation}
 \bm P_h(\qq,\ee)=
 \frac1{|Y|}\sum_{e\in\mathcal Z_s}v_e(\qq)
        \sum_{g\in e}\omega_g\bm P_{\mu,g}(\bm d_e),\qquad
 \Sg_h=\F^{-1}\bm P_h ,
 \label{eq:hyperreduced_stress}
\end{equation}
where $v_e$ is fixed for HPROM and state dependent for HPROM--ANN. Because this independently fitted rule need not give an exactly symmetric $\Sg_h$, its symmetric part supplies the returned components $\ssv_h(\qq,\ee)$.

Reduced equilibrium defines an implicit state $\qq^*(\ee)$ and hence the constitutive map $\ssv_h(\qq^*(\ee),\ee)$. If the residual Jacobian is nonsingular, differentiating equilibrium gives
\begin{equation}
 \bm D_h=\ssv_{h,\ee}
       -\ssv_{h,\qq}\bm r_{h,\qq}^{-1}\bm r_{h,\ee},
 \label{eq:reduced_ift}
\end{equation}
whose partial derivatives, evaluated at the converged state, include the decoder curvature and, for adaptive rules, the weight variation. This tangent, rather than the modified-Newton matrix of Section~\ref{sec:adaptive_weights}, is returned to the structural solver, and Appendix~\ref{app:cubature} records how its derivatives are evaluated.

When input-informed coordinates closely approximate the loading parameters, the decoder can also be evaluated without solving reduced equilibrium, as discussed by Hern\'andez et al.~\cite{hernandez2026dimensional} for their hyperelastic example. Prescribing $\bm\xi=\ee$ in Eq.~\eqref{eq:deployed_decoder} gives the direct decoder $\widetilde{\bm u}_{D}(\ee)=\widetilde{\bm u}(\bm\xi=\ee)$ of D-HPROM--ANN, whose state enters the same adaptive stress rule, with the same closure and stress weights, and whose tangent is the derivative of this direct stress map.

Figure~\ref{fig:reduced_model_hierarchy} traces the three resulting routes, from the offline RVE snapshots through the affine and nonlinear state representations to this common stress--tangent output.

\begin{figure}[p]
 \centering\input{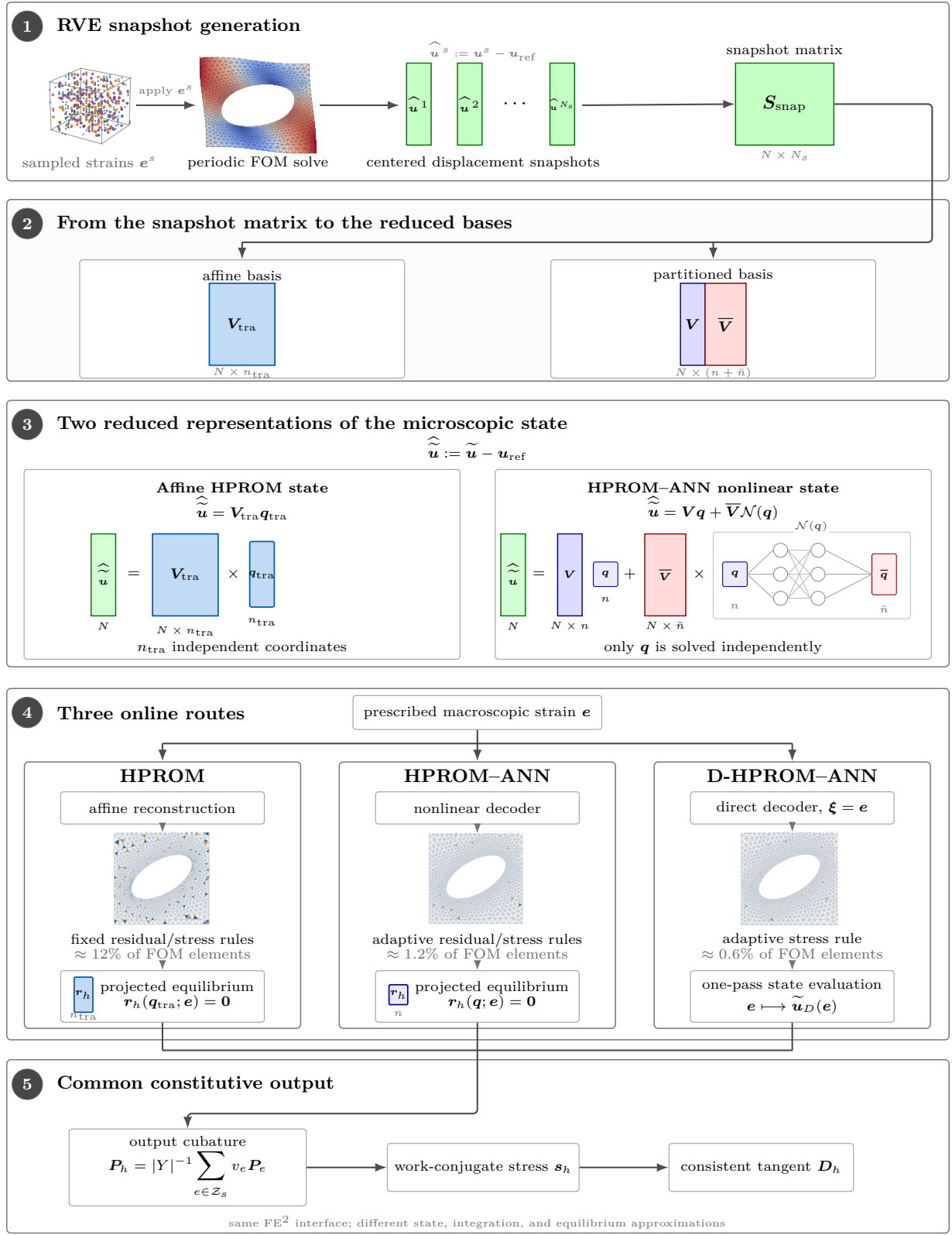}
 \caption{Intrusive HPROMs used in the comparison. FOM snapshots provide either an affine state representation or a primary--secondary partition in which the ANN predicts the secondary coordinates. Online, HPROM and HPROM--ANN solve projected equilibrium with fixed and adaptive cubature, respectively, whereas D-HPROM--ANN evaluates the nonlinear decoder directly. All three retain microscopic constitutive evaluations and return the same stress--tangent interface.}
\label{fig:reduced_model_hierarchy}
\end{figure}

\subsection{Mechanical structure retained by the HPROMs}
Every HPROM, unlike a direct constitutive surrogate, evaluates the original microscopic hyperelastic law, and the iterative variants also enforce projected equilibrium. Even so, the deployed macroscopic stress need not be the derivative of an effective potential. For fixed weights and consistent kinematics, Galerkin internal forces derive from a weighted reduced energy, as in the structure-preservation results for ECSW~\cite{farhat2015}. Here, however, equilibrium and stress use separate cubature rules, and state-dependent weights add the term $\sum_e\Pi_e\nabla_{\qq}w_e$ to the gradient of the weighted energy $\sum_e w_e(\qq)\Pi_e(\qq,\ee)$, a term that Eq.~\eqref{eq:maw} omits. Consequently, no exact effective potential or polyconvexity certificate is claimed for the HPROMs, whose consistent tangent is the derivative of the implemented stress map and need not be symmetric.

Table~\ref{tab:reduction_levels} summarizes the three HPROMs and their FOM reference.

\begin{table}[htbp]
 \centering\small
 \caption{HPROMs used in the accuracy--cost comparison and their FOM reference. All HPROMs retain microscopic constitutive evaluations.}
 \label{tab:reduction_levels}
 \setlength{\tabcolsep}{4pt}
 \renewcommand{\arraystretch}{1.2}
 \begin{tabularx}{\linewidth}{@{}llcYY@{}}
 \toprule
 Model & State representation & \begin{tabular}[t]{@{}c@{}}Equilibrium\\unknowns\end{tabular} & Integration & Learned component\\
 \midrule
 FOM & $\bm u$ & $N$ & Full mesh & None\\
 \addlinespace[.3em]
 HPROM & $\widetilde{\bm u}=\bm u_{\rm ref}+\bm V_{\rm tra}\qq_{\rm tra}$ & $n_{\rm tra}$ & Fixed ECM (residual and stress) & None\\
 \addlinespace[.3em]
 HPROM--ANN & $\widetilde{\bm u}(\bm\xi)=\bm u_{\rm ref}+\bm V\bm A_m\bm\xi+\overline{\bm V}\widehat{\mathcal N}(\bm\xi)$ & $n$ & Adaptive ECM (residual and stress) & Secondary coordinates; residual and stress weights\\
 \addlinespace[.3em]
 D-HPROM--ANN & $\widetilde{\bm u}_{D}(\ee)=\bm u_{\rm ref}+\bm V\bm A_m\ee+\overline{\bm V}\widehat{\mathcal N}(\ee)$ & None & Adaptive ECM (stress only) & Same closure and stress weights as HPROM--ANN\\
 \bottomrule
 \end{tabularx}
\end{table}
\FloatBarrier

\section{Numerical examples}
\label{sec:cell}
For the numerical assessment of the proposed construction, two square periodic cells are considered, both with total void fraction $0.20$ and the same compressible Neo-Hookean matrix in plane strain.  With $\C_\mu=\F_\mu^T\F_\mu$ and $J_\mu=\det\F_\mu>0$, the local energy density of the matrix is
\begin{equation}
 \psi_\mu=\frac\mu2(\tr\C_\mu-2)-\mu\log J_\mu
              +\frac\lambda2(\log J_\mu)^2,
\end{equation}
with Young's modulus $1.628$ GPa and Poisson ratio $0.4$, where $\mu$ and $\lambda$ are the Lam\'e constants.  Because the two cells share this law, their responses differ only through the arrangement of their cavities, and each cell receives its own independently trained models, so that the assessment does not rest on a single microstructure.

Figure~\ref{fig:common_rve_fields} shows how differently the two cells respond to the same multiaxial deformation, $(E_{11},E_{22},2E_{12})=(18,-4,8)\%$.  In the MC-RVE, four nonaligned cavities create several interacting localization regions, which makes it a demanding setting for determining how many paired features are needed, what learning them adds, and what the constraints secure beyond the data.  In the SC-RVE, a single inclined cavity concentrates the response around one dominant geometric direction.  This cell is the microscopic model propagated through the structural FE$^2$ calculations and shared with the HPROMs of Section~\ref{sec:reduced}.  On it, the learned energies are compared with the HPROMs, first at the constitutive level and then in a tensile coupon, before a confined compression test drives them beyond the data.

\begin{figure}[t]
 \centering\includegraphics[width=.7\linewidth]{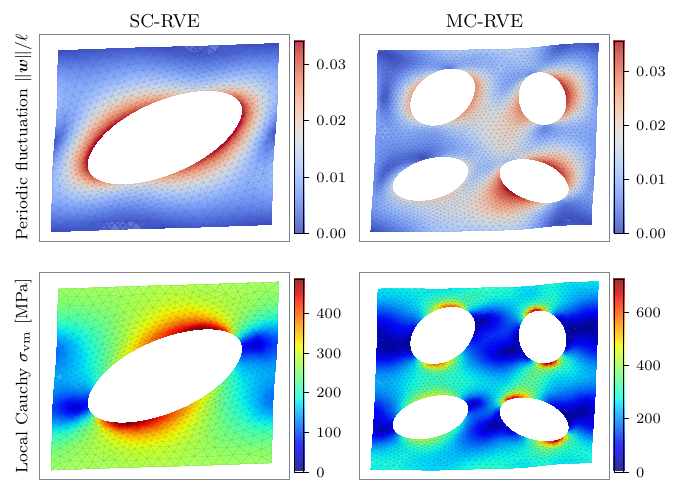}
 \caption{FOM microscopic fields at the common diagnostic state $(E_{11},E_{22},2E_{12})=(18,-4,8)\%$, shown on the actual deformed cells without displacement amplification.  Top: magnitude of the periodic displacement fluctuation $\bm w$ of Eq.~\eqref{eq:micro_deformation}, normalized by the cell side $\ell$.  Bottom: local three-dimensional Cauchy von Mises stress, including the plane-strain out-of-plane stress, with Gauss-point values averaged to the nodes for visualization.  Each panel uses its own color limits to resolve the pattern within its cell, so colors do not compare magnitudes across cells.}
 \label{fig:common_rve_fields}
\end{figure}

\subsection{Training protocol}
\label{sec:training_protocol}
The energies of both benchmarks are trained and selected in the same way, so that differences within each benchmark reflect the constitutive representation rather than the fitting procedure.  Every model is a scalar energy $W(\ee;\bm\theta)$, whose first derivative gives the stress and whose second derivative gives the tangent $\bm D$ of Eq.~\eqref{eq:tangent}.  It is fitted to FOM energies and stresses at the fitting states and to the effective tangent $\bm D_0=\bm D(\bm 0)$ of the periodic cell at the undeformed state.

The ICNN of Eq.~\eqref{eq:monotone_icnn} has two hidden layers of width 24 and the ICKAN of Eq.~\eqref{eq:monotone_ickan} two of width 12, and each ICKAN connection combines six integrated-hat bases on fixed knots from $-1.2$ to $1.2$ ($h=0.48$).  With six fixed paired features, the ICNN has 945 trainable parameters and the ICKAN 1629, and learning the features adds only five parameters per feature.  The initial features are selected from a finite bank of admissible candidates using only the fitting response and the reference tangent.  Before entering either core, each feature is centered at its undeformed value, $z_i^{\rm cen}=z_i^0=1/p_i+1/q_i$, and $\tau_i$ is its largest deviation from that value over the fitting states at initialization.  The conditioned inputs therefore start in $[-1,1]$, which the ICKAN knots cover with a small margin.  The Unconstrained energy uses three softplus hidden layers of widths 128, 128, and 64 on the inputs of Eq.~\eqref{eq:unconstrained_inputs}, for 25\,473 parameters.

All trainable parameters, including the feature parameters of the learned variants, minimize
\begin{equation}
 \mathcal L(\bm\theta)=
 \frac{\sum_k\|\ssv(\ee_k;\bm\theta)-\ssv_k\|_2^2}{\sum_k\|\ssv_k\|_2^2}
 +0.2\,\frac{\sum_k\bigl[W(\ee_k;\bm\theta)-W_k\bigr]^2}{\sum_k W_k^2}
 +0.02\,\frac{\|\bm D(\bm0;\bm\theta)-\bm D_0\|_F^2}{\|\bm D_0\|_F^2},
 \label{eq:training_loss}
\end{equation}
where the sums run over the fitting states $\ee_k$, with FOM energies $W_k$ and stresses $\ssv_k$.  Stress carries the largest weight because it is queried throughout the FE$^2$ calculation.  The energy term keeps the potential itself, and not only its gradient, close to the homogenized energy.  The small tangent term anchors the stiffness at the undeformed state, where the vanishing stresses contribute little to the first term.

Optimization is full-batch and in double precision.  The validation error is the first term of Eq.~\eqref{eq:training_loss} evaluated on the validation states.  Adam, with initial learning rates of $5\times10^{-3}$ for the constrained energies and $4\times10^{-4}$ for the Unconstrained energy, runs until this error stops improving, or for at most 200\,000 steps, and L-BFGS then refines the best Adam state.  The retained checkpoint is the one with the lowest validation error, and each model is trained from three seeds.  Independent test states are evaluated only after every checkpoint has been frozen.
\FloatBarrier

\subsection{Multicavity RVE: learning the paired features}
\label{sec:constitutive_training}
The paired-feature construction guarantees admissibility for every allowed parameter set, but it does not determine how much directional information the convex core needs.  This question is particularly relevant for the MC-RVE, whose nonaligned cavities introduce several geometric directions.  A small fixed feature set may then fail to distinguish deformation states with different directional responses.  One remedy is to supply more features, and another is to let their directions and exponents adapt to the RVE response.  Both effects are examined together by asking how accuracy changes with the number $m$ of paired features and, for each $m$, what is gained by learning rather than fixing their parameters.

\subsubsection{Cell, data, and controlled comparison}
\label{sec:material_b}
The MC-RVE distributes the void space among four unequal, off-center ellipses with different aspect ratios and orientations (Fig.~\ref{fig:rve_b}), and its working mesh has 4621 quadratic triangles.

\begin{figure}[htbp]
 \centering\includegraphics[width=.65\linewidth]{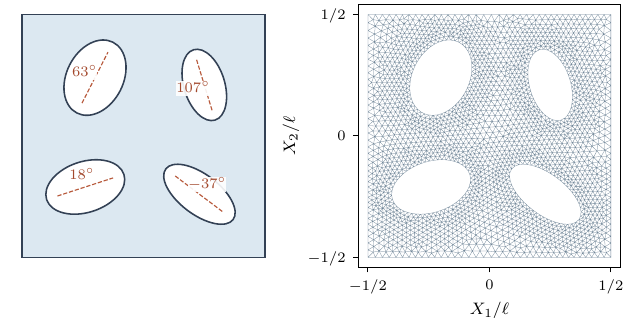}
 \caption{MC-RVE cavity geometry (left) and quadratic mesh (right).  Dashed lines identify cavity major axes.}
 \label{fig:rve_b}
\end{figure}

Refining the mesh to 8961 elements changes the homogenized energy, stress, and tangent by less than $0.1\%$ over 64 audited states that include the corners of the fitting box defined below.  The sampled maximum microscopic first Piola--Kirchhoff stress norm changes by at most $4.06\%$, which supports the working mesh for the representation study.

The fitting box spans $E_{11},E_{22}\in[-0.04,0.20]$ and $2E_{12}\in[-0.08,0.08]$ (Fig.~\ref{fig:constitutive_domains}, left).  The 4200 fitting states combine 4096 scrambled Sobol points in the interior with 96 points on the faces and the eight corners, so that the boundary of the box is sampled rather than reached by extrapolation.  Independent Sobol sequences supply 512 validation and 512 independent-test states.  For each $m\in\{2,4,6,8,16,24,32\}$, ICNN and ICKAN are trained in two variants that start from the same $m$ initial features.  The fixed variant holds their directions and exponents constant while training the convex core, whereas the learned variant optimizes those parameters subject to the constraints of Section~\ref{sec:learned}.  All other ingredients of Section~\ref{sec:training_protocol} are shared within each core, so the difference between the two variants isolates the effect of adapting the features.

\begin{figure}[htbp]
 \centering\includegraphics[width=.8\linewidth]{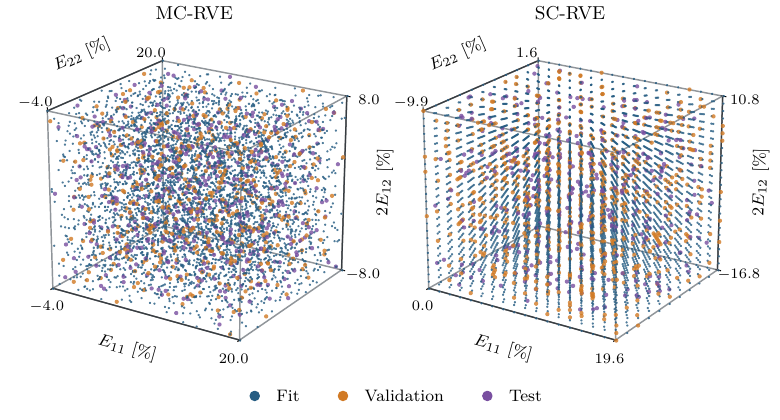}
 \caption{Constitutive sampling domains, in percent.  Left: MC-RVE representation benchmark.  Right: SC-RVE deployment benchmark, with fitting and validation states on a Cartesian grid and independent-test states off it.  Wireframes outline the fitting boxes, and colors identify fitting, validation, and independent-test states.}
 \label{fig:constitutive_domains}
\end{figure}

\subsubsection{Number of paired features}
Figure~\ref{fig:feature_count_sensitivity} shows the validation error of both convex cores as a function of $m$.  Two paired features leave the response under-resolved, and both variants improve as $m$ increases, but in different ways.  The learned variants gain most up to $m=6$ and change by less than a factor of two beyond it, whereas the fixed variants continue to improve gradually up to $m=32$.  The independent evaluation uses $m=6$, the first count at which both learned cores enter their low-error regime.

\begin{figure}[htbp]
 \centering\includegraphics[width=.98\linewidth]{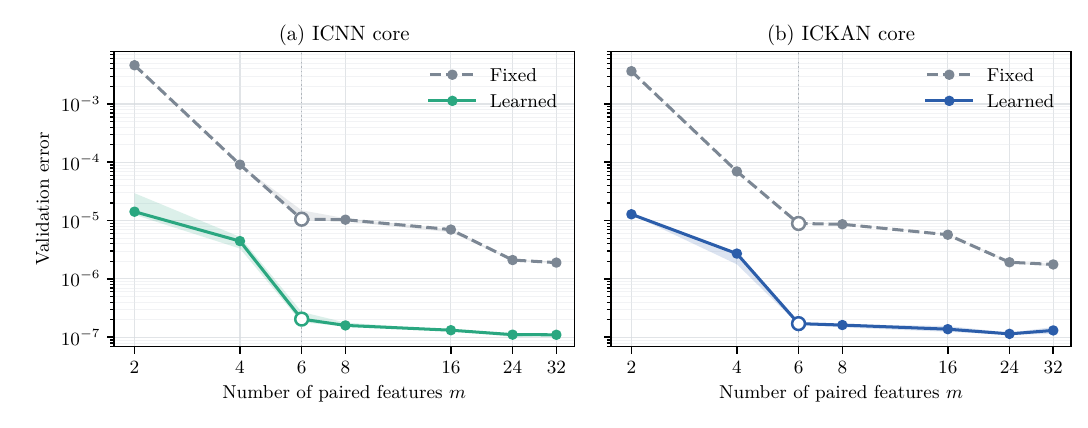}
 \caption{MC-RVE validation error, the relative squared stress error of Section~\ref{sec:training_protocol}, versus the number $m$ of paired directional features.  Lines show the median over three seeds and bands span their minimum--maximum range.  The vertical guide marks $m=6$, used for the independent comparison.}
 \label{fig:feature_count_sensitivity}
\end{figure}

\subsubsection{Effect of learning the features}
Table~\ref{tab:mc_rve_m06} compares the fixed and learned variants on the independent test states.  At $m=6$, learning the paired features lowers the ICNN median stress error from $0.328\%$ to $0.0445\%$ and the ICKAN error from $0.301\%$ to $0.0406\%$, a reduction by a factor of approximately $7.4$ in either core.  The energy and the tangent, which is never fitted away from the undeformed state, improve by factors of about 14 and 5, respectively.  Adding fixed directions does not close this gap.  In the validation sweep, even 32 fixed features, with more than twice the trainable parameters of the learned six-feature models, remain about three times less accurate in relative stress error, a factor of about ten in the squared validation error of Fig.~\ref{fig:feature_count_sensitivity}.  The similar gain for ICNN and ICKAN shows that the central improvement comes from adapting the directional representation rather than from one particular convex core.

The Unconstrained energy places these values in context.  Without convexity or growth conditions, it reaches $0.0073\%$ on the same states, about six times below either learned constrained model.  Its network is also considerably larger.  At matched capacity, the controls of Appendix~\ref{app:capacity} show that, inside the fitting box, the constraints cost a factor of 3.2 to 5.5 in stress error on this cell.

\begin{table}[htbp]
 \centering\small
 \caption{MC-RVE independent-test errors at $m=6$.  Each entry is the relative error norm of the stress $\ssv$, the energy $W$, or the tangent $\bm D$ of Eq.~\eqref{eq:tangent} over the 512 test states, as the median over three seeds.  For every quantity, each learned seed is more accurate than each fixed seed.  Adam stopped at a validation plateau for the fixed variants, whereas five of the six learned runs and the Unconstrained energy, which has no paired features, reached the 200\,000-step cap.}
 \label{tab:mc_rve_m06}
 \begin{tabular}{lccc}
\toprule
Model & Stress [\%] & Energy [\%] & Tangent [\%]\\
\midrule
ICNN, fixed & 0.3282 & 0.1096 & 1.9247\\
ICNN, learned & 0.0445 & 0.0079 & 0.3851\\
\addlinespace
ICKAN, fixed & 0.3014 & 0.0943 & 1.7600\\
ICKAN, learned & 0.0406 & 0.0066 & 0.3622\\
\midrule
Unconstrained energy & 0.0073 & 0.0011 & 0.0857\\
\bottomrule
\end{tabular}

\end{table}

\subsubsection{Admissibility beyond the data}
\label{sec:mc_guarantees}
For the ICNN and ICKAN energies, conditions (C1)--(C6) hold for every admissible parameter set by construction, whereas the Unconstrained energy is built to satisfy (C1)--(C4) only.  Inside the box the difference does not show, since a bounded rank-one search over 520 states and 32 directions finds curvature above 236~MPa for all fifteen fits of Table~\ref{tab:mc_rve_m06}.  Farther out the two classes separate.  Figure~\ref{fig:mc_beyond_data}(a) follows equibiaxial compression, $\F=\sqrt J\,\I$, beyond the box.  The minimum rank-one curvature of every Unconstrained fit changes sign between $J=0.71$ and $0.73$, whereas the learned constrained fits remain above 220~MPa.  A directed search over principal stretches between 0.55 and 1.5, in arbitrary orientations, shows that this loss is not tied to one path.  It finds negative curvature for all three Unconstrained fits at more than a third of its 4096 sampled states, confirmed by energy finite differences, and only positive curvature for the twelve constrained fits, as their construction requires.

Figure~\ref{fig:mc_beyond_data}(b) continues the same path toward volume collapse.  The constrained energies diverge, as their logarithmic barrier guarantees, whereas the Unconstrained energy levels off at finite limits between 397 and 440~MPa.  This contrast reflects the growth terms rather than polyconvexity, and an unconstrained network could receive the same barrier.  No data constrain either energy at these states, so only the construction can guarantee (C5) and (C6) there.

\begin{figure}[htbp]
 \centering\includegraphics[width=.98\linewidth]{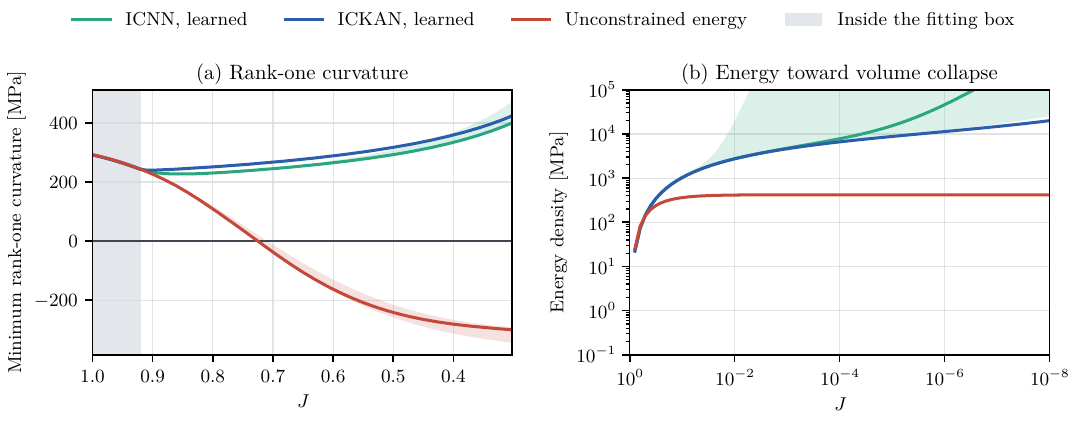}
 \caption{MC-RVE energies beyond the data along equibiaxial compression, $\F=\sqrt J\,\I$, for the learned ICNN and ICKAN energies and the Unconstrained energy.  Lines show the median of three seeds and shaded bands their range.  (a) Minimum rank-one curvature over unit directions.  (b) Energy density toward volume collapse.}
 \label{fig:mc_beyond_data}
\end{figure}

Global nonnegative energy, which the construction leaves to the trained parameters, is certified by the sufficient bound of Appendix~\ref{app:nonnegative} for all six learned fits.  For all fifteen fits, the reference energy and stress vanish to round-off, and centered differences reproduce the stress and the tangent to better than $10^{-9}$ of their largest entries.

Taken together, the MC-RVE supports the choices carried into the SC-RVE study.  Six paired features give a compact representation, and adapting their directional parameters supplies the main accuracy gain for either convex core.  The learned constrained energies then keep the median stress error below $0.05\%$ while securing, beyond the data, the rank-one convexity and growth that the Unconstrained energy lacks there.
\FloatBarrier

\subsection{Single-cavity RVE: constitutive accuracy and FE$^2$ deployment}
\subsubsection{Cell, data, and learned energies}
\label{sec:material_a}
The second benchmark, the SC-RVE, contains one centered elliptical void with aspect ratio $2$ and major-axis angle $30^\circ$ relative to the cell axes (Fig.~\ref{fig:rve}).  Its working mesh has 1546 quadratic triangles and 6320 independent microscopic displacement coordinates.

\begin{figure}[htbp]
 \centering\includegraphics[width=.65\linewidth]{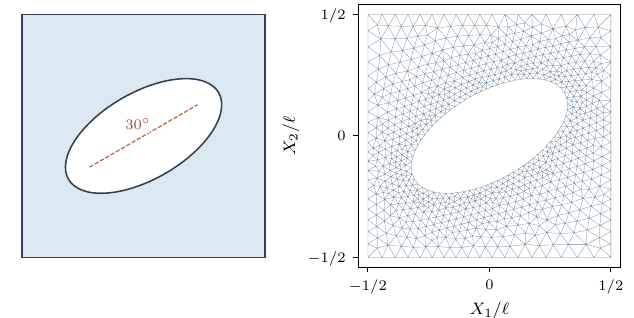}
 \caption{SC-RVE cavity geometry (left) and quadratic mesh (right).  The dashed line identifies the cavity major axis.}
 \label{fig:rve}
\end{figure}

At a uniaxial-stress state with $E_{11}=0.20$, refining the mesh to 3095 elements changes the maximum microscopic first Piola--Kirchhoff stress norm by $0.486\%$ and the average axial stress by less than $0.001\%$.  This state lies on the path followed by the narrow section of the coupon in Section~\ref{sec:coupon}, beyond the largest axial strain reached there, so the check supports the mesh used for deployment.

Unlike the MC-RVE, whose box is prescribed, the SC-RVE box must cover the strains that the macroscopic quadrature points visit during the coupon FE$^2$ calculation of Section~\ref{sec:coupon}.  A displacement-controlled coupon prepass with a St Venant--Kirchhoff approximation based on the cell reference tangent first recorded the strain envelope.  This envelope guides only the sampling, whereas RVE stresses provide the training labels.  A second prepass with a softer response shifted the componentwise extrema by at most $5.1\%$, so the spans were enlarged by $40\%$.  Because the coupon is tensile, the lower $E_{11}$ bound is zero.  In the engineering-strain coordinates $\ee=(E_{11},E_{22},2E_{12})^T$, the rounded bounds are
\begin{equation}
 \ee_{\min}=(0,-0.099,-0.168)^T,\qquad
 \ee_{\max}=(0.196,0.016,0.108)^T.
 \label{eq:box}
\end{equation}
A Cartesian grid supplies 4950 nonreference states, of which 4208 are used for fitting and 742 for validation, and an independent set of 400 off-grid states tests interpolation (Fig.~\ref{fig:constitutive_domains}, right).

Whether the count $m=6$ selected on the MC-RVE also suffices for the SC-RVE deployment is tested without transferring feature parameters between the two geometries.  Six initial features are selected anew from the SC-RVE fitting response and reference tangent, and the learned ICNN and ICKAN are trained as described in Section~\ref{sec:training_protocol}, as is the Unconstrained energy.  For each model, the seed with the lowest validation error supplies the deployed model.  The validation errors of the three seeds of each constrained energy lie within $15\%$ of one another.

\subsubsection{Construction of the HPROMs}
\label{sec:material_a_reduction}
To build the HPROMs, the FOM solutions at the 4950 grid states of Section~\ref{sec:material_a}, the same states that label the learned energies, are first collected as snapshots. The snapshot matrix of Eq.~\eqref{eq:snapshot_svd} is assembled from their displacement fields, and its SVD truncated at $\epsilon_{\rm POD}=10^{-6}$ in Eq.~\eqref{eq:pod_tail} retains $n_{\rm tra}=39$ modes, which form the traditional ROB $\bm V_{\rm tra}$ (Fig.~\ref{fig:pod}).

\begin{figure}[htbp]
 \centering\includegraphics[width=.6\linewidth]{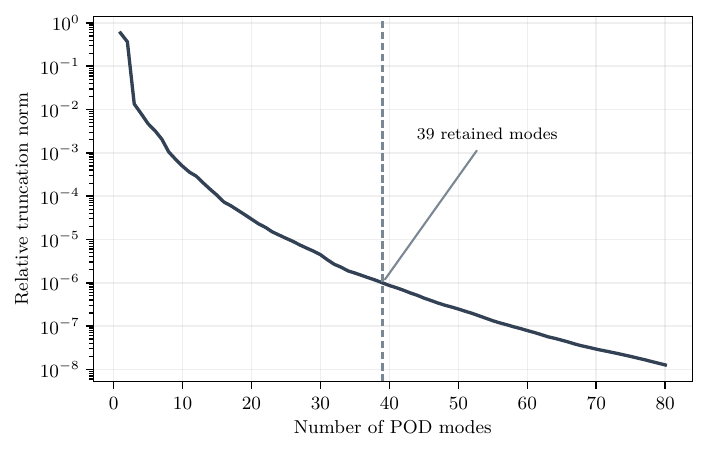}
 \caption{Relative truncation norm $\epsilon_{\rm POD}$ of Eq.~\eqref{eq:pod_tail} for the 4950 SC-RVE displacement snapshots versus the number of modes. The 39 retained modes meet the tolerance of $10^{-6}$, and their rotation into primary and secondary coordinates leaves this span unchanged.}
 \label{fig:pod}
\end{figure}

The affine HPROM uses this ROB directly. Specifically, its ECM training~\cite{hernandez2017} takes every tenth snapshot, $N_c=495$, and is executed separately for the residual and the stress. The element contributions are compressed by SVD with tolerances $\varepsilon_{\rm SVD}=10^{-3}$ for the residual and $10^{-5}$ for the stress, and the sparse fit of Eq.~\eqref{eq:ecm_tolerance} uses $\varepsilon_{\rm ECM}=10^{-6}$ for both. This results in a reduced mesh of 183 elements, $11.8\%$ of the FOM mesh, of which 135 support the residual rule and 73 the stress rule.

For the nonlinear HPROMs, the same 39-dimensional span is rotated into $n=3$ primary coordinates $\bm\xi$, approximately aligned with the macroscopic strain~\cite{hernandez2026dimensional}, and $\bar n=36$ secondary ones. A closure predicts the secondary coordinates from the primary ones, so the iterative HPROM--ANN solves equilibrium for only three unknowns. Following Section~\ref{sec:latent_closure}, the closure $\widehat{\mathcal N}$ is trained on the snapshot pairs $(\bm\xi^s,\overline{\qq}^{\,s})$ with Adam and the loss of Eq.~\eqref{eq:closure_loss}. In this case, the network has input and output sizes $n=3$ and $\bar n=36$ and two hidden layers of 64 neurons with $\tanh$ activations, and it standardizes its inputs internally while leaving its outputs unscaled.

In contrast with the affine case, HPROM--ANN uses the MAW--ECM of Hern\'andez et al.~\cite{hernandez2026dimensional}. Fixed ECM rules, obtained with $\varepsilon_{\rm SVD}=10^{-4}$ and $\varepsilon_{\rm ECM}=10^{-6}$, are pruned to ten elements each, resulting in a considerably smaller reduced mesh of 19 elements, $1.2\%$ of the FOM mesh. On this support, the weights are functions of $\bm\xi$, given by networks with three hidden layers of 128 neurons, GELU activations, and the scaled softmax output of Eq.~\eqref{eq:softmax_weights}. They are fitted after the closure and the pruning, with internally standardized inputs. These networks are larger than the closure, so the online cost does not follow from the three unknowns or the element count alone.

Finally, D-HPROM--ANN uses the same closure and stress rule but prescribes $\bm\xi=\ee$ instead of solving reduced equilibrium, as in the hyperelastic example of Hern\'andez et al.~\cite{hernandez2026dimensional}, so only the ten elements of the stress rule, $0.6\%$ of the FOM mesh, are evaluated. The substitution is accurate for this cell, since the primary coordinates of the snapshots differ from their macroscopic strains by less than $10^{-4}$ in relative terms, although it remains an additional approximation. All these element counts refer to reduced meshes on which only the selected elements are assembled (Fig.~\ref{fig:supports}). The three HPROMs are compared with the learned energies first on the independent test states of Section~\ref{sec:constitutive_accuracy} and then in the coupon FE$^2$ calculation of Section~\ref{sec:coupon}.

\begin{figure}[htbp]
 \centering\includegraphics[width=.99\linewidth]{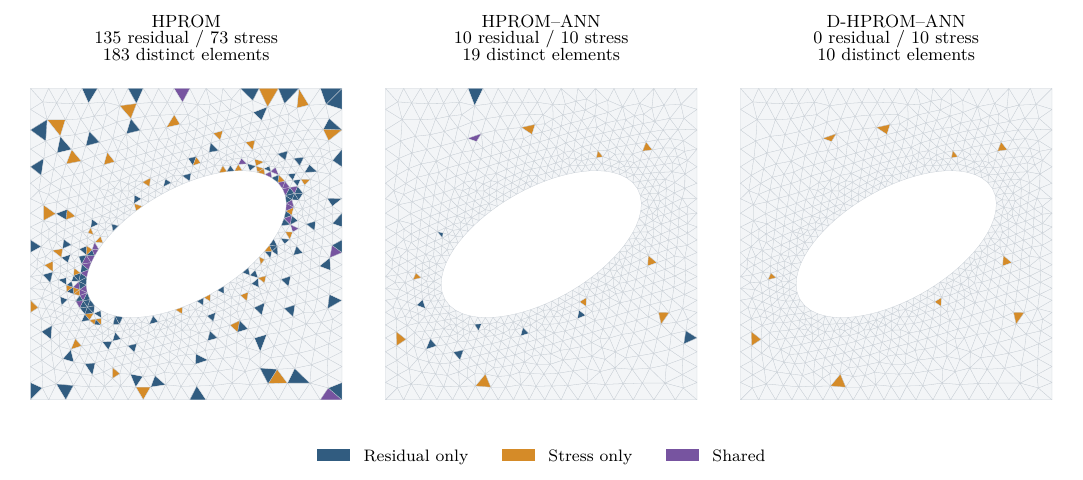}
 \caption{Element supports of the three HPROMs. The pale full-cell mesh provides location context only. Colors distinguish residual-only, stress-only, and shared elements.}
 \label{fig:supports}
\end{figure}

\subsubsection{Constitutive accuracy}
\label{sec:constitutive_accuracy}
On the 400 off-grid states, the compact six-feature ICNN and ICKAN keep the aggregate stress error below one percent, at $0.621\%$ and $0.636\%$, whereas the Unconstrained energy reaches $0.0052\%$, about 120 times lower (Table~\ref{tab:constitutive}).\footnote{Outside the fitting box, up to twice its size, the Unconstrained energy also remains more accurate than the constrained energies.}  On this cell the construction costs far more accuracy than on the MC-RVE, and capacity is not the reason.  In the controls of Appendix~\ref{app:capacity}, cores widened to about 25\,000 parameters change the constrained errors by less than $3\%$, whereas an Unconstrained energy with only 981 parameters remains about 75 to 80 times more accurate than either wide core.  Thirty-two learned features lower the errors only to $0.537\%$ and $0.618\%$.  This behavior is consistent with the ordering argument of Section~\ref{sec:feature_order}, since a monotone core cannot reverse the componentwise ordering of its features, whatever its width.  Because the 32-feature models gain little and lose the nonnegative-energy certificate of Appendix~\ref{app:nonnegative}, the coupon keeps the $m=6$ models.

The price of this restriction depends on the cell, a factor of 3.2 to 5.5 on the MC-RVE and of 76 to 126 on the SC-RVE at matched capacity.  Comparable costs are reported for other constrained constructions.  Linden et al.~\cite{linden2023} found that imposing all common conditions on a small invariant-based network raised its stress error by about one order of magnitude for transverse isotropy, and still recommended the constrained model for its extrapolation.  On homogenization data of a matrix with a stiff spherical inclusion, the polyconvex models of Klein et al.~\cite{klein2026} reached root-mean-square stress errors three to five times those of their non-polyconvex counterparts, and Kalina et al.~\cite{Kalina2024} dropped polyconvexity after finding it too restrictive for their fits.  The MC-RVE factors are of this size, whereas the SC-RVE shows that the price can be much larger for some microstructures.

As physics-based references, the HPROMs attain their accuracy without learning the stress response, since they evaluate the constituent law on their reduced meshes. The affine HPROM, which contains no trained network, reaches $0.0169\%$ with $11.8\%$ of the elements, within a factor of about three of the Unconstrained energy, and D-HPROM--ANN and HPROM--ANN reach $0.0959\%$ and $0.218\%$ with $0.6\%$ and $1.2\%$ of them.

Beyond the data, the comparison of Section~\ref{sec:mc_guarantees} carries over to this cell.  All three energies keep a rank-one curvature above 180~MPa on the test and 12\,000 further in-box states.  A radial expansion of the box to four times its size finds negative curvature for the Unconstrained energy at about 13\% of its states, first at 1.55 times the box, whereas the ICNN and ICKAN remain above 176~MPa.  For the two deployed models, the sufficient bound of Appendix~\ref{app:nonnegative} certifies nonnegative energy, and derivative checks and the vanishing work of Eq.~\eqref{eq:conservative_requirement} around a closed strain cycle confirm the potential structure of all three energies.

\begin{table}[htbp]
 \centering\small
 \caption{SC-RVE independent-test constitutive errors of the HPROMs and the selected energies, in percent.}
 \label{tab:constitutive}
 \begin{tabular}{lrr}
\toprule
Model & $e_S^{\rm test}$ [\%] & $e_W^{\rm test}$ [\%] \\
\midrule
HPROM & 0.0169 & -- \\
HPROM--ANN & 0.2177 & -- \\
D-HPROM--ANN & 0.0959 & -- \\
ICNN & 0.6214 & 0.1929 \\
ICKAN & 0.6359 & 0.1900 \\
Unconstrained energy & 0.0052 & 0.0010 \\
\bottomrule
\end{tabular}

\par\smallskip{\footnotesize All values are percentages. Stress and energy errors are aggregate array norms. SC-RVE: 400 independent test states. The HPROMs return no energy.}

\end{table}
\FloatBarrier

\subsubsection{Coupon FE$^2$: accuracy and online cost}
\label{sec:coupon}
The coupon has an ASTM-inspired dog-bone profile, analyzed in plane strain, with overall length 165 mm, grip width 19 mm, narrow-section width 13 mm, narrow-section length 57 mm, and fillet radius 76 mm.

\begin{figure}[htb]
 \centering\includegraphics[width=.99\linewidth]{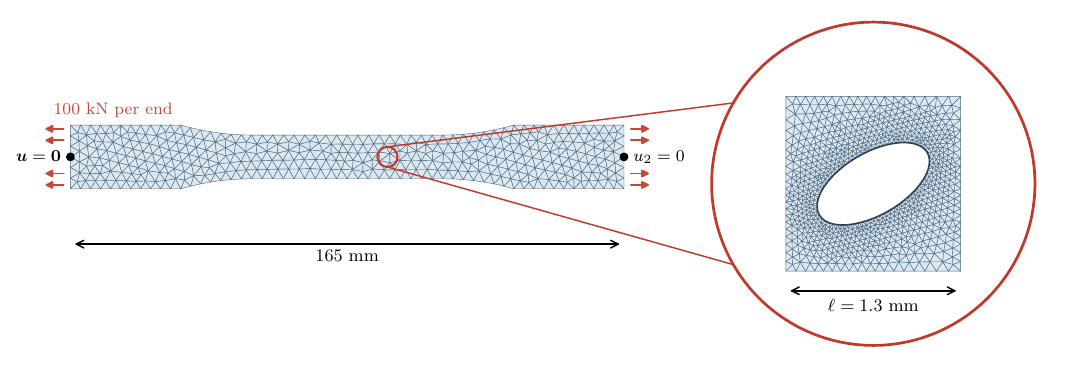}
 \caption{Coupon problem.  Undeformed macro mesh with equal-and-opposite resultants of 100~kN per end, which balance each other, so the three point constraints only remove rigid-body motion.  Every integration point queries the SC-RVE, solved by the FOM or replaced by an HPROM or a learned energy.}
 \label{fig:coupon_setup}
\end{figure}

The two end faces receive equal-and-opposite axial resultants of 100~kN per end, applied in 20 equal increments (Fig.~\ref{fig:coupon_setup}). Three point constraints remove the rigid modes, fixing the axial and transverse displacements at the left mid-edge node and the transverse displacement at the right one, so the end faces are not clamped and the anisotropic response may couple axial extension to shear.

The macro mesh contains 690 quadratic triangles, 1501 nodes, and 2070 integration points. A thickness of 50 mm converts the two-dimensional stress integrals into the 100~kN resultants, and a cell size of 1.3 mm, ten cells across the narrow width, is assumed for scale separation. The converged integration-point strains of every model remain inside the fitting box. Figure~\ref{fig:coupon} shows the reference fields at the final load. The equivalent stress is nearly uniform in the narrow section, peaks slightly where the fillets begin, and decreases toward the grips.

\begin{figure}[htb]
 \centering\includegraphics[width=.99\linewidth]{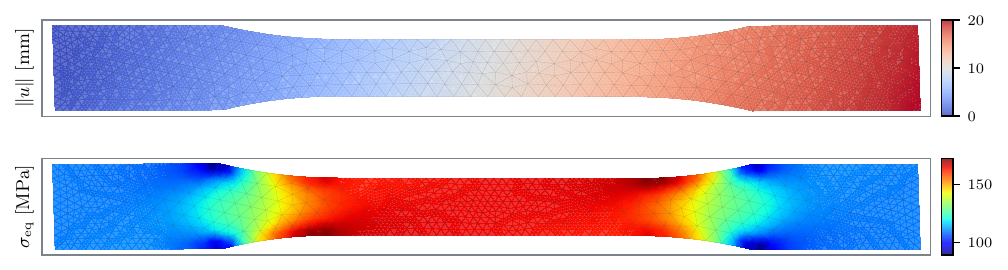}
 \caption{Nested FOM--FE$^2$ reference at 100~kN per end, shown on the deformed mesh without displacement amplification.  Top: displacement magnitude.  Bottom: in-plane von Mises equivalent of the macroscopic Cauchy stress, $\sigma_{\rm eq}=(\sigma_{11}^2-\sigma_{11}\sigma_{22}+\sigma_{22}^2+3\sigma_{12}^2)^{1/2}$. Gauss-point values are averaged to the nodes for visualization.}
 \label{fig:coupon}
\end{figure}

For a nodal or integration-point array $A$, the reported relative error is
\begin{equation}
 e_A=100\frac{\|\operatorname{vec}(A-A_0)\|_2}{\|\operatorname{vec}(A_0)\|_2},
 \label{eq:error}
\end{equation}
where $A_0$ is the nested reference on the same meshes. These unweighted array norms, with engineering shear in the strain array and tensor shear in the stress array, are the same for every model.

Online wall time covers the complete macroscopic solve, including every microscopic or surrogate evaluation, and excludes offline data generation, training, and reduced-mesh construction. The FOM and the HPROMs distribute the integration points over 20 processes, and the learned energies are evaluated in one batched process. Reduced-model times are medians of three runs, whereas the FOM and learned-energy times come from single runs.

The macroscopic solve stops when the free-degree-of-freedom residual norm falls below $10^{-5}$ N or its ratio to the initial increment residual falls below $10^{-7}$. The periodic FOM uses a relative microscopic displacement-increment tolerance of $10^{-10}$.
The affine HPROM uses its exact reduced residual Jacobian and HPROM--ANN the modified-Newton matrix of Section~\ref{sec:adaptive_weights}, both with a microscopic increment tolerance of $10^{-10}$, so the speed-up between them also reflects their different iteration matrices.

\begin{table}[htbp]
 \centering\small
 \caption{Coupon simulations: final-field errors of Eq.~\eqref{eq:error} with respect to the nested FOM, online wall time, and speed-up over the FOM.}
 \label{tab:fe2}
 \begin{tabular}{lrrrrr}
\toprule
Model & $e_u$ [\%] & $e_E$ [\%] & $e_S$ [\%] & Time [s] & Speed-up \\
\midrule
FOM & -- & -- & -- & 4268.68 & -- \\
HPROM & 0.0057 & 0.0165 & 0.0010 & 67.35 & 63.38 \\
HPROM--ANN & 0.0520 & 0.1497 & 0.0373 & 30.14 & 141.61 \\
D-HPROM--ANN & 0.0232 & 0.0485 & 0.0097 & 3.79 & 1126.90 \\
ICNN & 0.4066 & 0.6785 & 0.0739 & 1.04 & 4105.14 \\
ICKAN & 0.4197 & 0.6914 & 0.0732 & 1.95 & 2187.66 \\
Unconstrained energy & 0.0022 & 0.0050 & 0.0006 & 1.48 & 2884.31 \\
\bottomrule
\end{tabular}

\end{table}

Table~\ref{tab:fe2} shows distinct operating points rather than a single universally best model. The three HPROMs reproduce the nested reference to within $0.06\%$ in displacement, with speed-ups from 63 to 1127. Unlike the learned energies, the HPROMs can also reconstruct an approximate microscopic displacement field at every integration point. Among the HPROMs, the affine HPROM gives the smallest displacement error, $0.0057\%$, at a median online time of 67.35 s. The iterative HPROM--ANN reduces this time to 30.14 s, 2.23 times less, with a displacement error of $0.0520\%$. The direct D-HPROM--ANN reduces it to 3.79 s and has a displacement error of $0.0232\%$ on this path. Its lower error than the iterative model is possible because solving projected equilibrium corrects the primary state only within the chosen manifold and cubature approximation, so it need not give a smaller error than direct evaluation.

The six-feature ICNN and ICKAN take 1.04 s and 1.95 s, with displacement errors of $0.407\%$ and $0.420\%$, below their constitutive test errors but above those of the HPROMs on this benchmark. The Unconstrained energy takes 1.48 s and gives the smallest displacement error in Table~\ref{tab:fe2}, $0.0022\%$. In displacement error, the six models rank exactly as in Table~\ref{tab:constitutive}. Every coupon state lies inside the box, where all three energies have positive rank-one curvature at every sampled state, so this tensile path does not exercise the guarantees that separate them. On this path the constrained energies keep the displacement error near $0.4\%$ and add admissibility beyond the data, at a cost of about two orders of magnitude in structural accuracy against the Unconstrained energy.

All seven models, the nested FOM included, require four macroscopic iterations in every increment, so the tangents of Eq.~\eqref{eq:reduced_ift} and those of the learned energies preserve the convergence of the reference. Inside the box, the constrained energies thus run in the common structural solver like the Unconstrained energy, at comparable cost. Their guarantees matter at states that the coupon never reaches.

\subsubsection{Confined compression beyond the data}
\label{sec:confined}
As an illustrative case built to reach such states, consider the plane-strain counterpart of an oedometer test, a loading for which the tensile fitting box was never designed. Its only purpose is to expose what the guarantees secure once compaction carries the integration points beyond the data.  A square block of side 10~mm, meshed with 128 quadratic triangles, is compressed by a uniform traction on one face while rigid, frictionless walls fix the normal displacement of the other three (Fig.~\ref{fig:confined}a).  The three SC-RVE energies, the assembly, and the Newton iteration are those of the coupon, and the load is applied in 20 equal increments up to 500~MPa, in $x$ or in $y$.  Compression in $x$ leaves the fitting box with the first increment, since the box contains no axial compression, and compression in $y$ leaves it beyond about $10\%$ shortening.  At every converged increment, the rank-one curvature of Eq.~\eqref{eq:lh_background}, computed from the constitutive tangent and the stress through Eq.~\eqref{eq:strain_vs_F_curvature}, is minimized over unit directions at the 384 Gauss points of the block.  Where it is positive, the incremental equilibrium problem is strongly elliptic.

\begin{figure}[t]
 \centering\includegraphics[width=.86\linewidth]{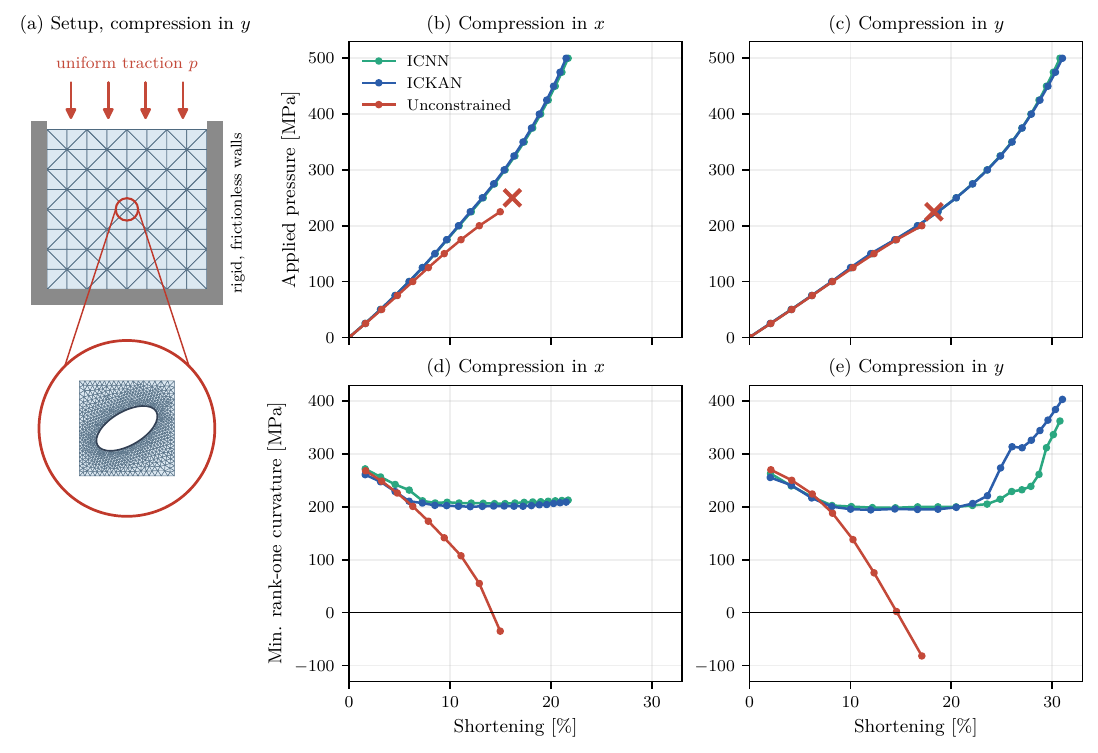}
 \caption{Confined compression beyond the data with the SC-RVE energies.  (a) Setup for compression in $y$, with the homogenized block of 128 quadratic triangles evaluating the SC-RVE energy at every integration point. For compression in $x$, the traction acts on a vertical face.  (b, c) Applied pressure versus shortening, with one marker per increment of 25~MPa and a cross at the increment in which Newton does not converge.  (d, e) Minimum rank-one curvature over the 384 Gauss points of the block at each converged increment.}
 \label{fig:confined}
\end{figure}

Both constrained energies reach 500~MPa in both directions, with four to nine Newton iterations per increment and final shortenings of about $22\%$ in $x$ and $31\%$ in $y$ (Fig.~\ref{fig:confined}b,c).  The Unconstrained energy follows a softer curve in $x$ and nearly the same curve in $y$, until Newton does not converge within 50 iterations in the increment to 250~MPa in $x$ and to 225~MPa in $y$.  This stop coincides with a loss of rank-one convexity in the block (Fig.~\ref{fig:confined}d,e).  The minimum curvature of the Unconstrained energy falls steadily with the load and turns negative at the last converged increment, at one of the 384 Gauss points in $x$ and at two in $y$.  These points lie in the elements at the corner where the left and bottom walls meet, where the local volume ratio has fallen to $J$ between 0.61 and 0.72.  The constrained energies, rank-one convex by construction, keep the minimum curvature above 194~MPa at every Gauss point and increment.  The test shows which energies keep the structural problem solvable once it leaves the data.

\section{Conclusions}
\label{sec:conclusions}
This work develops polyconvex energies for anisotropic periodic cells by composing paired directional features, with trainable orientations, stretch powers, and determinant exponents, with convex and componentwise nondecreasing cores.  Objectivity, an exactly stress-free reference, two-dimensional polyconvexity, and growth then hold by construction for every parameter set, and an affine determinant correction normalizes the stress without tying the normalization coefficient to the volumetric barrier stiffness.  Admissibility thus becomes a property of the architecture rather than of a successful fit.  Only nonnegative energy is left to the trained parameters, and a sufficient bound certifies it for every six-feature model of this study.

The numerical examples test this construction on two cells with complementary roles.  On the MC-RVE, learning six paired features lowers the held-out stress error sevenfold, to about $0.04\%$, whereas adding fixed directions does not close the gap.  On the SC-RVE, the selected ICNN and ICKAN keep the stress error near $0.6\%$ and the coupon displacement error near $0.4\%$.  The three HPROMs, which retain an approximate microscopic state, reproduce the test stresses within about $0.2\%$ and the nested FOM--FE$^2$ displacements within $0.06\%$, at speed-ups of about 60 to 1100.  Under confined compression beyond the data, the constrained energies carry the block to 500~MPa, whereas with the Unconstrained energy Newton stops below half that load, when its rank-one curvature turns negative.

Beyond these figures, the examples suggest how such energies should be built.  The anisotropy of both cells is captured better by adapting a few directions than by prescribing many, and six learned features already suffice on both.  This count need not generalize to other microstructures, but it is encouraging, because it keeps the energies compact and certifiable.  The accuracy that admissibility costs depends on the microstructure rather than on the network, from about six times the Unconstrained error on the MC-RVE to about 120 times on the SC-RVE.  Neither wider cores nor more features remove it, which is consistent with the ordering restriction of monotone cores and identifies where the construction can still improve.

The comparison also clarifies when the guarantees matter.  Inside the data, accuracy decides, and the Unconstrained energy is the most accurate of the six models.  Beyond the data, admissibility decides, and the constrained energies keep the structural problem solvable where the Unconstrained energy does not.  The HPROMs and the learned energies emerge as complementary tools, since the HPROMs retain the microscopic state and the constituent law at a higher online cost, whereas the learned energies reduce each query to a direct evaluation and, when constrained, carry guarantees that the HPROMs do not provide by construction.  Because the six models rank identically at the constitutive level and in the coupon displacement error, and all preserve the Newton convergence of the nested reference, constitutive tests offer an inexpensive guide to their structural performance on this benchmark.

In summary, the accuracy given up by the constrained energies buys a structural problem that remains solvable beyond the data, which is the property the construction is designed to secure.  Future work will seek to lower this price on cells such as the SC-RVE, for instance by passing the determinant to the core as an input that requires convexity but not monotonicity, and will extend the construction to three-dimensional cells.  The HPROMs could also generate additional training states at a fraction of the FOM cost.

\section*{Acknowledgements}
Joaqu\'in A. Hern\'andez acknowledges the support of Grant PID2024-158878OB-C21 funded by MICIU/AEI/10.13039/501100011033 and by ERDF/EU, which also supported S. Ares de Parga while this work was carried out at CIMNE. S. Ares de Parga also acknowledges partial support from the Department of Aeronautics and Astronautics at Stanford University. Alejandro Cornejo is a Serra H\'unter fellow.

\section*{Data availability}
The code, the trained models, and the data that support the findings of this study will be made openly available in a public repository upon publication.

\FloatBarrier
\appendix
\section{Rank-one paths and finite-strain energy curvature}
\label{app:finite_strain_curvature}
For a twice differentiable energy, nonnegative curvature along every rank-one perturbation is expressed by the Legendre--Hadamard condition
\begin{equation}
 \partial^2_{\F\F}W(\F)[\bm a\otimes\bm b,\bm a\otimes\bm b]\geq0
 \quad\hbox{for all }\bm a,\bm b\in\R^2 .
 \label{eq:lh_background}
\end{equation}
The polyconvex representation in Eq.~\eqref{eq:polyconvex_definition} implies this condition because the determinant is affine along rank-one paths. For $2\times2$ matrices,
\begin{equation}
 \det(\F+t\bm H)=\det\F+t\,\cof\F:\bm H+t^2\det\bm H,
 \qquad
 \cof\F=\begin{pmatrix}F_{22}&-F_{21}\\-F_{12}&F_{11}\end{pmatrix},
 \label{eq:rankone_minor_line}
\end{equation}
and $\det(\bm a\otimes\bm b)=0$. On a segment $\F_s=(1-s)\F_A+s\F_B$ with $\F_B-\F_A=\bm a\otimes\bm b$, the pair $(\F_s,\det\F_s)$ is therefore the corresponding convex combination of its endpoint values, and it stays admissible when both endpoint determinants are positive. Convexity of $\mathcal P$ then gives $W(\F_s)\leq(1-s)W(\F_A)+sW(\F_B)$, so a twice differentiable $W$ has nonnegative second derivative along the segment, which is Eq.~\eqref{eq:lh_background}~\cite{ball1977,Ciarlet1988}.

For an energy expressed in strain, the curvature along $\F+t\bm H$ follows from the chain rule. Let $\bm H\in\R^{2\times2}$ be an arbitrary perturbation and let $t$ parametrize $\F(t)=\F+t\bm H$. Dots denote derivatives with respect to $t$, and $\operatorname{sym}(\bm A)=(\bm A+\bm A^T)/2$ denotes the symmetric part of a matrix. Direct multiplication gives
\begin{align}
 \E(t)
 &=\tfrac12\bigl[\F^T\F-\I
 +t(\F^T\bm H+\bm H^T\F)+t^2\bm H^T\bm H\bigr]\nonumber\\
 &=\E(0)+t\operatorname{sym}(\F^T\bm H)
 +\tfrac12t^2\bm H^T\bm H.
\end{align}
It follows that $\dot\E(0)=\operatorname{sym}(\F^T\bm H)$ and $\ddot\E(0)=\bm H^T\bm H$. Using the engineering-strain convention of Eq.~\eqref{eq:voigt}, the chain and product rules give
\begin{align}
 \frac{\dd W}{\dd t}&=\ssv\cdot\dot\ee,\nonumber\\
 \frac{\dd^2 W}{\dd t^2}
 &=\sum_{i,j}\frac{\partial s_i}{\partial e_j}\dot e_j\dot e_i
 +\sum_i s_i\ddot e_i
 =\dot\ee^{\,T}\bm D\dot\ee+\Sg:\ddot\E.
\end{align}
Substituting $\ddot\E(0)$ gives
\begin{equation}
 \left.\frac{\dd^2 W(\F+t\bm H)}{\dd t^2}\right|_{t=0}
 =\dot{\ee}(0)^T\bm D\dot{\ee}(0)+\Sg:(\bm H^T\bm H).
 \label{eq:strain_vs_F_curvature}
\end{equation}
Here $\bm D=\partial\ssv/\partial\ee$ is the constitutive tangent and $\Sg$ is the second Piola--Kirchhoff stress, both evaluated at $t=0$. The first term accounts for the change of stress with strain, and the second for the curvature of the strain path, weighted by the current stress. Even if $\bm D$ is positive semidefinite, the stress term can be negative in compression.

A simple strain-convex energy shows that the stress term can make this curvature negative. Consider $W(\E)=\mu\E:\E$, where $\mu>0$ is a constant with units of energy density. One has $\Sg=2\mu\E$. Since
\begin{equation}
 W(\ee)=\mu(e_1^2+e_2^2+e_3^2/2),
 \qquad \bm D=\diag(2\mu,2\mu,\mu),
\end{equation}
the energy is strictly convex in engineering strain. At $\F=\lambda\I$, where $\lambda>0$ is the common in-plane stretch, choose the rank-one perturbation $\bm H=\diag(1,0)$, which keeps $\det(\lambda\I+t\bm H)=\lambda(\lambda+t)$ positive for $t>-\lambda$. Then
\begin{equation}
 \dot\ee=(\lambda,0,0)^T,\qquad
 \Sg=\mu(\lambda^2-1)\I,\qquad
 \bm H^T\bm H=\diag(1,0).
\end{equation}
Hence the tangent term is $2\mu\lambda^2$, and the stress term is $\mu(\lambda^2-1)$, giving
\begin{equation}
 \left.\frac{\dd^2 W(\lambda\I+t\bm H)}{\dd t^2}\right|_{t=0}
 =2\mu\lambda^2+\mu(\lambda^2-1)
 =\mu(3\lambda^2-1).
 \label{eq:strain_convex_counterexample}
\end{equation}
At $\lambda=1/2$, the curvature is $-\mu/4$, violating Eq.~\eqref{eq:lh_background} despite convexity in $\E$. The curvature is negative only for $\lambda<1/\sqrt3$, so the counterexample requires sufficiently strong compression.

\section{Proofs for the paired features and complete energy}
\label{app:convex}
Propositions~\ref{prop:paired_features} and~\ref{prop:complete_energy} both rest on the convexity of a single stretch--determinant contribution, established first.

\subsection{Joint convexity of the directional building block}
Every contribution in Eq.~\eqref{eq:features} has the form
\begin{equation}
 f(\bm x,J)=\|\bm x\|^aJ^{-b},
 \qquad \bm x=\F\bm d,
\end{equation}
up to a positive constant, with $a=2p$ or $2q$. It suffices to determine when this building block is jointly convex in $(\bm x,J)$. Let $t=\|\bm x\|>0$. Because $f$ depends on $\bm x$ only through $t$, its Hessian separates into a block coupling the radial variable $t$ to $J$ and transverse directions orthogonal to $\bm x$. The radial--determinant block is
\begin{equation}
 \begin{pmatrix}
 a(a-1)t^{a-2}J^{-b} & -ab t^{a-1}J^{-b-1}\\
 -ab t^{a-1}J^{-b-1} & b(b+1)t^aJ^{-b-2}
 \end{pmatrix}.
\end{equation}
For this symmetric $2\times2$ block to be positive semidefinite, its diagonal entries and determinant must be nonnegative. The diagonal entries are nonnegative for $a\geq1$ and $b\geq0$, while the determinant is
\begin{align}
 \det &=ab\bigl[(a-1)(b+1)-ab\bigr]t^{2a-2}J^{-2b-2}\nonumber\\
 &=ab(a-b-1)t^{2a-2}J^{-2b-2}\geq0.
\end{align}
The factor $a-b-1$ gives the upper bound $b\leq a-1$. If it is violated, the determinant is negative and the block has curvatures of opposite sign. The remaining, transverse Hessian eigenvalues are $a t^{a-2}J^{-b}\geq0$. Hence $a\geq1$ and $0\leq b\leq a-1$ make the complete Hessian positive semidefinite. Continuous extension gives convexity at $\bm x=\bm0$. That point is never reached on $\GL$, because $\F\bm d\neq\bm0$ for a nonzero direction, so the lack of smoothness of $\|\bm x\|^a$ there does not affect the stress or tangent. Substituting $a=2p$ gives $p\geq\tfrac12$ and $0\leq b\leq2p-1$, and the same argument with $(q,c)$ establishes all bounds in Eq.~\eqref{eq:featureconstraints}.

\subsection{Proof of Proposition~\ref{prop:paired_features}}
\begin{proof}
\emph{Convexity and objectivity.}
For any fixed learned parameter set, the directions $\bm d_i$ and $\bm d_i^\perp$ are independent of deformation, so the maps $\F\mapsto\F\bm d_i$ and $\F\mapsto\F\bm d_i^\perp$ are linear. Composing the convex building block above with these maps, multiplying by the positive factors $1/p_i$ and $1/q_i$, and adding the two contributions proves convexity of $z_i$ in independent $(\F,J)$. A left proper rotation changes $\F\bm d$ to $\bm Q\F\bm d$ without changing its norm, so each contribution is objective.

\emph{Reference value and derivative.}
The physical relation $J=\det\F$ is now restored, with every learned parameter still fixed. At $\E=\bm0$, both directional invariants and $J$ equal one, so the two contributions give $z_i^0=1/p_i+1/q_i$. Moreover, $I_4(\bm d_i)=\bm d_i^T(\I+2\E)\bm d_i$ and $J^2=(1+2E_{11})(1+2E_{22})-4E_{12}^2$ imply $\partial_{\E}I_4(\bm d_i)=2\bm d_i\otimes\bm d_i$ and $\partial_{\E}J|_0=\I$. Applying the product rule to the first contribution gives
\begin{equation}
 \left.\partial_{\E}\frac{I_4(\bm d_i)^{p_i}}{p_iJ^{b_i}}\right|_0
 =\frac1{p_i}\bigl[2p_i\bm d_i\otimes\bm d_i-b_i\I\bigr]
 =2\bm d_i\otimes\bm d_i-\frac{b_i}{p_i}\I.
\end{equation}
The perpendicular contribution has the same form with $(\bm d_i,p_i,b_i)$ replaced by $(\bm d_i^\perp,q_i,c_i)$. Since $\bm d_i\otimes\bm d_i+\bm d_i^\perp\otimes\bm d_i^\perp=\I$, adding the two contributions removes the directional dependence at the reference and gives
\begin{equation}
 \left.\partial_{\E}z_i\right|_0
 =2(\bm d_i\otimes\bm d_i+\bm d_i^\perp\otimes\bm d_i^\perp)
 -\left(\frac{b_i}{p_i}+\frac{c_i}{q_i}\right)\I
 =\rho_i\I.
\end{equation}
In the work-conjugate engineering-strain convention of Eq.~\eqref{eq:voigt}, this tensor gradient corresponds to $\partial_{\ee}z_i|_0=\rho_i(1,1,0)^T$, proving Eq.~\eqref{eq:referencefeatures}.
\end{proof}

\subsection{Proof of Proposition~\ref{prop:complete_energy}}
\label{app:complete_energy_proof}

\begin{proof}
\emph{Potential structure, objectivity, and stress symmetry (C1)--(C3).}
Every directional norm, $J$, and $\tr\C$ are unchanged by left proper rotations, since all of them can be expressed through $\bm d_i^T\C\bm d_i$, $\tr\C$, and $\sqrt{\det\C}$. Hence the complete energy is objective, and differentiation with respect to symmetric strain gives a symmetric material stress. Differentiating this scalar energy supplies (C1) and the major symmetry of its constitutive tangent. Smoothness holds on $\GL$ because $J>0$, nonzero material directions have nonzero deformed length, and both cores are $C^2$. Because this argument uses left rotations only, it does not impose a material symmetry group.

\emph{Reference normalization (C4).}
Equation~\eqref{eq:energy_reference_value} shows that $W_{\rm learn}$ sets the reference energy, and Eq.~\eqref{eq:energy_reference_stress} shows that $W_{\rm corr}$ cancels the learned reference stress. The calculations for $W_{\rm vol}$ and $W_{\rm dist}$ in Section~\ref{sec:energy_admissibility} show that both growth terms have zero value and zero strain gradient at the reference, so the complete energy satisfies both reference conditions.

\emph{Polyconvexity (C5).} Let $\bm G\in\R^{2\times2}$ and $j>0$ be independent. Define $\mathcal P(\bm G,j)$ by the right-hand side of Eq.~\eqref{eq:energy}, replacing $\bm z$ by $\bm z(\bm G,j)$, $J$ by $j$, and $\tr\C$ by $\|\bm G\|_F^2$. For two such pairs $\bm X_1=(\bm G_1,j_1)$, $\bm X_2=(\bm G_2,j_2)$ and $0\leq t\leq1$,
\begin{align}
 H\!\left(\bm z(t\bm X_1+(1-t)\bm X_2)\right)
 &\leq H\!\left(t\bm z(\bm X_1)+(1-t)\bm z(\bm X_2)\right)\nonumber\\
 &\leq tH(\bm z(\bm X_1))+(1-t)H(\bm z(\bm X_2)).
 \label{eq:convex_composition_proof}
\end{align}
The first inequality uses feature convexity and core monotonicity, and the second uses core convexity, so $W_{\rm learn}$ is convex up to its constant reference subtraction. The correction $W_{\rm corr}=-r(j-1)$ is affine for either sign of $r$, while $W_{\rm vol}$ and $W_{\rm dist}$ are convex in independent $(\bm G,j)$ because their nonaffine parts are positive multiples of $-\log j$, $(j-1)^2$, and $\|\bm G\|_F^2$. Consequently, $\mathcal P$ is convex and $W(\F)=\mathcal P(\F,\det\F)$, establishing Eq.~\eqref{eq:polyconvex_definition}.

\emph{Growth (C6).}
Convexity of $H$ gives its supporting-plane inequality at $\bm z^0$,
\begin{equation}
 H(\bm z)\geq H(\bm z^0)+\sum_i h_i(z_i-z_i^0)
 \geq H(\bm z^0)-\sum_i h_i z_i^0.
 \label{eq:supporting_plane}
\end{equation}
The second inequality uses $h_i,z_i\geq0$. Hence $W_{\rm learn}$ has a deformation-independent lower bound and cannot cancel the analytic growth terms. As $J\to0^+$, the logarithmic part of $W_{\rm vol}+W_{\rm dist}$ diverges while $W_{\rm corr}$ remains bounded. As $J\to\infty$, the positive quadratic part of $W_{\rm vol}$ dominates every potentially negative linear or logarithmic term. Finally, if $J$ remains in a compact subset of $(0,\infty)$ while $\|\F\|_F\to\infty$, the term $\varepsilon\|\F\|_F^2/2$ in $W_{\rm dist}$ diverges and all determinant-only terms remain bounded. These are the three stated growth limits.
\end{proof}

\section{Sufficient lower bound for nonnegative energy}
\label{app:nonnegative}
Define $k_i=1/p_i+1/q_i$ and $\eta_i=\rho_i/k_i$. Write $t=\|\F\bm d_i\|^2$, $u=\|\F\bm d_i^\perp\|^2$. Hadamard's determinant inequality gives $tu\geq J^2$. Since each feature is increasing in $t,u$, its minimum at fixed $J$ occurs on $tu=J^2$. Minimizing
\begin{equation}
 \frac{t^{p_i}}{p_iJ^{b_i}}+\frac{J^{2q_i-c_i}t^{-q_i}}{q_i}
\end{equation}
gives $t^{p_i+q_i}=J^{2q_i+b_i-c_i}$ and hence $z_i\geq k_iJ^{\eta_i}$. The bound holds separately for each feature, without assuming that the minima for different directions are attained simultaneously, and $h_i\geq0$ allows it to be inserted term by term in Eq.~\eqref{eq:supporting_plane}.

Inserting the feature bound and using $\tr\C/2\geq J$ in $W_{\rm dist}$ gives
\begin{equation}
 W\geq g(J):=\sum_i h_i k_i[J^{\eta_i}-1-\eta_i(J-1)]
 +(\alpha+\varepsilon)(J-1-\log J)+\frac\beta2(J-1)^2.
 \label{eq:lowerbound}
\end{equation}
Because $g(1)=g'(1)=0$, it remains to establish $g''(J)\geq0$ for every $J>0$. Differentiation gives
\begin{equation}
 J^2g''(J)=\alpha+\varepsilon+\beta J^2+\sum_i c_i^*J^{\eta_i},
 \qquad c_i^*=h_i k_i\eta_i(\eta_i-1).
\end{equation}
Only terms with $0<\eta_i<1$ can have negative coefficients. If $c_-=\sum_{c_i^*<0}|c_i^*|$, the two inequalities
\begin{align}
 \alpha+\varepsilon+\sum_{\eta_i<0}c_i^*-c_-&\geq0,\\
 \beta+\sum_{\eta_i\geq1}c_i^*-c_-&\geq0
\end{align}
are sufficient on $0<J\leq1$ and $J\geq1$, respectively. For the first region, $J^{\eta_i}\leq1$ for negative-coefficient terms and $J^{\eta_i}\geq1$ when $\eta_i<0$. For the second region, dividing $J^2g''$ by $J$ makes every negative-coefficient power satisfy $J^{\eta_i-1}\leq1$, while $J^{\eta_i-1}\geq1$ for $\eta_i\geq1$ and $\beta J\geq\beta$. Dropping the remaining nonnegative contributions gives the stated bounds. Their failure is inconclusive, so the trained parameters are checked with a sharper interval bound.

The audit partitions $[10^{-4},10^4]$ into 4000 logarithmically spaced intervals. On each interval $[J_L,J_R]$, a positive power term is bounded below by the smaller endpoint value, whereas a negative-coefficient term is bounded below by $c_i^*J_R^{\eta_i}$, since its exponent lies in $(0,1)$. The constant and quadratic terms are bounded below by $\alpha+\varepsilon$ and $\beta J_L^2$. Below the finite interval, the nonnegative negative-exponent terms increase and the magnitudes of the negative-coefficient terms decrease toward $J=0$. Above it, factoring out $J^{\eta_-}$ with $\eta_-=\max_{c_i^*<0}\eta_i<1$ yields a lower bound whose retained positive powers increase. These endpoint and tail inequalities cover every $J>0$, not only the grid points. Positive margins establish the sufficient bound for the trained parameters, up to floating-point accuracy, which random energy samples could not do.

\section{Latent-space closure details}
\label{app:latent_closure}
Let $\bm V_{\rm tra}$ be the retained orthonormal POD basis of Eq.~\eqref{eq:linear_prom}, with $n_{\rm tra}=n_{\rm tot}$ columns, let $\bm S_{\rm snap}$ collect the centered displacement snapshots, and collect the corresponding engineering-strain inputs in $\bm M=[\,\ee^1,\ldots,\ee^{N_s}\,]$. Their POD coefficients are
\begin{equation}
 \bm Q_{\rm POD}=\bm V_{\rm tra}^T\bm S_{\rm snap}.
\end{equation}
The input-informed construction first identifies the least-squares map from these coefficients to strain,
\begin{equation}
 \bm T_m=\bm M\bm Q_{\rm POD}^{\dagger}
 \in\arg\min_{\bm H}\|\bm H\bm Q_{\rm POD}-\bm M\|_F^2,
\end{equation}
where $\dagger$ denotes the Moore--Penrose pseudoinverse. Assuming rank three, the compact factorization
\begin{equation}
 \bm V_{\rm tra}\bm T_m^T=\bm V\bm\Sigma_m\bm Z_m^T,
 \qquad
 \bm A_m=\bm\Sigma_m^{-1}\bm Z_m^T
\end{equation}
defines the three-column orthonormal primary block $\bm V$ and the coordinate transformation $\bm A_m$, with $\bm\xi^s=\bm T_m\bm Q_{\rm POD}^{\,s}$ and primary modal amplitudes $\bm V^T(\bm u^s-\bm u_{\rm ref})=\bm A_m\bm\xi^s$. The secondary block $\overline{\bm V}$ is an orthonormal complement within the retained POD span, ordered by snapshot energy, so that $\bm V_{\rm tot}=[\,\bm V,\overline{\bm V}\,]$ spans the same subspace as $\bm V_{\rm tra}$ and differs from it only by an orthogonal change of basis. The two blocks give the deployed decoder of Eq.~\eqref{eq:deployed_decoder}.

The displacement error separates the two approximation mechanisms. For the snapshot coordinates $(\bm\xi^s,\overline{\bm q}^{\,s})$,
\begin{equation}
 \|\bm u^s-\widetilde{\bm u}(\bm\xi^s)\|_2^2
 =\|\overline{\bm q}^{\,s}-\widehat{\mathcal N}(\bm\xi^s)\|_2^2
 +\|(\bm I-\bm V_{\rm tot}\bm V_{\rm tot}^T)
       (\bm u^s-\bm u_{\rm ref})\|_2^2.
\end{equation}
The first term is the closure error within the retained span, and the second is the POD truncation error, which no increase in network capacity can recover.

Finally, let $\bm B_{\xi}=\partial_{\bm\xi}\widetilde{\bm u}$, $\bm R$ be the FOM residual, and $\bm K=\partial_{\bm u}\bm R$. Differentiating the projected residual $\bm r=\bm B_{\xi}^T\bm R$ gives
\begin{equation}
 \frac{\partial\bm r}{\partial\bm\xi}
 =\bm B_{\xi}^T\bm K\bm B_{\xi}
 +\sum_{j=1}^{N}R_j
   \frac{\partial^2\widetilde u_j}
        {\partial\bm\xi\,\partial\bm\xi}.
\end{equation}
The second term is the curvature contribution of the decoder manifold. It need not vanish at reduced equilibrium because $\bm B_{\xi}^T\bm R=\bm0$ does not imply $\bm R=\bm0$.

\section{Elementwise cubature conventions}
\label{app:cubature}
The general ECM construction and MAW--ECM pruning algorithm follow Hernández et al.~\cite{hernandez2017,hernandez2026dimensional}, and the expressions below specify how the deployed residual and stress rules are assembled.

For the fixed rule, let $\bm c_e^s\in\R^d$ contain the contribution of element $e$ at training state $s=1,\ldots,N_c$. Residual cubature uses $\bm c_e^s=\bm h_e^s$. Stress cubature uses the four components of
\begin{equation}
 \bm p_e^s=\frac{1}{|Y|}\sum_{g\in e}\omega_g\bm P_{\mu,g}^s,
 \qquad
 \bm c_e^s=\operatorname{vec}\bm p_e^s
 =\begin{bmatrix}p_{e,11}^s&p_{e,12}^s&p_{e,21}^s&p_{e,22}^s\end{bmatrix}^{T}.
\end{equation}
One element occupies each row of the sampled integrand matrix,
\begin{equation}
 \bm M_{\rm int}=
 \begin{bmatrix}
  (\bm c_1^1)^T&\cdots&(\bm c_1^{N_c})^T\\
  \vdots&&\vdots\\
  (\bm c_{N_{\rm el}}^1)^T&\cdots&(\bm c_{N_{\rm el}}^{N_c})^T
 \end{bmatrix}
 \in\R^{N_{\rm el}\times(dN_c)}.
\end{equation}
A truncated SVD, with the tolerance $\varepsilon_{\rm SVD}$ of Section~\ref{sec:hyperreduction}, gives $\bm M_{\rm int}\simeq\bm U_{\rm cub}\bm\Sigma_{\rm cub}\bm V_{\rm cub}^T$. To enforce the multiplier sum without duplicating a constant component already represented by the retained integration modes, the fitting system uses
\begin{equation}
 \bm p_\perp=(\bm I-\bm U_{\rm cub}\bm U_{\rm cub}^T)\bm1,
 \qquad
 \bm G_{\rm cub}=
 \begin{bmatrix}
  \bm U_{\rm cub}^T\\
  \bm p_\perp^T/\|\bm p_\perp\|_2
 \end{bmatrix},
 \qquad
 \bm b_{\rm cub}=\bm G_{\rm cub}\bm1.
\end{equation}
If the constant already lies in the retained span, the last row is redundant. Otherwise its exact reproduction gives $\sum_e w_e=N_{\rm el}$. The reference quadrature measures are already contained in $\bm c_e^s$, so the unknowns are element multipliers.

After MAW--ECM pruning selects a support $\mathcal Z=\{e_1,\ldots,e_{|\mathcal Z|}\}$, neural fitting returns to the uncompressed contributions at each decoded state and targets the full-mesh sum,
\begin{equation}
 \bm A^s=
 \begin{bmatrix}
  \bm c_{e_1}^s&\cdots&\bm c_{e_{|\mathcal Z|}}^s\\
  1&\cdots&1
 \end{bmatrix},
 \qquad
 \bm b^s=
 \begin{bmatrix}
  \displaystyle\sum_{e\in\mathcal E}\bm c_e^s\\
  N_{\rm el}
 \end{bmatrix}.
 \label{eq:statewise_cubature_appendix}
\end{equation}
Residual fitting uses $d=n$ and $\mathcal Z_r$, and stress fitting uses $d=4$ and $\mathcal Z_s$. With the scaled-softmax map of Eq.~\eqref{eq:softmax_weights}, the residual network minimizes
\begin{equation}
 \mathcal L_{\rm cub}(\bm\theta_r)=
 \frac{1}{|\mathcal I_{\rm fit}|}
 \sum_{s\in\mathcal I_{\rm fit}}
 \frac{\|\bm A^s\bm w(\bm\xi^s)-\bm b^s\|_2^2}
      {\max(\|\bm b^s\|_2^2,\delta_{\rm floor})},
\end{equation}
and the stress network uses the analogous loss with $\bm v$. A preliminary fit to the sampled MAW--ECM weights initializes each network, but the final objective is integration accuracy, because a reproducing weight vector is not unique.

For the tangent, the two iterative implementations use Eq.~\eqref{eq:reduced_ift} but evaluate its partial derivatives differently. For the affine HPROM, the assembled reduced stiffness and the analytic lifting derivative provide $\dd\qq_{\rm tra}^*/\dd\ee$ through three linear right-hand sides, and central differences are applied only to the stress along those linearized equilibrium directions. For HPROM--ANN, central differences at fixed $\qq$ or fixed $\ee$ evaluate all four partial-derivative blocks in Eq.~\eqref{eq:reduced_ift}, including decoder curvature and weight variation. The resulting small linear system does not require a new nonlinear microscopic solve for every perturbed strain. For D-HPROM--ANN, which solves no equilibrium, central differences of the complete direct stress map update the decoder, weights, lifting, and stress conversion at each perturbed input. Numerical differentiation is an implementation choice, and analytic or automatic differentiation would give the same formal tangent up to their respective numerical tolerances.

\section{Capacity and feature-count controls}
\label{app:capacity}
The comparisons of Section~\ref{sec:cell} pair compact constrained cores with a larger Unconstrained network.  To separate the effect of the constraints from that of capacity, both benchmarks complete a two-by-two design.  The constrained cores are widened to about 25\,000 parameters, comparable with the Unconstrained energy, and an Unconstrained energy with 981 parameters matches the compact ICNN.  Every new fit follows the protocol of Section~\ref{sec:training_protocol} with the same three seeds.

Table~\ref{tab:capacity_mc} reports the MC-RVE cells as medians over three seeds, as in Table~\ref{tab:mc_rve_m06}.  Widening changes the ICNN error from $0.0445\%$ to $0.0398\%$, while the Unconstrained energy improves from $0.0126\%$ to $0.0073\%$.  At matched capacity the constraints cost a factor of 3.2 to 5.5 in stress error.

Table~\ref{tab:capacity_sc} reports the SC-RVE cells for the validation-selected seed, as in Table~\ref{tab:constitutive}, together with their displacement errors in the coupon of Section~\ref{sec:coupon}.  Widening changes the constrained test errors by less than $3\%$, and the Unconstrained energy with 981 parameters remains about 75 to 80 times more accurate than either wide core.  In the coupon, the corresponding factor between matched cells lies between about 115 and 200.  Thirty-two learned features, trained with the same protocol, lower the test errors to $0.537\%$ and $0.618\%$ and the coupon errors to $0.343\%$ and $0.347\%$.  The sufficient bound of Appendix~\ref{app:nonnegative}, however, fails for both 32-feature parameter sets, so the six-feature models remain the deployed choice.

\begin{table}[htbp]
 \centering\small
 \caption{MC-RVE capacity controls: independent-test errors, in percent.}
 \label{tab:capacity_mc}
 \begin{tabular}{lrrrr}
\toprule
Model & Parameters & Stress [\%] & Energy [\%] & Tangent [\%] \\
\midrule
ICNN, $m=6$ & 975 & 0.0445 & 0.0079 & 0.3851 \\
ICNN, wide & 26\,919 & 0.0398 & 0.0069 & 0.2962 \\
ICKAN, $m=6$ & 1\,659 & 0.0406 & 0.0066 & 0.3622 \\
ICKAN, wide$^\ast$ & 24\,847 & 0.0744 & 0.0104 & 1.8311 \\
Unconstrained, small & 981 & 0.0126 & 0.0019 & 0.1579 \\
Unconstrained & 25\,473 & 0.0073 & 0.0011 & 0.0857 \\
\bottomrule
\end{tabular}

\par\smallskip{\footnotesize Relative error norm over the 512 test states, median over three seeds. $^\ast$Two of the three seeds stopped early at the validation plateau.}

\end{table}

\begin{table}[htbp]
 \centering\small
 \caption{SC-RVE capacity and feature-count controls, in percent.}
 \label{tab:capacity_sc}
 \begin{tabular}{lrrr}
\toprule
Model & Parameters & $e_S^{\rm test}$ [\%] & $e_u$ [\%] \\
\midrule
ICNN, $m=6$ & 975 & 0.6214 & 0.4066 \\
ICNN, wide & 26\,919 & 0.6066 & 0.4162 \\
ICNN, $m=32$ & 2\,379 & 0.5370 & 0.3433 \\
ICKAN, $m=6$ & 1\,659 & 0.6359 & 0.4197 \\
ICKAN, wide & 24\,847 & 0.6540 & 0.4322 \\
ICKAN, $m=32$ & 3\,999 & 0.6177 & 0.3469 \\
Unconstrained, small & 981 & 0.0081 & 0.0035 \\
Unconstrained & 25\,473 & 0.0052 & 0.0022 \\
\bottomrule
\end{tabular}

\par\smallskip{\footnotesize Validation-selected seed of each cell. Test: 400 states; $e_u$: final nodal displacement error of the coupon of Section~\ref{sec:coupon}.}

\end{table}


\begin{thebibliography}{99}

\bibitem{Miehe1999}
C. Miehe, J. Schr\"oder, J. Schotte, Computational homogenization
analysis in finite plasticity simulation of texture development in
polycrystalline materials, Computer Methods in Applied Mechanics
and Engineering 171(3--4) (1999) 387--418.
\doi{10.1016/S0045-7825(98)00218-7}

\bibitem{kouznetsova2001approach}
V.~G. Kouznetsova, W.~A.~M. Brekelmans, F.~P.~T. Baaijens,
An approach to micro--macro modeling of heterogeneous materials,
Computational Mechanics 27(1) (2001) 37--48.
\doi{10.1007/s004660000212}

\bibitem{geers2010multi}
M.~G.~D. Geers, V.~G. Kouznetsova, W.~A.~M. Brekelmans,
Multi-scale computational homogenization: Trends and challenges,
Journal of Computational and Applied Mathematics 234(7) (2010) 2175--2182.
\doi{10.1016/j.cam.2009.08.077}

\bibitem{feyel2003multilevel}
F. Feyel, A multilevel finite element method (FE$^2$) to describe the
response of highly non-linear structures using generalized continua,
Computer Methods in Applied Mechanics and Engineering 192(28--30) (2003)
3233--3244.
\doi{10.1016/S0045-7825(03)00348-7}

\bibitem{le2015computational}
B.~A. Le, J. Yvonnet, Q.-C. He, Computational homogenization of nonlinear
elastic materials using neural networks, International Journal for Numerical
Methods in Engineering 104(12) (2015) 1061--1084.
\doi{10.1002/nme.4953}

\bibitem{benner2015survey}
P. Benner, S. Gugercin, K.~E. Willcox, A survey of projection-based model
reduction methods for parametric dynamical systems, SIAM Review 57(4) (2015)
483--531.
\doi{10.1137/130932715}

\bibitem{yvonnet2007}
J. Yvonnet, Q.-C. He, The reduced model multiscale method (R3M) for the non-linear homogenization of hyperelastic media at finite strains, Journal of Computational Physics 223 (2007) 341--368. \doi{10.1016/j.jcp.2006.09.019}

\bibitem{Hernandez2014}
J.~A. Hern\'andez, J. Oliver, A.~E. Huespe, M.~A. Caicedo, J.~C.
Cante, High-performance model reduction techniques in computational
multiscale homogenization, Computer Methods in Applied Mechanics
and Engineering 276 (2014) 149--189.

\bibitem{hernandez2017}
J.A. Hern\'andez, M.A. Caicedo, A. Ferrer, Dimensional hyper-reduction of nonlinear finite element models via empirical cubature, Computer Methods in Applied Mechanics and Engineering 313 (2017) 687--722. \doi{10.1016/j.cma.2016.10.022}

\bibitem{Aldakheel2023}
F. Aldakheel, E.~S. Elsayed, T.~I. Zohdi, P. Wriggers, Efficient
multiscale modeling of heterogeneous materials using deep neural
networks, Computational Mechanics 72(1) (2023) 155--171.

\bibitem{kalina2023}
K.A. Kalina, L. Linden, J. Brummund, M. K\"astner, FEANN: an efficient data-driven multiscale approach based on physics-constrained neural networks and automated data mining, Computational Mechanics 71 (2023) 827--851. \doi{10.1007/s00466-022-02260-0}

\bibitem{Liu2022}
B. Liu, N. Kovachki, Z. Li, K. Azizzadenesheli, A. Anandkumar, A.~M.
Stuart, K. Bhattacharya, A learning-based multiscale method and its
application to inelastic impact problems, Journal of the
Mechanics and Physics of Solids 158 (2022) 104668.

\bibitem{Ghaderi2020}
A. Ghaderi, V. Morovati, R. Dargazany, A physics-informed assembly of
feed-forward neural network engines to predict inelasticity in
cross-linked polymers, Polymers 12(11) (2020) 2628.
\doi{10.3390/polym12112628}

\bibitem{AsadFarhat2026}
F. As'ad, C. Farhat, A staggered training framework for
mechanics-informed neural networks in tractable multiscale
homogenization with application to woven fabrics, Computer
Methods in Applied Mechanics and Engineering 452 (2026) 118666.

\bibitem{asad2022}
F. As'ad, P. Avery, C. Farhat, A mechanics-informed artificial neural network approach in data-driven constitutive modeling, International Journal for Numerical Methods in Engineering 123 (2022) 2738--2759. \doi{10.1002/nme.6957}

\bibitem{klein2022}
D.K. Klein, M. Fern\'andez, R.J. Martin, P. Neff, O. Weeger, Polyconvex anisotropic hyperelasticity with neural networks, Journal of the Mechanics and Physics of Solids 159 (2022) 104703. \doi{10.1016/j.jmps.2021.104703}

\bibitem{linden2023}
L. Linden, D.K. Klein, K.A. Kalina, J. Brummund, O. Weeger, M. K\"astner,
Neural networks meet hyperelasticity: A guide to enforcing physics,
Journal of the Mechanics and Physics of Solids 179 (2023) 105363.
\doi{10.1016/j.jmps.2023.105363}

\bibitem{GaoNeffRoventaThiel2017}
D.~Y. Gao, P. Neff, I. Roventa, C. Thiel, On the convexity of
nonlinear elastic energies in the right Cauchy--Green tensor,
Journal of Elasticity 127 (2017) 303--308.
\doi{10.1007/s10659-016-9601-6}

\bibitem{Xu2021}
K. Xu, D.~Z. Huang, E. Darve, Learning constitutive relations using
symmetric positive definite neural networks,
Journal of Computational Physics 428 (2021) 110072.
\doi{10.1016/j.jcp.2020.110072}

\bibitem{ball1977}
J.M. Ball, Convexity conditions and existence theorems in nonlinear elasticity, Archive for Rational Mechanics and Analysis 63 (1977) 337--403.

\bibitem{Ciarlet1988}
P.~G. Ciarlet, Mathematical Elasticity, Volume I: Three-Dimensional
Elasticity, North-Holland, 1988.

\bibitem{SchroderNeff2003}
J. Schr\"oder, P. Neff, Invariant formulation of hyperelastic
transverse isotropy based on polyconvex free energy functions,
International Journal of Solids and Structures 40(2) (2003) 401--445.

\bibitem{Gasser2006}
T.~C. Gasser, R.~W. Ogden, G.~A. Holzapfel, Hyperelastic modelling of
arterial layers with distributed collagen fibre orientations,
Journal of the Royal Society Interface 3(6) (2006) 15--35.
\doi{10.1098/rsif.2005.0073}

\bibitem{amos2017}
B. Amos, L. Xu, J.Z. Kolter, Input convex neural networks, Proceedings of Machine Learning Research 70 (2017) 146--155. \url{https://proceedings.mlr.press/v70/amos17b.html}

\bibitem{Kalina2024}
K.A. Kalina, J. Brummund, W. Sun, M. K\"astner,
Neural networks meet anisotropic hyperelasticity: A framework based on generalized structure tensors and isotropic tensor functions,
Computer Methods in Applied Mechanics and Engineering 437 (2025) 117725.
\doi{10.1016/j.cma.2024.117725}

\bibitem{klein2026}
D.K. Klein, K.A. Kalina, R. Ortigosa, J. Mart\'inez-Frutos, M. K\"astner, O. Weeger, Advances in polyconvex anisotropic hyperelasticity, arXiv:2605.27011v1 (2026). \url{https://arxiv.org/abs/2605.27011}

\bibitem{Linka2023}
K. Linka, E. Kuhl, A new family of Constitutive Artificial Neural
Networks towards automated model discovery,
Computer Methods in Applied Mechanics and Engineering 403 (2023) 115731.

\bibitem{Tac2022}
V. Tac, F. Sahli Costabal, A. Buganza Tepole, Data-driven tissue
mechanics with polyconvex neural ordinary differential equations,
Computer Methods in Applied Mechanics and Engineering 398 (2022) 115248.
\doi{10.1016/j.cma.2022.115248}

\bibitem{Abdolazizi2025}
K.P. Abdolazizi, R.C. Aydin, C.J. Cyron, K. Linka,
Constitutive Kolmogorov--Arnold networks (CKANs): Combining accuracy and interpretability in data-driven material modeling,
Journal of the Mechanics and Physics of Solids 203 (2025) 106212.
\doi{10.1016/j.jmps.2025.106212}

\bibitem{thakolkaran2025}
P. Thakolkaran, Y. Guo, S. Saini, M. Peirlinck, B. Alheit, S. Kumar, Can KAN CANs? Input-convex Kolmogorov-Arnold Networks (KANs) as hyperelastic constitutive artificial neural networks (CANs), Computer Methods in Applied Mechanics and Engineering 443 (2025) 118089. \doi{10.1016/j.cma.2025.118089}

\bibitem{Hernandez2020}
J.~A. Hern\'andez, A multiscale method for periodic structures using
domain decomposition and ECM-hyperreduction, Computer Methods in
Applied Mechanics and Engineering 368 (2020) 113192.

\bibitem{Amsallem2012}
D. Amsallem, M.~J. Zahr, C. Farhat, Nonlinear model order
reduction based on local reduced-order bases, International
Journal for Numerical Methods in Engineering 92(10) (2012) 891--916.

\bibitem{Lee2020}
K. Lee, K.~T. Carlberg, Model reduction of dynamical systems on
nonlinear manifolds using deep convolutional autoencoders,
Journal of Computational Physics 404 (2020) 108973.

\bibitem{Barnett2022}
J. Barnett, C. Farhat, Quadratic approximation manifold for
mitigating the Kolmogorov barrier in nonlinear projection-based model
order reduction, Journal of Computational Physics 464 (2022) 111348.

\bibitem{barnett2023}
J. Barnett, C. Farhat, Y. Maday, Neural-network-augmented projection-based model order reduction for mitigating the Kolmogorov barrier to reducibility, Journal of Computational Physics 492 (2023) 112420.
\doi{10.1016/j.jcp.2023.112420}

\bibitem{sibuet2025discrete}
N. Sibuet, S. Ares de Parga, J.~R. Bravo, R. Rossi,
Discrete physics-informed training for projection-based reduced-order models with neural networks,
Axioms 14(5) (2025) 385.
\doi{10.3390/axioms14050385}

\bibitem{aresdeparga2026nonlinear}
S. Ares de Parga, R. Tezaur, C.G. Hern\'andez, C. Farhat,
Nonlinear projection-based model order reduction with machine learning regression for closure error modeling in the latent space,
Computer Methods in Applied Mechanics and Engineering 448 (2026) 118443.
\doi{10.1016/j.cma.2025.118443}

\bibitem{Chmiel2024}
M.~R. Chmiel, J.~L. Barnett, C. Farhat, Assessment of
projection-based model order reduction for a benchmark hypersonic flow
problem, AIAA SciTech 2024 Forum (2024), paper 2024-0250.
\doi{10.2514/6.2024-0250}

\bibitem{Fritzen2016}
F. Fritzen, M. Hodapp, The finite element square reduced (FE$^{2R}$) method with GPU acceleration: towards three-dimensional two-scale simulations,
International Journal for Numerical Methods in Engineering 107 (2016) 853--881.
\doi{10.1002/nme.5188}

\bibitem{Grimberg2021}
S. Grimberg, C. Farhat, R. Tezaur, C. Bou-Mosleh, Mesh sampling and
weighting for the hyperreduction of nonlinear Petrov--Galerkin
reduced-order models with local reduced-order bases, International
Journal for Numerical Methods in Engineering 122(7) (2021) 1846--1874.

\bibitem{AresDeParga2023}
S. Ares de Parga, J.~R. Bravo, J.~A. Hern\'andez, R. Zorrilla,
R. Rossi, Hyper-reduction for Petrov--Galerkin reduced order models,
Computer Methods in Applied Mechanics and Engineering
416 (2023) 116298.

\bibitem{Chaturantabut2010}
S. Chaturantabut and D. C. Sorensen.
Nonlinear model reduction via discrete empirical interpolation.
\emph{SIAM Journal on Scientific Computing} 32(5) (2010) 2737--2764.
\doi{10.1137/090766498}

\bibitem{patera2017lp}
A.~T. Patera, M. Yano,
An LP empirical quadrature procedure for parametrized functions,
C. R. Acad. Sci. Paris, Ser. I 355 (2017) 1161--1167.
\doi{10.1016/j.crma.2017.10.020}

\bibitem{yano2019lp}
M. Yano, A.~T. Patera,
An LP empirical quadrature procedure for reduced basis treatment of parametrized nonlinear PDEs,
Computer Methods in Applied Mechanics and Engineering 344 (2019) 1104--1123.
\doi{10.1016/j.cma.2018.02.028}

\bibitem{Farhat2014}
C. Farhat, P. Avery, T. Chapman, J. Cortial, Dimensional reduction
of nonlinear finite element dynamic models with finite rotations and
energy-based mesh sampling and weighting for computational efficiency,
International Journal for Numerical Methods in Engineering
98(9) (2014) 625--662.

\bibitem{farhat2015}
C. Farhat, T. Chapman, P. Avery, Structure-preserving, stability, and accuracy properties of the energy-conserving sampling and weighting method for the hyper reduction of nonlinear finite element dynamic models, International Journal for Numerical Methods in Engineering 102 (2015) 1077--1110. \doi{10.1002/nme.4820}

\bibitem{Hernandez2024}
J.~A. Hern\'andez, J.~R. Bravo, S. Ares de Parga, CECM: A
continuous empirical cubature method with application to the
dimensional hyperreduction of parameterized finite element models,
Computer Methods in Applied Mechanics and Engineering
418 (2024) 116552.

\bibitem{Bravo2024}
J.~R. Bravo, J.~A. Hern\'andez, S. Ares de Parga, R. Rossi, A
subspace-adaptive weights cubature method with application to the
local hyperreduction of parameterized finite element models,
International Journal for Numerical Methods in Engineering
125(24) (2024) e7590.

\bibitem{hernandez2026dimensional}
J.~A. Hern\'andez, S. Ares de Parga, R. Rossi,
Dimensional hyperreduction of nonlinear finite element models via empirical cubature with manifold-adaptive weights,
arXiv:2609.03068v1 (2026).
\url{https://arxiv.org/abs/2609.03068}

\bibitem{Tac2024}
V. Tac, K. Linka, F. Sahli-Costabal, E. Kuhl, A. Buganza Tepole,
Benchmarking physics-informed frameworks for data-driven hyperelasticity,
Computational Mechanics 73 (2024) 49--65.
\doi{10.1007/s00466-023-02355-2}

\end{thebibliography}
\end{document}